\RequirePackage[l2tabu, orthodox]{nag} 
\documentclass{article}
\usepackage{sectsty}

\usepackage[top=0.5in, left=0.5in, right=0.5in, bottom=0.75in]{geometry}
\usepackage{amsmath}
\usepackage[round]{natbib}
\usepackage{amssymb,color}
\usepackage{appendix}
\usepackage{amsthm}
\usepackage{verbatim}
\usepackage{enumerate}

\allowdisplaybreaks[1]

\numberwithin{equation}{section}

\usepackage{titling}
\usepackage{hyperref}
\usepackage{mathtools} 
\mathtoolsset{showonlyrefs=true}  

\newtheorem{thm}{Theorem}[section]

\newtheorem{assume}{A\!}
\newtheorem{Bassume}{B\!}
\newtheorem{prop}[thm]{Proposition}
\newtheorem{lemma}[thm]{Lemma}

\newtheorem{remark}[thm]{Remark}

\newenvironment{proof32}
{\begin{proof}[Proof of Theorem 3.2]}
	{\end{proof}}

\newenvironment{proof73}
{\begin{proof}[Proof of Theorem 7.3]}
	{\end{proof}}

\let\origthanks\thanks
\renewcommand\thanks[1]{\begingroup\let\rlap\relax\origthanks{#1}\endgroup}

\usepackage[normalem]{ulem}
\newcommand{\blue}[1]{{ \color{blue}{{#1}}}}  

\definecolor{green}{RGB}{40, 100, 0}

\usepackage{tikz}
\usetikzlibrary{cd}
\usetikzlibrary{decorations.pathmorphing}

\begin{document}

\author{Hyungbin Park\thanks{Corresponding author. Department of Mathematical Sciences and RIMS, Seoul National University, 1, Gwanak-ro, Gwanak-gu, Seoul, Republic of Korea. Email: hyungbin@snu.ac.kr,  hyungbin2015@gmail.com} \ and \ Stephan Sturm\thanks{Department of Mathematical Sciences, Worcester Polytechnic Institute,
	100 Institute Rd, Worcester, MA, USA. Email: ssturm@wpi.edu} 
}
\title{A sensitivity analysis of the long-term expected utility\\ of optimal portfolios}

\maketitle	
 
\abstract{ 
This paper discusses the sensitivity of the long-term expected utility of optimal portfolios for an investor with constant relative risk aversion. Under an incomplete market given by a factor model, we consider the utility maximization problem with long-time horizon. The main purpose is to find the long-term sensitivity, that is, the extent how much the optimal expected utility is affected in the long run for small changes of the underlying factor model.  The factor model induces a specific eigenpair of an operator, and this eigenpair does not only characterize the long-term behavior of the  optimal expected utility but also provides an explicit representation of the expected utility on a finite time horizon. We conclude that this eigenpair therefore determines the long-term sensitivity. As examples, explicit results for several market models such as the Kim--Omberg model for stochastic excess returns and the Heston stochastic volatility model are presented.}

\setlength{\parskip}{6pt}
 

\setlength{\parskip}{6pt}


\setlength{\parskip}{6pt}


\setlength{\parskip}{6pt}

\section{Introduction}
\label{sec:intro}

Finding an optimal investment strategy is an important topic in mathematical finance. There are several ways to formulate the optimal investment problem and one of the commonly accepted formulations is the use of utility function. An agent wants to maximize the expectation of the utility $U$ by trading assets  in a market. This paper also concerns this formulation of optimal expected utility, that is,
\begin{equation}
\label{eqn:max_U}
	\sup_{\Pi\in\mathcal{X}}\mathbb{E}^\mathbb{P}\bigl[U(\Pi_T)\bigr]
\end{equation}
for $\mathcal{X}$ the family of  wealth processes of admissible portfolios.

The analysis of this problem depends on the market completeness/incompleteness. The complete market case is relatively easy to find the optimal expected utility (see Section \ref{eqn:complete_markets}), whereas the incomplete market case is more complicated and requires advanced techniques. This paper deals with an incomplete market modeled by a factor model. Such factor models are widely used in the quantitative finance literature. In the following we provide first an overview of the topic of the paper, review the relevant literature and present the relative straightforward case of a complete market given by one-dimensional diffusion model.

\subsection{Overview}

The main purpose of this paper is to develop a sensitivity analysis of the long-term optimal expected utility. We consider two kinds of sensitivities. The first is the sensitivity with respect to the initial factor, e.g., the current spot volatility if the factor process is modeling the evolution of the volatility. For the initial value $\chi=X_0$ of the factor process, we study the behavior of
\[
	\frac{\partial}{\partial\chi}\sup_{\Pi\in\mathcal{X}}\mathbb{E}^\mathbb{P}\bigl[U(\Pi_T)\bigr]
\]
for large $T.$ The second is the sensitivity with respect to a change in the drift or volatility function, e.g., reversion speed, mean reversion level and volatility of volatility for a mean-reverting volatility process. Let $\epsilon$ be a perturbation parameter and consider a perturbed asset price  $S^\epsilon$ with $S=S^0.$ Denote by $\mathcal{X}^\epsilon$ the family of  wealth processes of admissible portfolios with the perturbed asset model $S^\epsilon.$ The precise meanings of $S^\epsilon$ and $\mathcal{X}^\epsilon$ are discussed in Section \ref{sec:delta} and \ref{sec:para_perturb}. For the long-term sensitivity, we are interested in the behavior of 
\[
	\frac{\partial}{\partial\epsilon}\Big\vert_{\epsilon=0}\sup_{\Pi\in\mathcal{X}^\epsilon}\mathbb{E}^\mathbb{P} \bigl[U(\Pi_T)\bigr]
\]
for large $T.$

To achieve this, we combine several techniques: the duality approach (\cite{kramkov1999asymptotic}), the dynamic programming principle, the ergodic Hamilton--Jacobi--Bellman (HJB) equation (\cite{knispel2012asymptotics}), the Hansen--Scheinkman decomposition (\cite{hansen2009long}, \cite{qin2016positive})
and results on sensitivities for long-term cash flows (\cite{park2015sensitivity}). The asymptotic behavior of the sensitivities of \eqref{eqn:max_U} can be characterized by a solution pair $(\lambda,\phi)$ of an ergodic HJB equation. Theorem \ref{thm:decompose} provides an exact representation of the optimal expected utility on a finite time horizon in terms of the asymptotic parameters $(\lambda,\phi)$ with a multiplicative error term. Besides being the main tool for the derivation of the results for the sensitivities, we believe this result is of interest on its own and might be of use for further analysis.  
A precise formulation of the results and a detailed discussion on how the mentioned techniques can be brought together to achieve these results will be given in Section \ref{sec:heuristic}.

To make the objective of this paper clear and discern the problem at hand from similar problems, let us make the formulation we study precise: We consider the problem	
\[
	\lim_{T\rightarrow \infty} \frac{1}{T}\frac{\partial}{\partial\epsilon}\Big\vert_{\epsilon=0}\ln\Bigl\vert \sup_{\Pi\in\mathcal{X}^\epsilon}\mathbb{E}^\mathbb{P} \bigl[U(\Pi_T)\bigr]\Bigr\vert
\]
for an investor with constant relative risk aversion larger than $1$, i.e. utility function $U(x)=\frac{x^p}{p}$, $p<0$. Specifically, we calculate the normalized asymptotic  behavior of the derivative on a log scale. This is different from the problem to optimize the long-term growth rate, where one optimizes over the normalized optimal growth rate on a logarithmic scale and then analyzes its sensitivity. While both questions are economically meaningful, we focus in the current paper on the first type of sensitivity.

\subsection{Related literature}

Many authors have worked on the optimal long-term investment problem: \cite{fleming1995risk} solve the optimization problem of the long-term growth of expected utility for an investor with constant relative risk aversion by reformulating it as an infinite-time horizon risk-sensitive control problem. \cite{guasoni2012portfolios} develops a method to derive optimal portfolios explicitly in a general diffusion model of incomplete markets for an investor with power utility. \cite{liu2013portfolio} explain how to compute optimal portfolios using stochastic control and  convex duality. Special emphasis is placed on long-horizon asymptotics that lead to particularly tractable results. \cite{robertson2015large} study the large time behavior of solutions to semi-linear Cauchy problems with quadratic gradients. Their analysis has direct applications to risk-sensitive control and long-term portfolio choice problems.

Sensitivity analysis of optimal investment for fixed time horizon has also attracted many authors: \cite{kramkov2006sensitivity} conduct a sensitivity analysis of the optimal expected utility with respect to a small change in initial capital or in a portfolio constraint. \cite{larsen2007stability} investigate the stability of utility-maximization in complete and incomplete markets under small perturbations. They identify the topologies on the parameter process space and the solution space under which utility-maximization is a continuous operation. \cite{backhoff2014some} conduct a first order sensitivity analysis of some parameterized stochastic optimal control problems. Their main tool is the one-to-one correspondence between the adjoint states appearing in a weak form of the stochastic Pontryagin principle and the Lagrange multipliers associated to the state equation. \cite{larsen2014expansion} study the first-order approximation for the power investor's value function and its second-order error is quantified in the framework of an incomplete financial market.
\cite{mostovyi2017sensitivity} investigate the sensitivity of the optimal expected utility in a continuous semimartingale market with respect to small changes in the market price of risk. For a general utility function, they derive a second-order expansion of the value function, a 	first-order approximation of the terminal wealth, and construct trading strategies. 
\cite{mostovyi2018asymptotic} develops a sensitivity analysis for the expected utility maximization problem with respect to small perturbations in the numeraire in an incomplete market model, where under an appropriate numeraire the stock price process is driven by a sigma-bounded semimartingale. The author also establishes a second-order expansion of the value function and a first-order approximation
of the terminal wealth. \cite{monin2014optimal} explore ``portfolio Greeks," which measure the sensitivities of an investor's optimal wealth to changes in cumulative excess
stock return, time, and other market parameters. \cite{veraguas2015sensitivity} study the issue of sensitivity with respect to model parameters for the problem of utility maximization from final wealth in an incomplete Samuelson model for utility functions of positive-power type by reformulating  the maximization problem in terms of a convex-analytical support function of a weakly-compact set.

This paper is closely related to and builds upon \cite{park2015sensitivity} who  investigates the long-term sensitivity of the expectation 
\[
	\mathbb{E}\Bigl[e^{-\alpha \int_0^t\theta^2(S_u)\,du}\Bigr]
\]
for a  perturbation of the underlying stochastic process $S.$ In   complete markets, we can use the result of his paper because the optimal expected utility can be expressed by an expectation of this form as we will see in Section \ref{eqn:complete_markets}. In incomplete markets, however, the optimal expected utility cannot be expressed in the above form,
thus one cannot rely on  his   result. We have to use in the current paper more advanced and complicated techniques to tackle the case of incomplete market models.

The current paper is structured as follows.  In the remainder of Section \ref{sec:intro} we discuss the case of a complete market model as a warm up. Section \ref{sec:model_setup} provides the model set up and specifies the market model and the optimization problem. The main idea of this paper is presented in Section \ref{sec:heuristic} using a heuristic argument. In Section \ref{sec:examples}, we display two examples: the Kim--Omberg model and the Heston model. The dual formulation of the utility maximization problem and the Hansen--Scheinkman decomposition are discussed in Section \ref{sec:utility_max} to provide rigorous results in the following: Sensitivity with respect to the initial factor is studied in Section \ref{sec:delta}, and those with respect to the drift and volatility are presented in Section \ref{sec:drift_vol}. Section \ref{sec:conclusion} summarizes the results of this paper. Proofs and detailed calculations are given in appendices.

\subsection{Complete markets}
\label{eqn:complete_markets}

As a warm up, this section discusses the long-term sensitivity of the optimal expected utility in a complete market as this follows easily from \cite{park2015sensitivity}. He investigates the long-term sensitivity of the expectation 
\[
	\mathbb{E}\Bigl[e^{-\alpha \int_0^t\theta^2(S_u)\,du}\Bigr]
\]
for a real number $\alpha,$ a continuous function $\theta$ and an underlying asset process $S$. We  show that the optimal expected utility in a complete market can be expressed as
this form of expectation, and so the results of \cite{park2015sensitivity} directly applied.

We consider the following market model: The price $S$ of a risky asset (e.g., stock) satisfies 
\[
	dS_t =b(S_t)\,dt+\varsigma(S_t)\,dW_t,\qquad S_0=s,
\]
with  $b$ and $\varsigma$ continuous functions, $\varsigma$ positive, such that this SDE has a unique non-explosive strong solution. Here, the process $W$ is a Brownian motion under the physical probability measure $\mathbb{P}.$ Without loss of generality we assume that the short interest rate is zero, so the market price of risk is
\[
	\theta_t:=\theta(S_t)=\frac{b(S_t)}{\varsigma(S_t)}.
\]
An investor with constant relative risk aversion $1-p$, $p<0$, aims to maximize the expected power utility at the terminal time
\begin{equation} \label{eqn:complete_U}
\mathcal{U}(\chi,T):=	\sup_{\Pi\in\mathcal{X}}\mathbb{E}^\mathbb{P}\bigl[U(\Pi_T)\bigr]=\frac{\,1\,}{p}\inf_{\Pi\in\mathcal{X}}\mathbb{E}^\mathbb{P}\bigl[\Pi_T^p\bigr].
\end{equation} 
By the homotheticity of power utility we can assume without loss of generality unit initial capital. Let $\mathbb{P}^*$ be the unique risk-neutral measure, and denote by $L_T$ the Radon--Nikod\'{y}m derivative on $\mathcal{F}_T,$ that is,
\[
	L_T=\frac{d\mathbb{P}^*}{d\mathbb{P}\,}\Big\vert_{\mathcal{F}_T}.
\]
It is known that the optimal investment portfolio  value is 
\[
	\hat{\Pi}_T=c_T(U')^{-1}(L_T)
\]
where $c_T$ is a constant determined by the budget constraint
\[
	1=\mathbb{E}^{\mathbb{P}^*}\bigl[\hat{\Pi}_T\bigr]=c_T\mathbb{E}^{\mathbb{P}^*}\bigl[(U')^{-1}(L_T)\bigr]=c_T\mathbb{E}^{\mathbb{P}^*}\Bigl[L_T^{1/(p-1)}\Bigr].
\]
Thus the optimal expected utility is
\begin{align*}
	\mathbb{E}^\mathbb{P}\bigl[U(\hat{\Pi}_T)\bigl]
	&=\mathbb{E}^\mathbb{P}\bigl[U\bigl(c_T\, U'^{-1}(L_T)\bigl)\bigl]=\frac{1}{p}c_T^p\mathbb{E}^\mathbb{P}\bigl[L_T^{p/(p-1)}\bigl]\\
	&=\frac{1}{p} \frac{\mathbb{E}^\mathbb{P}\bigr[L_T^{p/(p-1)}\bigr]}{\mathbb{E}^{\mathbb{P}^*}\bigl[L_T^{1/(p-1)}\bigr]^p}
	=\frac{1}{p} \frac{\mathbb{E}^{\mathbb{P}^*}\bigl[L_T^{1/(p-1)}\bigr]}{\mathbb{E}^{\mathbb{P}^*}\bigl[L_T^{1/(p-1)}\bigr]^p}
	=\frac{1}{p} \mathbb{E}^{\mathbb{P}^*}\bigl[L_T^{1/(p-1)}\bigr]^{1-p}.
\end{align*}

This expectation can be expressed in terms of the market price of risk since the Radon--Nikod\'{y}m derivative $L_T$ is
\[
	L_t=e^{-\int_0^t\theta_s\,dW_s-\frac{1}{2}\int_0^t\theta_s^2\,ds}=e^{-\int_0^t\theta_s\,dW_s^*+\frac{1}{2}\int_0^t\theta_s^2\,ds},\qquad 0\leq t\leq T,
\]
with $W_t^*:=W_t+\int_0^t\theta_s\,ds$ a $\mathbb{P}^*$-Brownian motion. If we define a measure $\hat{\mathbb{P}}$ from $\mathbb{P}^*$ with the Girsanov kernel $\frac{1}{1-p}\theta_t,$ then   
\begin{align*}
	\mathbb{E}^\mathbb{P}\bigr[U\bigl(\hat{\Pi}_T\bigr)\bigr] 
	&=\frac{1}{p} \mathbb{E}^{\mathbb{P}^*}\bigl[L_T^{1/(p-1)}\bigr]^{1-p}
	=\frac{1}{p} \mathbb{E}^{\mathbb{P}^*}\Bigl[e^{\frac{1}{1-p}\int_0^T\theta_s\,dW_s^*+\frac{1}{2(p-1)}\int_0^T\theta_s^2\,ds}\Bigr]^{1-p}\\
	&=\frac{1}{p}\cdot\mathbb{E}^{\hat{\mathbb{P}}}\Bigl[e^{\frac{p}{2(1-p)^2}\int_0^T\theta_s^2\,ds}\Bigr]^{1-p}
	=\frac{1}{p} \mathbb{E}^{\hat{\mathbb{P}}}\Bigl[e^{-\alpha \int_0^T\theta_s^2\,ds}\Bigr]^{1-p} 
\end{align*}
with $\alpha:=-\frac{p}{2(1-p)^2}.$ The $\hat{\mathbb{P}}$-dynamics of $S$ is 
\begin{equation}
\label{eqn:S-under-Phat}
	dS_{t}=\Bigl(b(S_{t})+\frac{p\varsigma(S_t)\theta(S_t)}{1-p}\Bigr) dt+\varsigma(S_{t})\, d\hat{W}_t
\end{equation}
where $\hat{W}$ is a $\hat{\mathbb{P}}$-Brownian motion. Thus the sensitivity analysis of the optimal expected utility boils down to the sensitivity analysis of
\begin{equation}
\label{eqn:v_s_T}
	v(s,T):=\mathbb{E}^{\hat{\mathbb{P}}}\Bigl[e^{-\alpha \int_0^T\theta^2(S_u)\,du}\Bigr]
\end{equation}
for $s=S_0$ which can be done using the results of  \cite{hansen2009long} and \cite{park2015sensitivity} for the underlying process \eqref{eqn:S-under-Phat}. 

From the Hansen--Scheinkman decomposition, one can find an eigenvalue and an eigenfunction $(\lambda,\phi)$ (called the recurrent eigenpair) of the pricing operator $\mathcal{P}_T$, defined by
\[
	\mathcal{P}_T \phi(s)  = e^{-\lambda T} \phi(s)
\]
where
\[
	\mathcal{P}_T \phi(s) = \mathbb{E}^{\hat{\mathbb{P}}} \Bigl[e^{-\alpha \int_0^T\theta^2(S_u)\,du} \phi(S_T) \, \Big\vert S_0 = s\Bigr].
\]
They characterize the long-term behavior of $v(s,T)$. Specifically, under some assumptions, the limit
\[
	\lim_{T\rightarrow\infty} \frac{v(s,T)}{e^{-\lambda T}\phi(s)}
\]
exists and is independent of $s.$ For the long-term initial-value sensitivity, one can show that  
\[
	\lim_{T\rightarrow\infty}\frac{\partial}{\partial s} \ln v(s,T)=\frac{\phi'(s)}{\phi(s)}.
\]
From this one can derive the actual sensitivity of the expected utility noting
\[
	\frac{\partial}{\partial s} \ln \bigl\vert \mathcal{U}(s,T)\bigr\vert =
	\frac{\partial}{\partial s} \ln \Bigl(-\frac{1}{p} v^{1-p}(s,T)\Bigr) = (1-p) \frac{\partial}{\partial s} \ln v(s,T).
\]
For the parameter sensitivity with respect to the drift and volatility, let $\epsilon$ be the perturbation parameter in the drift or volatility, and denote by $v_\epsilon(s,T)$ the corresponding expectation in Eq.\eqref{eqn:v_s_T}. Using the family of recurrent eigenpairs $(\lambda_\epsilon,\phi_\epsilon)_{\epsilon >0},$ one can prove that   
\[
	\lim_{T\rightarrow\infty}\frac{1}{T}\frac{\partial}{\partial\epsilon}\Big\vert_{\epsilon=0}\ln v_\epsilon(s,T)=-\frac{\partial\lambda_\epsilon}{\partial\epsilon}\Big\vert_{\epsilon=0}
\]
from which the sensitivity of the actual expected utility can be inferred as above by multiplying with $1-p$.

We emphasize that the main line argument in this section cannot be applied to incomplete markets. Our reasoning relies on the fact that in the complete market case the dual optimization problem is posed over a single risk-neutral measure and thus has a trivial solution. This cannot be generalized to a factor diffusion model describing an incomplete market where we have an optimization problem over infinitely many risk-neutral measures.
%

For the rest of this section, we investigate an example of a complete market as discussed above. We study the long-term sensitivity of the optimal expected utility when the underlying asset follows an Ornstein--Uhlenbeck process. 
This is often assumed when modeling commodities as gold, silver and oil. Assume that the asset follows 
\begin{equation}
\label{eqn:OU}
	dS_t=(\mu-bS_t)\,dt+\varsigma \,dW_t,\qquad S_0=s,
\end{equation} 
under the physical measure and the short interest rate is zero. Then, the market price of risk is
\[
	\theta(S_t):=\frac{\mu-bS_t}{\varsigma},
\]
which connects the two Brownian motions, under the physical measure and under the risk-neutral measure, via
\[
	dW_t^*=dW_t+\theta(S_t)\,dt.
\]
We want to analyze the value function $v(s,T)$ given in Eq.\eqref{eqn:v_s_T}. In this case the generator corresponding to the asset price dynamics under $\hat{\mathbb{P}}$ is given by
\[
	\mathcal{L}\phi(s)=\frac{1}{2}\varsigma^2\phi''(s)+\frac{\mu-bs}{1-p}\phi'(s)-\alpha \theta^2(s)\phi(s)
\]
and one can show that the recurrent eigenvalue $\lambda$ and the recurrent eigenfunction $\phi$ of $-\mathcal{L}$ are \begin{align*}
	\lambda =\frac{b(\sqrt{1-p}-1)}{2(1-p) }, \qquad \phi(s) =e^{-\frac{1}{2}As^2-Bs},
\end{align*}
where
\[
A=\frac{b(\sqrt{1-p}-1)}{(1-p)\varsigma^2}, \qquad B=- \frac{\mu(\sqrt{1-p}-1)}{(1-p) \varsigma^2} .
\]
\begin{thm}
 	Under the Ornstein--Uhlenbeck model  in Eq.\eqref{eqn:OU},  the long-term sensitivities  of the optimal expected utility in Eq.\eqref{eqn:complete_U} are
 	\begin{align*}
		 	\lim_{T\rightarrow\infty}\frac{\partial }{\partial s}\ln \bigl\vert\mathcal{U}(s,T)\bigr\vert & =(1-p)\frac{\phi'(s)}{\phi(s)}=-(1-p)(As+B),\\ 
		 	\lim_{T\rightarrow\infty}\frac{1}{T}\frac{\partial }{\partial\mu}\ln \bigl\vert \mathcal{U}(s,T)\bigr\vert
		 	& =-(1-p)\frac{\partial\lambda }{\partial \mu}=0,\\ 
		 	\lim_{T\rightarrow\infty}\frac{1}{T}\frac{\partial }{\partial b}\ln \bigl\vert\mathcal{U}(s,T)\bigr\vert
		 	& =-(1-p)\frac{\partial\lambda }{\partial b}=-\frac{\sqrt{1-p}-1}{2},\\ 
		 	\lim_{T\rightarrow\infty}\frac{1}{T}\frac{\partial }{\partial \varsigma}\ln \bigl\vert \mathcal{U}(s,T)\bigr\vert
		 	& =-(1-p)\frac{\partial\lambda }{\partial \varsigma}=0.\\ 
 	\end{align*}
\end{thm}

\section{Model setup}
\label{sec:model_setup}

The model setup of the current paper is as follows: Let $(\Omega,\mathcal{F},(\mathcal{F}_t)_{t\ge0},\mathbb{P})$ be the canonical path space of a two-dimensional Brownian motion $(W_{1,t},W_{2,t})_{t\ge0}.$ The filtration $(\mathcal{F}_t)_{t\ge0}$ is the usual completion of the natural filtration of $(W_{1,t},W_{2,t})_{t\ge0}.$ The measure $\mathbb{P}$  is referred to as the physical measure. The dynamics of the risky asset    is  given by the following stochastic differential equations (SDEs) 
\begin{align}
	 dS_t & = b(X_t)S_t\,dt+\varsigma(X_t)S_t\,dW_{1,t},\qquad \qquad \qquad \qquad S_0=1, \label{eqn:sv_model_stock} \\
	 dX_t & = m(X_t)\,dt+\sigma_1(X_t)\,dW_{1,t}+\sigma_2(X_t)\,dW_{2,t}, \qquad X_0=\chi, \label{eqn:SDE_X} 
\end{align}
which is a typical way to define a stochastic factor model. The processes $S$ and $X$ describe an asset price and its underlying factor process, respectively. The five functions $m,$ $\sigma_1,$ $\sigma_2,$ $b,$ $\varsigma$ and the real number $\chi$ satisfy the following assumptions. Let $(\ell,r)$ be an open interval in $\mathbb{R}$ for $-\infty\leq \ell<r\leq \infty.$
\begin{assume}
\label{assume:SDE_X}
	Let $\chi\in(\ell,r)$ and let $m,$ $\sigma_1,$ $\sigma_2$ be continuous functions on $(\ell,r)$ such that $\sigma_1^2+\sigma_2^2>0.$ The SDE \eqref{eqn:SDE_X} has a unique non-explosive (i.e., $\mathbb{P}[X_t \in (\ell,r) \text{ for all } t \geq 0]=1$) strong solution $X$.  
\end{assume}

\begin{assume}
\label{assume:b_v}
	The functions $b,$ $\varsigma$  are continuous  and $\varsigma$ is strictly positive on $(\ell,r).$ 
\end{assume}

\noindent Under these assumptions the asset price process is well-defined and can be written as
\[
	S_t=e^{\int_0^t(b-\frac{1}{2}\varsigma^2)(X_s)\,ds+\int_0^t\varsigma(X_s)\,dW_{1,s}}.
\]

\begin{assume}
\label{eqn:NFLVR}
	For each fixed time $T,$ there exists a probability measure on $\mathcal{F}_T$ such that the discounted asset price process is a local martingale on $[0,T]$. 
\end{assume}

\noindent It is well-known that this assumption is equivalent to the absence of arbitrage in the market in the sense of no free lunch with vanishing risk (\cite{delbaen1994general}).

Without loss of generality we will assume that the short interest rate is zero so that the value of the money market account is one at all time $t.$ The market price of risk is then given by
\begin{equation}
\label{eqn:MPR}\theta_t=\theta(X_t)=\frac{b(X_t)}{\varsigma(X_t)}.
\end{equation}
An investor wants to maximize the expected utility of the value of their portfolio at terminal time $T$ by trading the asset and the money market account. A  portfolio is a  predictable processes ${\psi}$ which is $S$-integrable. The value process $\Pi={\Pi}^\psi$ of the  portfolio ${\psi}$ is
\[
	\Pi_t= \Pi_0+\int_0^t {\psi}_{u}\,dS_u, \qquad 0\leq t\leq T.
\]
We denote by $\mathcal{X}$ the family of nonnegative value processes with initial wealth $\Pi_0$ equal to  $1$, that is,
\begin{equation}
\label{eqn:wealth_proc}
	\mathcal{X}=\{\Pi^\psi\geq0: \psi \textnormal{ is a portfolio and } \Pi_0^\psi=1\}. 
\end{equation} 
The investor is assumed to have constant relative risk aversion $1-p>1$, i.e., the utility function corresponding to their  preferences is of negative power type
\[
	U(x)=\frac{x^p}{p},\qquad p<0.
\]
For given initial capital, the goal of the investor is to maximize the expected value at the terminal wealth, that is,
\begin{equation}
\label{eqn:primal}
	\sup_{\Pi\in\mathcal{X}}\mathbb{E}^\mathbb{P}\bigl[U(\Pi_T)\bigr]=\frac{\,1\,}{p}\inf_{\Pi\in\mathcal{X}}\mathbb{E}^\mathbb{P}\bigl[\Pi_T^p\bigr].
\end{equation} 
Without loss of generality we can assume that the initial capital is equal to one, thanks to the homotheticity of the investor's preferences.

%

\section{Heuristic arguments and main results}
\label{sec:heuristic}

The main purpose of the current paper is to investigate two types of long-term sensitivity with respect to the perturbation of $S$ and $X.$  One is the sensitivity with respect to the initial value $\chi=X_0$ of the factor process \eqref{eqn:SDE_X},
\begin{equation}
\label{eqn:delta}
	\frac{\partial}{\partial \chi}\ln \Bigl\vert \sup_{\Pi\in\mathcal{X}}\mathbb{E}^\mathbb{P}[U(\Pi_T)]\Bigr\vert .
\end{equation} 
The other type concerns the sensitivities with respect to the five functions $m,$ $\sigma_1,$ $\sigma_2,$ $b,$ $\varsigma.$  Let $m_\epsilon,$ $\sigma_{1,\epsilon},$ $\sigma_{2,\epsilon},$ $b_\epsilon,$ $\varsigma_\epsilon$ be perturbed functions with perturbation parameter $\epsilon$ (for a precise definition, see Section \ref{sec:para_perturb}). Denote by $S^\epsilon$ the perturbed asset process induced by these perturbed functions, and consider the family $\mathcal{X}^\epsilon$ of wealth processes given by Eq.\eqref{eqn:wealth_proc} generated by the perturbed asset process $S^\epsilon.$ The sensitivity of interest is that with respect to the $\epsilon$-perturbation,
\begin{equation}
\label{eqn:eps_pert}
	\frac{\partial}{\partial \epsilon}\Big\vert_{\epsilon=0}\ln \Bigl\vert \sup_{\Pi\in\mathcal{X}^\epsilon}\mathbb{E}^\mathbb{P}[U(\Pi_T)]\Bigr\vert.
\end{equation}  

\begin{remark}
	 We note that the assumption $S_0=1$ in Eq.\eqref{eqn:sv_model_stock} does not restrict the generality of the results. In fact, in the factor model, the optimal expected utility is independent of the initial value of the stock price as the stock dynamics scale linearly. This is in contrast to the results for the complete market case in Section \ref{eqn:complete_markets}, as there also drift and volatility functions depend on the stock price.
\end{remark}

In the following, we will present the main ideas how to derive the long-term initial-factor sensitivity by surveying the essential steps of the argument. The technical details are relegated to  Section \ref{sec:delta}. 
 
\begin{enumerate}[(i)]
	\item From the dual formulation of utility maximization problem (\cite{kramkov1999asymptotic}, details will be surveyed in Section \ref{sec:dual}), we know that
	\begin{equation}
		\label{eqn:U}
			\mathcal{U}(\chi,T):=	\sup_{\Pi\in\mathcal{X}}\mathbb{E}[U(\Pi_T)] =\frac{1}{\,p\,}\Bigl(\mathbb{E}^\mathbb{P}\bigl[\hat{Y}_T^q\bigr]\Bigr)^{1-p} 
		\end{equation}
	for some nonnegative supermartingale $\hat{Y}$ and $q:=-{p}/{(1-p)};$ define 
	\[
			v(\chi,T):= \mathbb{E}^\mathbb{P}\bigl[\hat{Y}_T^q\bigr] = \mathbb{E}^\mathbb{P}\bigl[\hat{Y}_T^q \, \big\vert \, X_0 = \chi\bigr].
		\]
 	\item The sensitivity in Eq.\eqref{eqn:delta} is 
 	\begin{equation}  
		 	\frac{\partial}{\partial\chi} \ln \Bigl\vert \sup_{\Pi\in\mathcal{X}}\mathbb{E}^\mathbb{P}\bigl[U(\Pi_T)\bigr]\Bigr\vert
		 	=(1-p)\frac{\partial}{\partial\chi} \ln  v(\chi,T),
		\end{equation}
	so it suffices to evaluate  the long-term behavior of	
	\[
			\frac{\partial}{\partial \chi} \ln  v(\chi,T).
		\]  
 	\item The function $v$ satisfies a HJB equation (details are given in Section \ref{sec:dual}).  
	\item The function $v$ can be approximated by a solution pair $(\lambda,\phi)$ of an ergodic HJB equation (see Eq.\eqref{eqn:EBE}) in the sense that $e^{-\lambda T}\phi(\chi)$ is asymptotically equal to $v(\chi,T)$ up to  a constant factor, that is,
	\[
			v(\chi,T)\simeq e^{-\lambda T}\phi(\chi)
		\]
	(where we use the notation $f_T\simeq g_T$ to denote that the limit $\lim_{T\rightarrow\infty}\frac{f_T}{g_T}$ for two positive functions $f_T$ and $g_T$ converges to a positive constant).
	To derive this result, we rely on the HJB representation of $v$ derived in (iii).
	\item  By taking the partial derivative to the above asymptotics, one can anticipate that
	\begin{equation}
		\label{eqn:long_term_delta_value}
			\frac{\partial}{\partial \chi} \ln  v(\chi,T)\simeq \frac{\phi'(\chi)}{\phi(\chi)}
		\end{equation} 
	and this is indeed one of the main results of this paper and is stated in detail in Theorem \ref{thm:delta}. This approach is motivated by Section 3 in  \cite{park2015sensitivity}.
	\item To make this asymptotic result rigorous, one needs to control the error terms. This can be done using a probabilistic representation of the function $v$,
	\[
			v(\chi,T)=e^{-\lambda T}\phi(\chi)\,\mathbb{E}^{\mathbb{Q}}\Bigl[\frac{1}{\phi (X_T)}\, e^{\int_0^Tf(X_s,s;T)\,ds} \Bigr],
		\]
	for a probability measure $\mathbb{Q}$ and a continuous function $f.$ The precise result is given in Theorem \ref{thm:decompose}, the proof relies on an adaption of the Hansen--Scheinkman decomposition to the current context. Thus, by taking the partial derivative directly, we get
	\[
			\frac{\partial}{\partial\chi} \ln v(\chi,T)=\frac{\phi'(\chi)}{\phi(\chi)} +\frac{\partial}{\partial \chi} \ln  \mathbb{E}^{\mathbb{Q}}\Bigl[\frac{1}{\phi (X_T)}\, e^{\int_0^Tf(X_s,s;T)\,ds} \Bigr].
		\]
	Under reasonable conditions the error term 
	\[
			\frac{\partial}{\partial \chi} \ln  \mathbb{E}^{\mathbb{Q}}\Bigl[\frac{1}{\phi (X_T)}\, e^{\int_0^Tf(X_s,s;T)\,ds} \Bigr]
		\]
	 goes to zero as $T\to\infty$ and we obtain Eq.\eqref{eqn:long_term_delta_value}, the desired result.
\end{enumerate}

The following theorem is the main result on the sensitivity with respect to the initial-factor. The proof will be given in Section \ref{sec:delta}.

\begin{thm}
\label{thm:delta}
	Assume A\ref{assume:SDE_X} -- \ref{assume:Q} (stated in Sections \ref{sec:model_setup} and \ref{sec:dual}) and additionally that  the map
	\[
		\chi \mapsto \mathbb{E}^{\mathbb{Q}} \Bigl[\frac{1}{\phi (X_T)}\, e^{\int_0^Tf(X_s,s;T)\,ds} \, \Big\vert \, X_0 = \chi	\Bigr]
		\] 
	is continuously differentiable with derivative converging to zero as $T\to\infty.$ Then 
	\begin{equation}
		\label{eqn:long_term_delta}
			\lim_{T\rightarrow\infty}\frac{\partial }{\partial \chi}\ln v(\chi,T)=\frac{\phi'(\chi)}{\phi(\chi)}.
		\end{equation}
\end{thm}

\begin{remark}
	This result is very similar in spirit to the results by \cite[Eq.(1.4) and Theorem 2.11]{robertson2015large}. They also discuss asymptotic behavior of the type as Eq.\eqref{eqn:long_term_delta}. Their approach as well as the assumptions needed are however different from the current paper.
\end{remark}

For the second topic of the paper, the sensitivities with respect to small perturbation parameters, we proceed in the same way and provide an overview of the main steps of the argument; the technical details will be given at Section \ref{sec:drift_vol}. 
\begin{enumerate}
	\item[(i') -- (iv')] For each $\epsilon,$ we can follow the approach of the sensitivity analysis with respect to the initial factor. Specifically conducting steps (i) -- (iv) as above and defining $v_\epsilon(\chi,T)$ and $(\lambda_\epsilon,\phi_{\epsilon})$ accordingly, we obtain
	\[
			v_\epsilon(\chi,T)\simeq e^{-\lambda_\epsilon T}\phi_\epsilon(\chi).
		\]
	\item[(v')] By taking the partial derivative to the above asymptotics, we have
	\begin{equation}
		\label{eqn:long_term_para_perturb}
			\frac{1}{T}\frac{\partial}{\partial \epsilon}\Big\vert_{\epsilon=0} \ln v_\epsilon(\chi,T) \simeq -\frac{\partial}{\partial \epsilon}\Big\vert_{\epsilon=0} \lambda_\epsilon.
		\end{equation} 
	\item[(vi')] The function $v_\epsilon(x,T)$ has the probabilistic representation
	\[
			v_\epsilon(\chi,T)=e^{-\lambda_\epsilon T}\phi_\epsilon(\chi)\,\mathbb{E}^{\mathbb{Q}_\epsilon}\Bigl[\frac{1}{\phi_\epsilon(X_T^\epsilon)}\, e^{\int_0^Tf_\epsilon(X_s^\epsilon,s;T)\,ds} \Bigr].
		\]
 Thus, by taking the partial derivative, it follows that
	\[
			\frac{1}{T}\frac{\partial}{\partial\epsilon}\Big\vert_{\epsilon=0}\ln v_\epsilon(\chi,T)
			=-\frac{\partial \lambda_\epsilon}{\partial\epsilon}\Big\vert_{\epsilon=0} +\frac{1}{T}\frac{\partial}{\partial\epsilon}\Big\vert_{\epsilon=0}\ln\phi_\epsilon(\chi)
			+\frac{1}{T} \frac{\partial}{\partial\epsilon}\Big\vert_{\epsilon=0}\ln\mathbb{E}^{\mathbb{Q}_\epsilon} \Bigl[\frac{1}{\phi_\epsilon(X_T^\epsilon)}\, e^{\int_0^Tf_\epsilon(X_s^\epsilon,s;T)\,ds}\Bigr].
		\]			
	The second term goes to zero as $T \to \infty$ and under reasonable conditions also the error term 
	\[
			\frac{1}{T}\frac{\partial}{\partial \epsilon}\Big\vert_{\epsilon=0} \ln \mathbb{E}^{\mathbb{Q}_\epsilon}\Bigl[\frac{1}{\phi_\epsilon(X_T^\epsilon)}\, e^{\int_0^Tf_\epsilon(X_s^\epsilon,s;T)\,ds} \Bigr]
		\]
	vanishes as $T \to \infty$, thus we obtain Eq.\eqref{eqn:long_term_para_perturb}.
\end{enumerate}

\begin{thm}
\label{thm:para_sensitivity}
	Assume B\ref{bassume:perturb} -- \ref{bassume:HS_eps}, conditions (i) -- (iii) in Theorem \ref{thm:total_chain} and additionally that the map
	\[
			\epsilon \mapsto \mathbb{E}^{\mathbb{Q}_\epsilon} \Bigl[\frac{1}{\phi(X_T^\epsilon)}\, e^{\int_0^Tf(X_s^\epsilon,s;T)\,ds} \Bigr]
		\]
	 is continuously differentiable at $\epsilon=0$ with
	\begin{equation} 
			\frac{1}{T}\frac{\partial}{\partial\epsilon}\Big\vert_{\epsilon=0}\mathbb{E}^{\mathbb{Q}_\epsilon} \Bigl[\frac{1}{\phi(X_T^\epsilon)}\, e^{\int_0^Tf(X_s^\epsilon,s;T)\,ds} \Bigr]
		\end{equation}
	converging to zero as $T\to\infty.$ Then  
	\[
			\lim_{T\rightarrow \infty}\frac{1}{T} \frac{\partial}{\partial \epsilon}\,\Big\vert_{\epsilon =0}\ln v_\epsilon(\chi,T)
			=-\frac{\partial \lambda_\epsilon}{\partial\epsilon}\Big\vert_{\epsilon=0}.
		\]
\end{thm}

\section{Examples}
\label{sec:examples}

Before implementing the sketched program rigorously, we want to show in this section which results can actually be achieved in specific examples. The power of our approach is demonstrated by deriving explicit formulas for the Kim--Omberg model of stochastic excess returns and the Heston stochastic volatility model.

\subsection{The Kim--Omberg model}
\label{sec:KO_model}

In the Kim--Omberg model (\cite{kim1996dynamic})  the asset price $S$ and the stochastic excess returns $X$ satisfy 
\begin{align}\label{eqn:KO_model}
	dS_t & =\mu X_tS_t\,dt+\varsigma S_t\,dW_{1,t}, \qquad  S_0 =1, \nonumber\\
	dX_t &=k(\overline{m}-X_t)\,dt+\sigma \,dZ_t,\qquad  X_0=\chi
\end{align} 
for correlated Brownian motions $W_{1}$ and $Z$ with correlation parameter $\rho \in(-1,1)$. Here the parameters for the reversion speed $k,$ the volatilities $\varsigma,$ $\sigma$ are positive and the return $\mu,$ the mean reversion level $\overline{m}$ are real numbers. 

This fits into the standard model by setting $\sigma_1 = \rho \sigma,$ $\sigma_2 = \sqrt{1-\rho^2} \sigma$ and $W_{2,t}=\frac{1}{\sqrt{1-\rho^2}}Z_t-\frac{\rho}{\sqrt{1-\rho^2}}W_{1,t}$ so that
\[
	dX_t=k(\overline{m}-X_t)\,dt+\sigma_1 \,dW_{1,t}+\sigma_2 \,dW_{2,t}, \qquad X_0=\chi.
\]    
The market price of risk is given as $\theta_t:=\frac{\mu}{\varsigma}X_t.$ Define
\begin{equation} 
\alpha_1=k+\frac{q\mu\sigma_1}{\varsigma}, \qquad \alpha_2=\sigma_1^2+\frac{\sigma_2^2}{1-q}, \qquad \alpha_3=k\overline{m}, \qquad \alpha_4=\sqrt{\alpha_1^2+q(1-q)\alpha_2\mu^2/\varsigma^2}.
\end{equation} 
and
\begin{equation} 
	B=\frac{\alpha_4-\alpha_1}{\alpha_2}, \qquad C=\frac{\alpha_3(\alpha_4-\alpha_1)}{\alpha_2\alpha_4}
\end{equation}  
for $q$ being the dual exponent of the utility function, $q = -\frac{p}{1-p}$. Then the recurrent eigenpair is 
\begin{align*}
	\lambda &=-\frac{1}{2}\alpha_2C^2+\alpha_3C+\frac{1}{2}\sigma^2B,\\
	\phi(x) &= e^{-\frac{1}{2}Bx^2 -Cx}.
\end{align*}  

\begin{thm}\label{thm:KO_asymp}
	In the  Kim--Omberg model, assume the parameters satisfy
	\begin{equation}
	\label{eqn:KO_para_condi}k+\frac{q\mu\rho\sigma}{\varsigma}+\frac{B\sigma^2}{2}>0.
	\end{equation}   Then   the long-term sensitivities of the optimal expected utility in Eq.\eqref{eqn:U} are given by 
	\begin{align*}
			\lim_{T\rightarrow\infty}\frac{\partial }{\partial \chi}\ln \bigl\vert \mathcal{U}(s,T)\bigr\vert & = -(1-p)(B\chi+C), \\
			\lim_{T\rightarrow\infty}\frac{1}{T}\frac{\partial }{\partial k}\ln \bigl\vert\mathcal{U}(s,T)\bigr\vert & = (1-p)\alpha_2\Bigl(\frac{\overline{m}}{\alpha_3}-\frac{\alpha_4+\alpha_1}{\alpha_4^2 }\Bigr)C^2-(1-p)\Bigl( 2\overline{m} -\frac{\alpha_3(\alpha_4+\alpha_1)}{\alpha_4^2 }\Bigr)C+\frac{(1-p)\sigma^2B}{2\alpha_4},\\
			\lim_{T\rightarrow\infty}\frac{1}{T}\frac{\partial }{\partial \overline{m}}\ln \bigl\vert\mathcal{U}(s,T)\bigr\vert & = \frac{\alpha_2}{\alpha_3}(1-p)kC^2-2(1-p)kC,\\ 
			\lim_{T\rightarrow\infty}\frac{1}{T}\frac{\partial }{\partial \mu}\ln\bigl\vert \mathcal{U}(s,T)\bigr\vert =&\;-\frac{ p\sigma\alpha_1\alpha_3^2( \rho\varsigma\alpha_4^2 -k\rho\varsigma\alpha_1-\mu\sigma\alpha_1) }{\varsigma^2\alpha_2\alpha_4^4}
			-\frac{p\sigma^3(\rho\varsigma\alpha_4-q\rho\varsigma-\mu\sigma)}{2\varsigma^2\alpha_2\alpha_4 } , \\ 		
		\end{align*} 	
	\begin{align*}
			\lim_{T\rightarrow\infty}\frac{1}{T}\frac{\partial }{\partial \varsigma}\ln\bigl\vert \mathcal{U}(s,T)\bigr\vert  &= \frac{p q\mu^2\sigma^2\alpha_1\alpha_3^2(\rho\varsigma\alpha_4^2-k\rho\varsigma\alpha_1  -\mu\sigma\alpha_1)   }{\varsigma^3\alpha_2\alpha_4^4}
			+\frac{p\mu\sigma^3(\rho\varsigma\alpha_4-k\rho\varsigma-\mu\sigma) }{2\varsigma^3\alpha_2\alpha_4}  \\
			\lim_{T\rightarrow\infty}\frac{1}{T}\frac{\partial }{\partial \rho}\ln \bigl\vert \mathcal{U}(s,T)\bigr\vert & =
			-\alpha_2pC^2\Bigl(\frac{(k-\alpha_4)\mu\sigma}{\varsigma\alpha_4(\alpha_4-\alpha_1)}+\frac{\rho \sigma^2}{(1-q)\alpha_2}-\frac{k\mu\sigma}{\varsigma\alpha_4^2}\Bigr)\\
			&+\alpha_3pC\Bigl(\frac{(k-\alpha_4)\mu\sigma}{\varsigma\alpha_4(\alpha_4-\alpha_1)}+\frac{2\rho \sigma^2}{(1-q)\alpha_2}-\frac{k\mu\sigma}{\varsigma\alpha_4^2}\Bigr)\\
			&+\frac{1}{2}p\sigma^2B\Bigl(\frac{(k-\alpha_4)\mu\sigma}{\varsigma\alpha_4(\alpha_4-\alpha_1)}+\frac{2\rho \sigma^2}{(1-q)\alpha_2}\Bigr),\\
			\lim_{T\rightarrow\infty}\frac{1}{T}\frac{\partial }{\partial \sigma}\ln \bigl\vert \mathcal{U}(s,T)\bigr\vert =&\;\alpha_2\Bigl(\frac{p\mu(\rho\varsigma\alpha_4-k\rho\varsigma-\mu\sigma)}{\varsigma^2\alpha_4(\alpha_4-\alpha_1)}- \frac{1-p}{\sigma} +\frac{p\mu(k\rho\varsigma+\mu\sigma)}{\varsigma^2\alpha_4^2} \Bigr)C^2\\
			&-\alpha_3\Bigl(\frac{p\mu(\rho\varsigma\alpha_4-k\rho\varsigma-\mu\sigma)}{\varsigma^2\alpha_4(\alpha_4-\alpha_1)}- \frac{2(1-p) }{\sigma} +\frac{p\mu(k\rho\varsigma+\mu\sigma)}{\varsigma^2\alpha_4^2} \Bigr)C\\ &-\frac{1}{2}p\mu\sigma^2\Bigl(\frac{\rho\varsigma\alpha_4-k\rho\varsigma-\mu\sigma}{\varsigma^2\alpha_4(\alpha_4-\alpha_1)}\Bigr)B.
		\end{align*} 
\end{thm} 
\noindent The proof of these asymptotic results can be found at Appendix \ref{app:KO}.

\subsection{The Heston model}
\label{sec:Heston_model}

In the Heston stochastic volatility model (\cite{heston1993closed}) the asset price $S$ and the stochastic variance process $X$ satisfy 
\begin{align*}
	dS_t &=\mu X_tS_t\,dt+\varsigma\sqrt{X_t} S_t\,dW_{1,t}, \qquad S_0=1,\\
	dX_t &=k(\overline{m}-X_t)\,dt+\sigma \sqrt{X_t} \,dZ_t,\qquad X_0=\chi	
\end{align*} 
for correlated Brownian motions $W_{1}$ and $Z$ with correlation parameter $\rho \in(-1,1)$. Here the parameters for the  reversion speed $k,$ the mean reversion level $\overline{m},$ the volatilities $\varsigma,$ $\sigma$ are positive, and the return $\mu$ is a real number. 
Assume the Feller condition $2k\overline{m}> \sigma^2$, which ensures that the zero boundary of $X$ is inaccessible.

This fits into the standard model by setting $\sigma_1 = \rho \sigma,$ $\sigma_2 = \sqrt{1-\rho^2} \sigma$ and $W_{2,t}=\frac{1}{\sqrt{1-\rho^2}}Z_t-\frac{\rho}{\sqrt{1-\rho^2}}W_{1,t}$ so that
\[
	dX_t=k(\overline{m}-X_t)\,dt+\sigma_1 \sqrt{X_t} \,dW_{1,t}+\sigma_2 \sqrt{X_t} \,dW_{2,t}, \qquad X_0=\chi.
\]   
The market price of risk is $\theta_t:=\frac{\mu}{\varsigma}\sqrt{X_t}.$
Define
\[
	\beta_1:=k+\frac{q\mu \rho\sigma}{\varsigma}, \qquad \beta_2:=\sqrt{\beta_1^2+ q (1-q\rho^2)\mu^2\sigma^2/\varsigma^2 },
\]
and
\[
	B=\frac{(1-q)(\beta_2-\beta_1) }{(1-q\rho^2) \sigma^2}.
\]
Then the recurrent eigenpair is
\[
	\lambda=k\overline{m}B, \qquad \phi(x) = e^{-Bx}.
\]

\begin{thm}\label{thm:Heston}
	In the  Heston model,  assume the Feller condition $2k\overline{m}> \sigma^2$ and
		\begin{equation}
		\label{eqn:Hes_para_condi}
		k+\frac{q\mu\rho\sigma}{\varsigma}>0.
		\end{equation} 
	Then 
	 the long-term sensitivities  of the optimal expected utility in Eq.\eqref{eqn:U} are
	\begin{align}
			\lim_{T\rightarrow\infty}\frac{\partial }{\partial \chi}\ln \bigl\vert \mathcal{U}(s,T)\bigr\vert & =-(1-p)B,  \\
			\lim_{T\rightarrow\infty}\frac{1}{T}\frac{\partial }{\partial k}\ln \bigl\vert \mathcal{U}(s,T)\bigr\vert & =(1-p)\overline{m}B(\frac{k}{\beta_2}-1),\\
			\lim_{T\rightarrow\infty}\frac{1}{T}\frac{\partial }{\partial \overline{m}}\ln \bigl\vert \mathcal{U}(s,T)\bigr\vert &=-(1-p)kB, \\ 
			\lim_{T\rightarrow\infty}\frac{1}{T}\frac{\partial }{\partial \mu}\ln \bigl\vert\mathcal{U}(s,T)\bigr\vert &=\frac{k\overline{m}q(\rho\varsigma\beta_2-k \rho  \varsigma- \mu\sigma) }{(1-q\rho^2)\sigma\varsigma^2\beta_2} , \\ 			
			\lim_{T\rightarrow\infty}\frac{1}{T}\frac{\partial }{\partial \varsigma}\ln \bigl\vert \mathcal{U}(s,T)\bigr\vert &=\frac{ k\overline{m}p\mu\sigma B (\rho\varsigma\beta_2-\rho k \varsigma-\mu\sigma) }{\varsigma^3\beta_2(\beta_2-\beta_1)},\\
			\lim_{T\rightarrow\infty}\frac{1}{T}\frac{\partial }{\partial \rho}\ln \bigl\vert\mathcal{U}(s,T)\bigr\vert &=k\overline{m}B\Bigl(-\frac{p\mu \sigma (\beta_2-k)}{\varsigma\beta_2(\beta_2-\beta_1)}+\frac{2p\rho}{1-q\rho^2}\Bigr),\\
			\lim_{T\rightarrow\infty}\frac{1}{T}\frac{\partial }{\partial \sigma}\ln \bigl\vert\mathcal{U}(s,T)\bigr\vert &=k\overline{m}B\Bigl(\frac{2(1-p)}{\sigma}+\frac{p\mu(k \rho\varsigma+\mu\sigma-\rho\varsigma\beta_2) }{\varsigma^2\beta_2(\beta_2-\beta_1)}\Bigr).
		\end{align}	
\end{thm}
\noindent The proof of these asymptotic results can be found at Appendix \ref{app:Heston}.  

\begin{remark}
\label{remark:para_condition}
	The conditions in  Eq.\eqref{eqn:KO_para_condi} and  Eq.\eqref{eqn:Hes_para_condi} are there to guarantee that the process $X$ is still mean-reverting under the measures relevant for the analysis (details are discussed in the Appendices \ref{app:KO} and \ref{app:Heston}). This condition is in spirit similar to the conditions one finds in the long term analysis of implied volatility in these models where the asymptotic regime depends on the mean-reversion property under the share measure (see, e.g., \cite[Theorem 2.1]{forde2011large} and \cite[Section 6.1]{keller2011moment}.		
\end{remark}

\section{Utility maximization problem}
\label{sec:utility_max}

We provide a mathematical background for the heuristic argument given in Section \ref{sec:heuristic}. First we discuss the dual formulation of the utility maximization problem and its characterization via the solution of an HJB equation. Then we introduce the ergodic HJB equation who can characterize the long-run problem and analyze it in terms of its eigenpair. Finally we generalize the Hansen--Scheinkmann decomposition to functionals of time-inhomogeneous Markov process to lay the ground for the following sensitivity analysis. On the way we make precise the assumptions that are needed for our conclusions.

\subsection{Dual formulation and HJB equations}   
\label{sec:dual}

One of the main ideas is to employ the dual formulation of the utility maximization problem as presented in \cite{kramkov1999asymptotic}. We recall (see Eq.\eqref{eqn:primal}) the primal problem of utility maximization is
\begin{equation}
	\sup_{\Pi\in\mathcal{X}}\mathbb{E}^\mathbb{P}\bigl[U(\Pi_T)\bigr]
\end{equation} 
This primal problem is related to the following dual formulation, which is a minimization problem
\begin{equation}
\label{eqn:lambda}
	\inf_{Y\in\mathcal{Y}}\mathbb{E}^\mathbb{P}\bigl[V(Y_T)\bigr]= \inf_{Y\in\mathcal{Y}}\mathbb{E}^\mathbb{P}\bigl[-{Y_T^q}/{q}\bigr],
\end{equation}
where $q=-\frac{p}{1-p}$ is the conjugate exponent of $p$ and $V(y)=-\frac{y^q}{q}$ is the dual conjugate of the utility function $U.$  Here, $\mathcal{Y}$ is the family of nonnegative semimartingales $Y$ with $Y_0=1$   such that the product $(X_tY_t)_{t\ge0}$ is a supermartingale for any $X\in\mathcal{X}.$ Denote  by  $\hat{Y}$ the optimal element in $\mathcal{Y}$ (Theorem 2.2 in \cite{kramkov1999asymptotic} guarantees the existence of this optimum) and define
\begin{equation}
\label{eqn:v}
	v(\chi,T):= \mathbb{E}^\mathbb{P}\bigl[\hat{Y}_T^q \bigr] = \mathbb{E}^\mathbb{P}\bigl[\hat{Y}_T^q \, \big\vert \, X_0=\chi\bigr]. 
\end{equation}
We emphasize that here $\chi$ is the initial value of the factor process. Note that the function $v$ is not the actual dual value function but a constant multiple of it. This follows from normalizing the dual initial condition which can be done thanks to the homotheticity of the power function $y^q$. From   Eq.(4.10)   in \cite{larsen2014expansion}, we know
\begin{equation}\label{eqn:dual_form}
	\sup_{\Pi\in\mathcal{X}}\mathbb{E}\bigl[U(\Pi_T)\bigr]   =\frac{1}{\,p\,}v^{1-p}(\chi,T)
\end{equation}  
so that the long-term growth rate of the optimal expected utility in Eq.\eqref{eqn:long_term_rate} is 
\begin{equation}
	\lim_{T\rightarrow\infty}\frac{1}{T}\ln \Bigl\vert\sup_{\Pi\in\mathcal{X}}\mathbb{E}^\mathbb{P}\bigl[U(\Pi_T)\bigr]\Bigr\vert =(1-p)\lim_{T\rightarrow\infty}\frac{1}{T}\ln v(\chi,T).
\end{equation} 

Under some conditions, we can characterize the function $v$ as a solution of a HJB equation
\begin{equation}
\label{eqn:HJB}
	v_t=\frac{1}{2}\bigl(\sigma_1^2(x)+\sigma_2^2(x)\bigr)v_{xx} +\sup_{\xi\in\mathbb{R}}\bigl\{l(\xi,x)v+h(\xi,x)v_x\bigr\}, \qquad v(x,0)=1 
\end{equation}
where
\begin{align}
\label{eqn:l_h}
	l(\xi,x) & :=-\frac{q}{2}(1-q)\bigl(\theta^2(x)+\xi^2\bigr) \nonumber \\ h(\xi,x) & := m(x)-q\theta(x)\sigma_1(x)-q\xi\sigma_2(x).
\end{align}
Moreover, the optimal element $\hat{Y}\in\mathcal{Y}$ of Eq.\eqref{eqn:lambda} can be expressed as
\begin{equation}
\label{eqn:opt}
	\hat{Y}_t=e^{-\int_0^t\theta(X_s)\,dW_{1,s}-\frac{1}{2}\int_0^t\theta^2(X_s)\,ds-\int_0^t\hat{\xi}(X_s,s;T)\,dW_{2,s}-\frac{1}{2}\int_0^t\hat{\xi}^2(X_s,s;T)\,ds}, \qquad 0\leq t\leq T,
\end{equation}
where $\theta(X_t):=\frac{b(X_t)}{\varsigma(X_t)}$ is the market price of risk and 
\begin{equation}\label{eqn:opt_xi}
	\hat{\xi}(x,t;T) :=-\frac{\sigma_2(x)v_x(x,T-t)}{(1-q)v(x,T-t)} 
\end{equation} 
is the optimal control of the HJB equation \eqref{eqn:HJB}. Under appropriate conditions the function $v$ can be approximated using a solution pair $(\lambda,\phi)$ of
\begin{equation}\label{eqn:EBE}
	-\lambda\phi(x)=\frac{1}{2}\bigl(\sigma_1^2(x)+\sigma_2^2(x)\bigr)\phi_{xx}+\sup_{\xi\in\mathbb{R}}\bigl\{l(\xi,x)\phi+h(\xi,x)\phi_x\bigr\}
\end{equation}  
which is called the ergodic HJB equation. It is noteworthy that the real number $\lambda$ and the function $\phi$ can be regarded as an eigenvalue and an eigenfunction of the operator $-\mathcal{L}$ where
\[
	\mathcal{L}\phi=\frac{1}{2}\bigl(\sigma_1^2(x)+\sigma_2^2(x)\bigr)\phi_{xx}+\sup_{\xi\in\mathbb{R}}\bigl\{l(\xi,x)\phi+h(\xi,x)\phi_x\bigr\}.
\]
We will review the motivation of these arguments and the derivation in Appendix \ref{app:conn_HJB}.

We make the following assumptions on the function $v$ and the structure of the optimal element $\hat{Y}\in\mathcal{Y}$ of the dual problem without going into further details. For sufficient conditions and a more detailed discussion we refer to \cite[p. 10--12]{knispel2012asymptotics}, \cite[Section 4]{hernandez2006robust} and \cite[Sections 3 and 5]{kaise2006structure}.

\begin{assume}
\label{assume:diff_v}
	The function $v(x,t)$ given by \eqref{eqn:v} is twice continuously differentiable in $x$ and once in $t$ and satisfies the PDE \eqref{eqn:HJB}.
\end{assume}

\begin{assume}
\label{assume:struc_opt}
	The optimizer $\hat{Y}$ of Eq.\eqref{eqn:lambda} is given by Eq.\eqref{eqn:opt}.
\end{assume}

\begin{assume}
\label{assume:long-term}
	There  exist a real number $\lambda$ and a continuously twice-differentiable positive function $\phi$ satisfying Eq.\eqref{eqn:EBE}	such that
\begin{equation}
\label{eqn:conv}
\frac{v(x,t)}{e^{-\lambda t}\phi(x)}\rightarrow C\;\textnormal{ as }\;\; t\rightarrow\infty 
\end{equation} 
	for a  positive  constant $C$ not depending on $x.$  
\end{assume}

We can represent the function $v$ in a simpler way. From  Eq.\eqref{eqn:v}, Eq.\eqref{eqn:opt} and A\ref{assume:struc_opt}, it follows that
\begin{align}
\label{eqn:v_expression}
	v(\chi,T):= \mathbb{E}^\mathbb{P}\bigl[\hat{Y}_T^q \bigr] 
	&=\mathbb{E}^\mathbb{P}\bigl[e^{-q\int_0^T\theta(X_s)\,dW_{1,s}-\frac{q}{2}\int_0^T\theta^2(X_s)\,ds-q\int_0^T\hat{\xi}(X_s,s;T)\,dW_{2,s}-\frac{q}{2}\int_0^T\hat{\xi}^2(X_s,s;T)\,ds}\bigr] \nonumber\\
	&=\mathbb{E}^\mathbb{\hat{P}}\bigl[e^{- \frac{q}{2}(1-q)\int_0^T(\theta^2(X_s)+\hat{\xi}^2(X_s,s;T))\,ds}\bigr] 
\end{align}
where $\mathbb{\hat{P}}$ is a measure on $\mathcal{F}_T$ defined as
\begin{equation}
\label{eqn:hat_P_from_P_RN}
\frac{d\mathbb{\hat{P}}}{d\mathbb{P}}=\mathcal{E}\Bigl(-q\int_0^{\cdot}\theta(X_s)\,dW_{1,s}-q\int_0^\cdot\hat{\xi}(X_s,s;T)\,dW_{2,s}\Bigr)_T
\end{equation} 
under A\ref{assume:P} stated below.
The $\mathbb{\hat{P}}$-dynamics of $X$ is 
\[
dX_t= (m(X_t)-q\theta(X_t)\sigma_1(X_t)-q\hat{\xi}(X_t,t;T)\sigma_2(X_t))\,dt+\sigma_1(X_1)\,d\hat{W}_{1,t}+\sigma_2(X_t)\,d\hat{W}_{2,t}
\]
for a $\mathbb{\hat{P}}$-Brownian motion $(\hat{W}_{1,t},\hat{W}_{2,t}).$
 
\begin{assume}\label{assume:P}
	For the function $\hat{\xi}$ given by Eq.\eqref{eqn:opt_xi}, the local martingale
\begin{equation}
	\biggl(\mathcal{E}\Bigl(-q\int_0^\cdot\theta(X_s)\,dW_{1,s}-q\int_0^\cdot\hat{\xi}(X_s,s;T)\,dW_{2,s}\bigr)_t	\biggr)_{0\leq t\leq T} 
\end{equation}  
is a true martingale under the measure  ${\mathbb{P}}$.
\end{assume}

The solution pair $(\lambda,\phi)$ describes the long-term behavior of $v(\chi,T)$ for initial factor $\chi$ and maturity $T$ as $T\to\infty.$ The long-term growth rate of the optimal expected utility is defined as
\begin{equation}
\label{eqn:long_term_rate}
	 \lim_{T\rightarrow\infty}\frac{1}{T}\ln\Bigl\vert \sup_{\Pi\in\mathcal{X}}\mathbb{E}^\mathbb{P}\bigl[U(\Pi_T)\bigr]\Bigr\vert
\end{equation} 
and can be described by the eigenvalue $\lambda$ since
\begin{equation} 
	 -\lambda=\lim_{T\rightarrow\infty}\frac{1}{T}\ln v(\chi,T)=\frac{1}{1-p}\lim_{T\rightarrow\infty}\frac{1}{T}\ln\Bigl\vert \sup_{\Pi\in\mathcal{X}}\mathbb{E}^\mathbb{P}\bigl[U(\Pi_T)\bigr]\Bigr\vert,
\end{equation}
which follows from Eq.\eqref{eqn:dual_form}. The optimal control of the ergodic HJB equation \eqref{eqn:EBE} is a function of $x,$ so we denote by $\xi^*(x).$ It is easy to check that $\xi^*$ is given by
\begin{equation}
\label{eqn:xi_star}	
	\xi^*(x)= -\frac{\sigma_2(x)\phi_x(x)}{(1-q)\phi(x)}
\end{equation}
and Eq.\eqref{eqn:EBE} becomes
\begin{equation}\label{eqn:EBEs}
-\lambda\phi(x)=\frac{1}{2}\bigl(\sigma_1^2(x)+\sigma_2^2(x)\bigr)\phi_{xx}+h(\xi^*(x),x)\phi_x+l(\xi^*(x),x)\phi.
\end{equation}  
The long-term growth rate  can be calculated as
\[
	-\lambda=\lim_{T\rightarrow\infty}\frac{1}{T}\ln\mathbb{E}^\mathbb{\hat{P}}\Bigl[e^{ \int_0^Tl(\xi^*(X_s),X_s)\,ds}\Bigr].
\]

\subsection{Hansen--Scheinkman decomposition}
\label{sec:HS}

This section is inspired by the Hansen--Scheinkman decomposition in \cite{hansen2009long} and is adapted to the current context. They study the decomposition of a multiplicative functional of a time-homogeneous Markov process into a product of an exponential of the eigenvalue, the eigenfunction and an error term. In this section, we adapt their method to a time-inhomogeneous Markov case. Recall from Eq.\eqref{eqn:v_expression} that 
\begin{equation}
\label{eqn:v_eps}
		v(\chi,T)=\mathbb{E}^\mathbb{\hat{P}}\Bigl[e^{-	\frac{q}{2}(1-q)\int_0^T \bigl(\theta^2(X_s)+\hat{\xi}^2(X_s,s;T)\bigr)\,ds} \, \Big\vert \, X_0 = \chi \Bigr]
\end{equation}
and the $\hat{\mathbb{P}}$-dynamics of $X$ is
\[
	dX_t= (m(X_t)-q\theta(X_t)\sigma_1(X_t)-q\hat{\xi}(X_t,t;T)\sigma_2(X_t))\,dt+\sigma_1(X_1)\,d\hat{W}_{1,t}+\sigma_2(X_t)\,d\hat{W}_{2,t}
\] 
for $0\leq t\leq T$ and a two-dimensional $\hat{\mathbb{P}}$-Brownian motion $(\hat{W}_{1,t},\hat{W}_{2,t}).$ 

We assume the following condition:
\begin{assume}
\label{assume:P_hat}
	For the functions $\hat{\xi}$ given by Eq.\eqref{eqn:opt_xi} and $\xi^*$ given by Eq.\eqref{eqn:xi_star}, the local martingale
	\[
			\biggl(\mathcal{E}\Bigl(q\int_0^\cdot\hat{\xi} (X_s,s;T)-\xi^*(X_s)\,d\hat{W}_{2,s}\Bigr)_t\biggr)_{0\leq t\leq T}
		\] 
	is a true martingale under the measure  $\hat{\mathbb{P}}$.
\end{assume} 
\noindent Define a new measure $\tilde{\mathbb{P}}_T$ on $\mathcal{F}_T$ by
\begin{equation}
\label{eqn:tilde_P}
	\frac{d\tilde{\mathbb{P}}_T}{d\hat{\mathbb{P}}}\Big\vert_{\mathcal{F}_T}  =\mathcal{E}\Bigl(q\int_0^\cdot\hat{\xi} (X_s,s;T)-\xi^*(X_s)\,d\hat{W}_{2,s}\Bigr)_T .
\end{equation} 
For simplicity we drop the subscript $T$ and write just $\tilde{\mathbb{P}}.$
The $\tilde{\mathbb{P}} $-dynamics of $X$ is
\begin{equation}
\begin{aligned}
dX_t 
&= \bigl(m (X_t )-q\theta (X_t )\sigma_{1}(X_t )-q\xi ^*(X_t )\sigma_{2}(X_t )\bigr)\,dt+\sigma_{1}(X_t )\,d\tilde{W}_{1,t}+\sigma_{2}(X_t )\,d\tilde{W}_{2,t} 
\end{aligned}
\end{equation} 
with two-dimensional $\hat{\mathbb{P}}$-Brownian motion 
\begin{equation}
	\begin{pmatrix}
		d\tilde{W}_{1,t}\\
		d\tilde{W}_{2,t}
	\end{pmatrix}
	=\begin{pmatrix}
		0\\
		q\xi^*(X_t )-q\hat{\xi} (X_t,t;T)
	\end{pmatrix}dt
		+\begin{pmatrix}
		d\hat{W}_{1,t}\\
		d\hat{W}_{2,t}
	\end{pmatrix}.
\end{equation}  
From Eq.\eqref{eqn:v_eps} it follows that
\begin{align*}
	v (\chi,T) &=\mathbb{E}^\mathbb{\hat{P} } \Bigl[e^{-\frac{q}{2}(1-q)\int_0^T	(\theta ^2(X_s)+\xi ^{*2}(X_s ))\,ds}\,e^{-\frac{q}{2}(1-q)\int_0^T	(\hat{\xi} ^2(X_s,s ;T)-\xi ^{*2}(X_s ))\,ds}\Bigr]\\
&=\mathbb{E}^\mathbb{\tilde{P} } \Bigl[e^{-	\frac{q}{2}(1-q)\int_0^T (\theta ^2(X_s)+\xi ^{*2}(X_s ))\,ds}\, e^{-\frac{q}{2}(1-q)\int_0^T (\hat{\xi} ^2(X_s,s ;T)-\xi ^{*2}(X_s ))\,ds}\,\frac{d\hat{\mathbb{P}} }{d\tilde{\mathbb{P}} }\Bigr].
\end{align*}

Define 
	\begin{align*}
M_t &:=\frac{\phi(X_t)}{\phi(\chi)}\,e^{\lambda t-	\frac{q}{2}(1-q)\int_0^t (\theta^2(X_s)+\xi^{*2}(X_s))\,ds}, \qquad ;0\le t\le T.
\end{align*}
Then applying the It\^{o} formula to $M$ and using the ergodic HJB equation \eqref{eqn:EBEs}, it can be checked that
\[
M_t =\mathcal{E}\Bigl(\int_0^\cdot\frac{\phi'(X_s)}{\phi(X_s)}\sigma_{1}(X_s)\,d\tilde{W}_{1,s}+\int_0^\cdot\frac{\phi'(X_s)}{\phi(X_s)}\sigma_{2}(X_s)\,d\tilde{W}_{2,s} \Bigr)_t, \qquad 0\le t\le T ,
\]
and thus a $\tilde{\mathbb{P}}$-local martingale.

\begin{assume}\label{assume:M}
	With the solution pair $(\lambda,\phi)$ of A\ref{assume:long-term}, the $\tilde{\mathbb{P}}$-local martingale $(M_t)_{0\leq t\leq T}$
	is a true martingale.
\end{assume} 


We use this random variable $M_T$ as a Radon--Nikod\'{y}m derivative to defined a new measure $\overline{\mathbb{P}},$ that is
\begin{equation}
\label{eqn:overline_P}
	\frac{d\overline{\mathbb{P}}}{d\tilde{\mathbb{P}}}\Big\vert_{\mathcal{F}_T}=M_T.
\end{equation} 
This measure $\overline{\mathbb{P}}$ depends on $T, $ but we suppress in the notation the dependence on $T$ as before. Then
\begin{align*}
	v(\chi,T) &=e^{-\lambda T}\phi(\chi)\mathbb{E}^\mathbb{\tilde{P}} \Bigl[\frac{M_T}{\phi(X_T)}\, e^{-\frac{q}{2}(1-q)\int_0^T (\hat{\xi}^2(X_s,s;T)-\xi^{*2}(X_s))\,ds}\,\frac{d\hat{\mathbb{P}}}{d\tilde{\mathbb{P}}}\Bigr]\\
	&=e^{-\lambda T}\phi(\chi)\mathbb{E}^\mathbb{\overline{P}} \Bigl[\frac{1}{\phi(X_T)}\, e^{-\frac{q}{2}(1-q)\int_0^T (\hat{\xi}^2(X_s,s;T)-\xi^{*2}(X_s))\,ds}\,\frac{d\hat{\mathbb{P}}}{d\tilde{\mathbb{P}}}\Bigr].
\end{align*}
The process 
\begin{equation}
	\begin{pmatrix}
		d\overline{W}_{1,t}\\
		d\overline{W}_{2,t}
	\end{pmatrix}
	=-\begin{pmatrix}
		\frac{\phi'(X_t)}{\phi(X_t)}\sigma_{1}(X_t)\\
		\frac{\phi'(X_t)}{\phi(X_t)}\sigma_{2}(X_t)
	\end{pmatrix}dt
	+\begin{pmatrix}
		d\tilde{W}_{1,t}\\
		d\tilde{W}_{2,t}
	\end{pmatrix}
\end{equation} 
is a Brownian motion under  $\overline{\mathbb{P}}$ and the $\overline{\mathbb{P}}$-dynamics of $X$ is
\begin{align}
\label{eqn:bar_P_dynamics_X}
	dX_t=&\,\Bigl(m(X_t)-q\theta(X_t)\sigma_{1}(X_t)-q\xi^*(X_t)\sigma_{2}(X_t)+\frac{\phi'(X_t)}{\phi(X_t)}\bigl(\sigma_{1}^2(X_t)+\sigma_{2}^2(X_t)\bigr)\Bigr)\,dt \nonumber\\
	&\,+\sigma_{1}(X_t)\,d\overline{W}_{1,t}+\sigma_{2}(X_t)\,d\overline{W}_{2,t}.  
\end{align}
 
 We now perform another change of measure to express the function $v(\chi,T)$ in a more manageable way. Before doing so, we express the
 Radon--Nikod\'{y}m derivative $\frac{d\hat{\mathbb{P}}}{d\tilde{\mathbb{P}}}$ in a different way to facilitate the calculation:
\begin{align*}
	\frac{d\hat{\mathbb{P}}}{d\tilde{\mathbb{P}}}
	&=e^{q\int_0^T\xi^*(X_s)-\hat{\xi}(X_s,s;T)\,d\hat{W}_{2,s}+\frac{q^2}{2}\int_0^T(\xi^*(X_s)-\hat{\xi}(X_s,s;T))^2\,ds}\\
	&=e^{q\int_0^T\xi^*(X_s)-\hat{\xi}(X_s,s;T)\,d\tilde{W}_{2,s}-\frac{q^2}{2}\int_0^T(\xi^*(X_s)-\hat{\xi}(X_s,s;T))^2\,ds} \\
	&=e^{q\int_0^T\xi^*(X_s)-\hat{\xi}(X_s,s;T)\,d\overline{W}_{2,s}-\frac{q^2}{2}\int_0^T(\xi^*(X_s)-\hat{\xi}(X_s,s;T))^2\,ds+q\int_0^T(\xi^*(X_s)-\hat{\xi}(X_s,s;T))\frac{\phi'(X_s)}{\phi(X_s)}\sigma_{2}(X_s)\,ds}.
\end{align*} 
We will need an additional assumption that the argument works:

\begin{assume}\label{assume:Q}
	The local martingale
	\[
			\left(\mathcal{E}\left(q\int_0^\cdot\xi^*(X_s)-\hat{\xi} (X_s,s;T)\,d\overline{W}_{2,s}\right)_t\right)_{0\leq t\leq T}
		\] 
	is a true martingale under the measure $\overline{\mathbb{P}}$.
\end{assume}
\noindent We now define a new measure $\mathbb{Q}$ by
\begin{equation}
\label{eqn:bar_P_Q}
	\frac{d\mathbb{Q}}{d\overline{\mathbb{P}}}=\mathcal{E}\left(q\int_0^\cdot\xi^*(X_s)-\hat{\xi} (X_s,s;T)\,d\overline{W}_{2,s}\right)_T.
\end{equation}
This measure $\mathbb{Q}$  depends on $T, $ but we suppress in the notation the dependence on $T$ as before.
Then 
\begin{align}
\label{eqn:final_g_eps}
	& \phantom{==} v(\chi,T) \nonumber \\
	& = e^{-\lambda T}\phi(\chi)\,\mathbb{E}^\mathbb{\overline{P}} \Bigl[\frac{1}{\phi(X_T)}\, e^{-\frac{q}{2}(1-q)\int_0^T (\hat{\xi}^2(X_s,s;T)-\xi^{*2}(X_s))\,ds}\, e^{q\int_0^T(\xi^*(X_s)-\hat{\xi}(X_s,s;T))\frac{\phi'(X_s)}{\phi (X_s)}\sigma_{2}(X_s)\,ds}\,\frac{d\mathbb{Q}}{d\overline{\mathbb{P}}}\Bigr] \nonumber\\
	& = e^{-\lambda T}\phi(\chi)\,\mathbb{E}^\mathbb{Q} \Bigl[\frac{1}{\phi(X_T)}\, e^{-\frac{q}{2}(1-q)\int_0^T(\hat{\xi}^2(X_s,s;T)-\xi^{*2}(X_s))\,ds}\, e^{q\int_0^T(\xi^*(X_s)-\hat{\xi}(X_s,s;T))\frac{\phi'(X_s)}{\phi (X_s)}\sigma_{2}(X_s)\,ds}\Bigr] \nonumber \\
	& = e^{-\lambda T}\phi(\chi)\,\mathbb{E}^\mathbb{Q} \Bigl[\frac{1}{\phi(X_T)}\, e^{-\frac{q}{2}(1-q)\int_0^T(\xi^*(X_s)-\hat{\xi}(X_s,s;T))^2\,ds}\Bigr].
\end{align}
For the last equality, we used Eq.\eqref{eqn:xi_star}. 
The $\mathbb{Q}$-dynamics of $X$ is
\begin{align}
\label{eqn:final_X}
	dX_t  & = \Bigl(m(X_t)-q\theta(X_t)\sigma_{1}(X_t)-q\hat{\xi}(X_t,t;T)\sigma_{2}(X_t)+\frac{\phi'(X_t)}{\phi(X_t)}\bigl(\sigma_{1}^2(X_t)+\sigma_{2}^2(X_t)\bigr)\Bigr)\,dt \nonumber \\
	&\phantom{==} +\sigma_{1}(X_t)\,dB_{1,t}+\sigma_{2}(X_t)\,dB_{2,t}
\end{align}
where
\begin{equation}
	\begin{pmatrix}
		dB_{1,t}\\
		dB_{2,t}
	\end{pmatrix}
	=\begin{pmatrix}
		0\\
		\hat{\xi}(X_t,t;T)-\xi^*(X_t)
	\end{pmatrix}dt+
	\begin{pmatrix}
		d\overline{W}_{1,t}\\
		d\overline{W}_{2,t}
	\end{pmatrix}
\end{equation} 
is a $\mathbb{Q}$-Brownian motion.

In conclusion, we can express the function $v(\chi,T)$  and the dynamics of $X$ in a simpler way. The following theorem follows from Eq.\eqref{eqn:final_g_eps} and Eq.\eqref{eqn:final_X}.
\begin{thm}
\label{thm:decompose}
	Assume  A\ref{assume:SDE_X} -- \ref{assume:Q}. Then  the function $v(\chi,T)$ can be decomposed as  
	\begin{equation}
		\label{eqn:v_decomposition}
			v(\chi,T)=e^{-\lambda T}\phi(\chi)\,\mathbb{E}^{\mathbb{Q}}\Bigl[\frac{1}{\phi (X_T)}\, e^{\int_0^Tf(X_s,s;T)\,ds} \Bigr]
		\end{equation}	 
	and the $\mathbb{Q}$-dynamics of $X$ is
	\begin{equation}
		\label{eqn:dX_kappa}
			dX_t = \kappa(X_t,t;T) \,dt+\sigma_{1}(X_t)\,dB_{1,t}+\sigma_{2}(X_t)\,dB_{2,t}, \qquad 0\leq t\leq T,
		\end{equation}
	where  
	\begin{align}
		\label{eqn:kappa}
			 f(x,t;T)&=-\frac{q}{2}(1-q)\bigl(\xi^*(x)-\hat{\xi}(x,t;T)\bigr)^2 \nonumber \\
			\kappa(x,t;T)&=m(x)-q\theta(x)\sigma_{1}(x)-q\hat{\xi}(x,t;T)\sigma_{2}(x)+\frac{\phi'(x)}{\phi(x)}\bigl(\sigma_{1}^2(x)+\sigma_{2}^2(x)\bigr) .
		\end{align} 
\end{thm}


\begin{remark}
	A way to understand the above theorem is to consider the  commutative diagram 

		\vspace{10pt}
		
		\begin{tikzcd}
		(\Omega,\mathcal{F}, \hat{\mathbb{P}}) \arrow[r, rightsquigarrow] \arrow[d,"\frac{d\tilde{\mathbb{P}}}{d\hat{\mathbb{P}}}"]
		& (\Omega,\mathcal{F}, \mathbb{Q}) 
		& \text{non-ergodic HJB (optimal control } \hat{\xi}(x,t;T) \text{ in drift of $X$)}\\
		(\Omega,\mathcal{F}, \tilde{\mathbb{P}}) \arrow[r,"M_T =\frac{d\overline{\mathbb{P}}}{d\tilde{\mathbb{P}}}"]
		& (\Omega,\mathcal{F}, \overline{\mathbb{P}}) \arrow[u,"\frac{d\mathbb{Q}}{d\overline{\mathbb{P}}}"] & \text{ergodic HJB (optimal control } \xi^*(x) \text{ in drift of $X$)}
		\end{tikzcd}

		\vspace{10pt}

\noindent  To be able to express the dual value function in terms of an ergodic HJB eigenpair, we have first to switch to a measure $\hat{\mathbb{P}}$ under which the drift of the underlying diffusion factor process $X$ is independent of time $t$ and the time horizon $T$. Under this measure the corresponding HJB equation is ergodic and we can rewrite the multiplicate functional in terms of the associated eigenpair in the sense of Hansen--Scheinkman (even though the functions in the multiplicative functional depend on the time horizon $T$). After that, we can switch back to the original, maturity-dependent drift process. This procedure can be performed as long as all measure changes are well defined, i.e., the corresponding Radon--Nikod\'{y}m derivatives are true martingales (see A\ref{assume:P_hat} and A\ref{assume:Q}), which means that the original optimal control $\hat{\xi}$  is not ``too far from ergodic" optimal control $\xi^*$.
 \end{remark}

The long-term asymptotic behavior of the function $v(\chi,T)$ is given by \[	v(\chi,T)\simeq e^{-\lambda T}\phi(\chi),\] thus in the decomposition in Eq.\eqref{eqn:v_decomposition}, the expectation\[	\mathbb{E}^{\mathbb{Q}}\Bigl[\frac{1}{\phi (X_T)}\, e^{\int_0^Tf(X_s,s;T)\,ds}\Bigr]\] can be understood as an error term. Our derivation of the long-term sensitivity relies mainly on estimations of this error term.

\section{Sensitivity analysis with respect to initial factor}
\label{sec:delta}

This section studies the sensitivity of the optimal expected utility with respect to the initial factor $\chi=X_0.$ Using the dual formulation of Eq.\eqref{eqn:dual_form}, the initial-value sensitivity in Eq.\eqref{eqn:delta} can be expressed as
\begin{equation} 
\label{eqn:delta_dual}
	\frac{\partial}{\partial \chi} \ln\Bigl\vert \sup_{\Pi\in\mathcal{X}}\mathbb{E}^\mathbb{P}\bigl[U(\Pi_T)\bigr]\Bigr\vert = (1-p)\frac{\partial}{\partial \chi} \ln  v(\chi,T)
\end{equation} 
and thus we are interested in the sensitivity $\frac{\partial }{\partial \chi}\ln v(\chi,T)$ for large time $T$. The sensitivity for large time $T$ is described in Theorem \ref{thm:delta}, which states that $\frac{\partial }{\partial \chi}\ln v(\chi,T)$ is asymptotically equal to $\frac{\phi'(\chi)}{\phi(\chi)}.$ The proof is following.

\begin{proof32}  
	The function $\phi$ is continuously differentiable by A\ref{assume:long-term}.
	From Eq.\eqref{eqn:v_decomposition}, applying the chain rule, we obtain the differentiability of $v(\chi,t)$ and   
	\[
			\frac{\frac{\partial}{\partial\chi}v(\chi,T)}{v(\chi,T)}=\frac{\phi'(\chi)}{\phi(\chi)}+
			\frac{\frac{\partial }{\partial \chi}\mathbb{E}^{\mathbb{Q}}\bigl[\frac{1}{\phi(X_T)}\, e^{\int_0^Tf(X_s,s;T)\,ds} \bigr]}{\mathbb{E}^{\mathbb{Q}}\bigl[\frac{1}{\phi(X_T)}\, e^{\int_0^Tf(X_s,s;T)\,ds} \bigr]}.
		\]
	The nominator of the second term goes to zero by assumption, and A\ref{assume:long-term} and Eq.\eqref{eqn:v_decomposition} give  the convergence to a positive constant of the denominator. This completes the proof.
\end{proof32}

In order to utilize Theorems \ref{thm:delta} and \ref{thm:decompose}, we have to provide sufficient conditions under which the  mapping $ \chi \mapsto \mathbb{E}^{\mathbb{Q}}\bigl[\frac{1}{\phi (X_T)}\, e^{\int_0^Tf(X_s,s;T)\,ds}\bigr]$ is continuously differentiable and its derivative converges to zero as $T\to\infty.$ To denote the dependence of the solution $X$ of the SDE \eqref{eqn:SDE_X} on the initial value $x$, we write $X^x.$ Assume that for almost all $\omega\in\Omega$ the map $x\mapsto X_t^x$ is continuously differentiable and the derivative process $(Y_t)_{0\le t\le T}:=(\frac{\partial X_t^x}{\partial x})_{0\le t\le T},$ which is called the first variation process, satisfies
\begin{equation}
\label{eqn:1st_var_proc}
	dY_t=\kappa_x(X_t,t;T)Y_t \,dt+\sigma_{1}'(X_t)Y_t\,dB_{1,t}+\sigma_{2}'(X_t)Y_t\,dB_{2,t},  \qquad Y_0=1.
\end{equation}  
This holds, as a particular case, if the derivative of $\kappa( \cdot ,t;T)$ is jointly continuous in $x$ and $t$ for fixed $T$ and $\sigma_1$ and $\sigma_2$ are continuously differentiable with bounded derivatives (for details, see  \cite[ Theorem V.39]{protter2005stochastic}).

\begin{prop}
\label{prop:delta}
	Additionally to A\ref{assume:SDE_X} -- \ref{assume:Q}, assume that  for almost all $\omega\in\Omega$ the map $x\mapsto X_t^x$ is continuously differentiable and the first variation process $(Y_t)_{0\le t\le T}$ satisfies Eq.\eqref{eqn:1st_var_proc} and $f$ is continuously differentiable. Suppose that  there exist an open neighborhood $I_\chi$ of $\chi$ and  positive constants $u,v,w$ with  $\frac{1}{u}+\frac{1}{v}+\frac{1}{w}=1$     satisfying the following conditions.
	\begin{itemize}
		\item[(i)]  As a function of two variables $(x,T),$ the expectation
		\[
					\Gamma_u(x,T)=\mathbb{E}^{\mathbb{Q}}\Bigl[\frac{1}{\phi^{u}({X}_T)}\, e^{u\int_0^Tf(X_s,s;T)\,ds} \,\Big\vert \, X_0=x\Bigr]
				\]
		is uniformly bounded on $I_\chi\times(0,\infty).$		
		\item[(ii)] As a function of two variables $(x,T),$ the expectation
		\[
					\mathbb{E}^{\mathbb{Q}} \biggl\vert \frac{\phi'(X_T)}{\phi(X_T)}\biggr\vert^{v}
				\]
		is uniformly  bounded on $I_\chi\times(0,\infty).$
		\item[(iii)]  As a function of two variables $(x,T),$ the expectation $\mathbb{E}^{\mathbb{Q}} \vert Y_{T;T}\vert^{w}$ is uniformly bounded on $I_\chi$ for each $T$ and  converges to zero as $T\to\infty$ for each $x\in I_\chi.$  
		\item[(iv)] The expectation
		\[
					\mathbb{E}^{\mathbb{Q}}\biggl[\biggl(\int_0^T \vert f_x(X_s,s;T)Y_{s;T} \vert \,ds\biggr)^{m}\biggr]
				\]
		is uniformly bounded on $I_\chi$ for each $T$ and  converges to zero as $T\to\infty$ for each $x\in I_\chi.$ Here, $m=\frac{u}{u-1},$ i.e.,  $\frac{1}{u}+\frac{1}{m}=1.$ 
	\end{itemize}	
	Then the map $x\mapsto\mathbb{E}^{\mathbb{Q}}\bigl[\frac{1}{\phi(X_T)}\, e^{\int_0^Tf(X_s,s;T)\,ds} \, \big\vert \, X_0=x\bigr]$ is continuously differentiable in $x$ on $I_\chi,$ and
	\[
			\frac{\partial }{\partial x}\mathbb{E}^{\mathbb{Q}}\Bigl[\frac{1}{\phi(X_T)}\, e^{\int_0^Tf(X_s,s;T)\,ds} \, \Big\vert \, X_0=x\Bigr]
		\]
	converges to zero as $T\to\infty.$
\end{prop}

\begin{proof}
	First we observe that
	\[
			\frac{\partial }{\partial x}\mathbb{E}^{\mathbb{Q}}\Bigl(\frac{1}{\phi(X_T)}\, e^{\int_0^Tf(X_s,s;T)\,ds} \Bigr)  = \mathbb{E}^{\mathbb{Q}}\Bigl(\frac{1}{\phi(X_T)}\, e^{\int_0^Tf(X_s,s;T)\,ds}\int_0^Tf_x(X_s,s;T)Y_s\,ds -\frac{\phi'(X_T)}{\phi^2(X_T)}\, e^{\int_0^Tf(X_s,s;T)\,ds} Y_T\Bigr)
		\]
	holds and the derivative is a continuous function of $x.$ This equality can be obtained by 
	interchanging the derivative and the expectation, and this is justified  since
	\begin{align*} 
			&\,\mathbb{E}^{\mathbb{Q}}\biggl\vert \Bigl(\frac{1}{\phi(X_T)}\Bigr)'\, e^{\int_0^Tf(X_s,s;T)\,ds} Y_T+\frac{1}{\phi(X_T)}\, e^{\int_0^Tf(X_s,s;T)\,ds}\int_0^Tf_x(X_s,s;T)Y_s\,ds\biggr\vert \\
			\leq&\, \mathbb{E}^{\mathbb{Q}}\biggl\vert \frac{\phi'(X_t)}{\phi^2(X_t)}\, e^{\int_0^Tf(X_s,s;T)\,ds} Y_T\bigr\vert +\mathbb{E}^{\mathbb{Q}}\biggl\vert \frac{1}{\phi(X_T)}\, e^{\int_0^Tf(X_s,s;T)\,ds}\int_0^Tf_x(X_s,s;T)Y_s\,ds\biggr\vert \\
			\leq&\, \Gamma_u(x,T)^{\frac{1}{u}}\left(\mathbb{E}^{\mathbb{Q}}\biggl\vert \frac{\phi'(X_t)}{\phi(X_t)}\biggr\vert^{v}\right)^{\frac{1}{v}}(\mathbb{E}^{\mathbb{Q}} \vert Y_T \vert^{w})^{\frac{1}{w}}+\Gamma_u(x,T)^{\frac{1}{u}}\biggl(\mathbb{E}^{\mathbb{Q}}\Bigl(\int_0^T \vert f_x(X_s,s;T)Y_s \vert \,ds\Bigr)^{m}\biggr)^{\frac{1}{m}}  	 	
		\end{align*}
	is uniformly bounded on $I_\chi$ by (i)-(iv).
	Moreover, the same inequality gives that the derivative goes to zero as $T\to\infty.$ 
\end{proof}

\section{Sensitivity analysis with respect to drift and volatility}
\label{sec:drift_vol}

This section studies the sensitivities with respect to the drift and volatility perturbations. The arguments in this section is similar to \cite{park2015sensitivity}.

\subsection{Parameter perturbations}
\label{sec:para_perturb}

We provide a precise meaning of the perturbed  the drift and volatility functions.

\begin{Bassume}	\label{bassume:perturb}
	Let  $m_\epsilon,$ $\sigma_{1,\epsilon},$ $\sigma_{2,\epsilon},$ $b_\epsilon,$ $\varsigma_\epsilon$ be continuous functions in the  variables $(\epsilon,x)\in I\times\mathbb{R}$ for a neighborhood $I$ of $0$ such that they are continuously differentiable in $\epsilon$ on $I$ and $m_0=m,$ $\sigma_{1,0}=\sigma_1,$ $\sigma_{2,0}=\sigma_2,$ $b_0=b,$ $\varsigma_0=\varsigma.$  
\end{Bassume}

\begin{Bassume}\label{bassume:HS_eps}
	For each $\epsilon\in I,$ the functions $m_\epsilon,$ $\sigma_{1,\epsilon},$ $\sigma_{2,\epsilon},$ $b_\epsilon,$ $\varsigma_\epsilon$ satisfy A\ref{assume:SDE_X} -- \ref{assume:Q}. The domain $(\ell_\epsilon,r_e\samepage)$ in A\ref{assume:SDE_X} of the process  $X^\epsilon$ may depend on $\epsilon,$ and the constant $C$ in A\ref{assume:long-term} can also depend on $\epsilon.$ 
\end{Bassume}

From theses assumptions, we can construct the following objects. Let $X^\epsilon$ be the solution of the SDE
\[
	dX_t^\epsilon=m_\epsilon(X_t^\epsilon)\,dt+\sigma_{1,\epsilon}(X_t^\epsilon)\,dW_{1,t}+\sigma_{2,\epsilon}(X_t^\epsilon)\,dW_{2,t},\qquad X_0^\epsilon=\chi
\]
with perturbation parameter $\epsilon.$ The initial value $\chi$ is not perturbed. We denote by $\mathcal{X}^\epsilon$ the family of  wealth processes of admissible portfolios in the perturbed market model. Define
\[
	v_\epsilon(\chi,T) :=\mathbb{E}^\mathbb{P}\bigl[(\hat{Y}_T^{\epsilon})^q\bigr]=\mathbb{E}^\mathbb{P}\bigl[(\hat{Y}_T^{\epsilon})^q \, \big\vert \, X_0=\chi\bigr]
\]
where $\hat{Y}_T^{\epsilon}$ is the optimizer of the dual problem in the perturbed market. We are interested in the sensitivity
\begin{equation}
\label{eqn:sensitivity_lambda}
	\frac{\partial}{\partial \epsilon}\Big\vert_{\epsilon=0} \ln\Bigl\vert \sup_{\Pi\in\mathcal{X}^\epsilon}\mathbb{E}^\mathbb{P}\bigl[U(\Pi_T)\bigr]\Bigr\vert
\end{equation}
for large time $T.$ From the dual formulation in Eq.\eqref{eqn:dual_form}, we know that the long-term sensitivity can be obtained by evaluating
\begin{equation}
	\frac{\partial}{\partial \epsilon}\Big\vert_{\epsilon=0} \ln v_\epsilon(\chi,T).
\end{equation}   

We can transform this sensitivity   into a simpler form similar to Eq.\eqref{eqn:v_decomposition} by using an exponential change of measure. Then    
\begin{equation}\label{eqn:v_eps_decomposition}
	v_\epsilon(\chi,T)=e^{-\lambda_\epsilon T}\phi_\epsilon(\chi)\,\mathbb{E}^{\mathbb{Q}_\epsilon}\Bigl[\frac{1}{\phi_\epsilon(X_T^\epsilon)}\, e^{\int_0^Tf_\epsilon(X_s^\epsilon,s;T)\,ds} \Bigr]
\end{equation}
and the $\mathbb{Q}^\epsilon$-dynamics of $X$ is
\begin{equation}
\label{eqn:dX_eps}
	dX_t^\epsilon = \kappa_\epsilon(t,X_t^\epsilon;T) \,dt+\sigma_{1,\epsilon}(X_t^\epsilon)\,dB_{1,t}^\epsilon+\sigma_{2,\epsilon}(X_t^\epsilon)\,dB_{2,t}^\epsilon,  \qquad  0\leq t\leq T 
\end{equation}
for a two-dimensional $\mathbb{Q}^\epsilon$-Brownian motion $(B_{1,t}^\epsilon,B_{2,t}^\epsilon)_{t\ge0}.$ Here, the functions $f_\epsilon$ and $\kappa_\epsilon$ are defined as
\begin{align*}
	& f_\epsilon(x,t;T):=-\frac{q}{2}(1-q)\bigl(\xi_\epsilon^*(x)-\hat{\xi}_\epsilon(x,t;T)\bigr)^2\\ 
	&\kappa_\epsilon(x,t;T):=m_\epsilon(x)-q\theta_\epsilon(x)\sigma_{1,\epsilon}(x)-q\hat{\xi}_\epsilon(x,t;T)\sigma_{2,\epsilon}(x)+\frac{\phi_\epsilon'(x)}{\phi_\epsilon(x)}\bigl(\sigma_{1,\epsilon}^2(x)+\sigma_{2,\epsilon}^2(x)\bigr) 
\end{align*} 
where $\theta_\epsilon,$ $\hat{\xi}_\epsilon,$ $\xi_\epsilon^*,$ $\phi_{\epsilon},$ are functions defined as in Eq.\eqref{eqn:MPR}, Eq.\eqref{eqn:opt_xi}, Eq.\eqref{eqn:xi_star}, A\ref{assume:long-term}, respectively for the perturbed market. We use the prime notation to denote the derivative with respect to $x.$

For the sensitivity analysis, we assume the following regularity conditions. We want to separate the perturbation effects of the underlying diffusion process and the functionals applied to it. Therefore, we define
\[
	w_{\eta,\epsilon}(\chi,T):=\mathbb{E}^{\mathbb{Q}_\epsilon}\Bigl[\frac{1}{\phi_\eta(X_T^\epsilon)}\, e^{\int_0^Tf_\eta(X_s^\epsilon,s;T)\,ds} \Bigr]
\]
so that $v_\epsilon(\chi,T)=e^{-\lambda_\epsilon T}\phi_\epsilon(\chi)w_{\epsilon,\epsilon}(\chi,T).$ We call this function $w$ the error term.

\begin{thm}
\label{thm:total_chain}
	Additionally to B\ref{bassume:perturb} -- \ref{bassume:HS_eps}, we assume  the following conditions:
	\begin{enumerate}[(i)]
		\item 	The two functions $\epsilon\mapsto\lambda_\epsilon$ and $\epsilon\mapsto\phi_{\epsilon}(\chi)$ are continuously differentiable on $I.$
		\item 	The partial derivative $\frac{\partial }{\partial \eta}w_{\eta,\epsilon}(\chi,T)$ exists and is continuous on $I^2.$ Moreover,
		\[
				\lim_{T\rightarrow\infty}\frac{1}{T}\frac{\partial}{\partial\eta}\Big\vert_{\eta=0}w_{\eta,0}(\chi,T)=0.
				\]
		\item The partial derivative $\frac{\partial }{\partial \epsilon}w_{\eta,\epsilon}(\chi,T)$ exists and is continuous on $I^2.$  Moreover,  
		\[
				\lim_{T\rightarrow \infty}\frac{1}{T}\frac{\partial}{\partial\epsilon}\Big\vert_{\epsilon=0}w_{0,\epsilon}(\chi,T)=0.
				\]
	\end{enumerate}
	Then the perturbed function $\ln v_\epsilon(\chi,T)$ is differentiable at $\epsilon=0$ and
	\begin{align}
		\label{eqn:chain_rule}
			& \phantom{==}\frac{1}{T}\frac{\partial}{\partial\epsilon}\Big\vert_{\epsilon=0}\ln v_\epsilon(\chi,T)\\
			& =  -\frac{\partial \lambda_\epsilon}{\partial\epsilon}\Big\vert_{\epsilon=0} +\frac{\frac{\partial}{\partial\epsilon}\big\vert_{\epsilon=0}\phi_\epsilon(\chi)}{T\,\phi(\chi)} +\frac{\frac{\partial}{\partial\epsilon}\big\vert_{\epsilon=0}\mathbb{E}^{\mathbb{Q}} \bigl[\frac{1}{\phi_\epsilon(X_T)}\, e^{\int_0^Tf_\epsilon(X_s,s;T)\,ds}\bigr]}{T\,\mathbb{E}^{\mathbb{Q}} \bigl[\frac{1}{\phi(X_T)}\, e^{\int_0^Tf(X_s,s;T)\,ds}\bigr]} +\frac{\frac{\partial}{\partial\epsilon}\big\vert_{\epsilon=0}\mathbb{E}^{\mathbb{Q}_\epsilon} \bigl[\frac{1}{\phi(X_T^\epsilon)}\, e^{\int_0^Tf(X_s^\epsilon,s;T)\,ds} \bigr]}{T\,\mathbb{E}^{\mathbb{Q}}\bigl[\frac{1}{\phi(X_T)}\, e^{\int_0^Tf(X_s,s;T)\,ds}\bigr]}. \nonumber
		\end{align}
	Furthermore, 
	\begin{equation}
		\label{eqn:final_eqn}
			\lim_{T\rightarrow\infty}\frac{1}{T}\frac{\partial}{\partial\epsilon}\Big\vert_{\epsilon=0}\ln v_\epsilon(\chi,T) = - \frac{\partial\lambda_\epsilon }{\partial\epsilon}\Big\vert_{\epsilon=0}. 
		\end{equation}
\end{thm}

\begin{proof}
	Define a function $V$ on $I^4$ by
	\[
			V(\epsilon_1,\epsilon_2,\epsilon_3,\epsilon_4)
			:=  \,e^{-\lambda_{\epsilon_1} T}\phi_{\epsilon_2}(\chi)\,\mathbb{E}^{\mathbb{Q}_{\epsilon_4}}\Bigl[\frac{1}{\phi_{\epsilon_3}(X_T^{\epsilon_4})}\, e^{\int_0^Tf_{\epsilon_3}(X_s^{\epsilon_4},s;T)\,ds} \Bigr] 
			=  \,e^{-\lambda_{\epsilon_1} T}\phi_{\epsilon_2}(\chi)w_{\epsilon_3,\epsilon_4}(\chi,T) 
		\]
	then $v_\epsilon(\chi,T)=V(\epsilon,\epsilon,\epsilon,\epsilon).$ The chain rule gives the differentiability of $\ln v_\epsilon(\chi,T)$ at $\epsilon=0$ and allows us to write the derivative as in Eq.\eqref{eqn:chain_rule}. Because $\mathbb{E}^{\mathbb{Q}}(\frac{1}{\phi(X_T)}\, e^{\int_0^Tf(X_s,s;T)\,ds})$ converges to a  positive constant as $T\to\infty$ by A\ref{assume:long-term} and Eq.\eqref{eqn:v_eps_decomposition}, we obtain Eq.\eqref{eqn:final_eqn} from conditions (i) -- (iii) and Eq.\eqref{eqn:chain_rule}.
\end{proof}

Let us discuss conditions (i) -- (iii) in Theorem \ref{thm:total_chain} given above in more detail. Condition (i) is satisfied for many financially meaningful models. Condition (ii) is easy to check because the continuous differentiability of
\begin{equation}
\label{eqn:derivative}
	\mathbb{E}^{\mathbb{Q}} \Bigl[\frac{1}{\phi_\epsilon(X_T)}\, e^{\int_0^Tf_\epsilon(X_s,s;T)\,ds}\Bigr] 
\end{equation}
is a standard problem of differentiation and integration. An easier to check condition that is sufficient to imply condition (ii) and is used in the calculation of the examples of Section \ref{sec:examples} will be given in Appendix \ref{app:condi_2}. Conditions (i) and (ii) can be checked case-by-case, thus we do not go into further details of the first three terms of Eq.\eqref{eqn:chain_rule}. However, condition (iii) is involved as it concerns the perturbation in the underlying process $X^\epsilon$ and the measure $\mathbb{Q}^\epsilon,$ which are not trivial to analyze. We will provide  a sufficient condition such that condition (iii) holds true  in Theorems \ref{thm:rho} and \ref{thm:vega_vari}. 

For the analysis of these parameter sensitivities, the following expression for the $\mathbb{Q}^\epsilon$-dynamics of $X$ is useful. Let 
\[
	\sigma_\epsilon(\cdot):=\sqrt{\sigma_{1,\epsilon}^2(\cdot)+\sigma_{2,\epsilon}^2(\cdot)}, \qquad \sigma(\cdot):=\sigma_0(\cdot)
\]
and define a new process $B^\epsilon=(B_t^\epsilon)_{t\ge0}$ by
\[
	dB_t^\epsilon=\frac{\sigma_{1,\epsilon}(X_t^\epsilon)}{\sigma_\epsilon(X_t^\epsilon)}\,dB_{1,t}^\epsilon+\frac{\sigma_{2,\epsilon}(X_t^\epsilon)}{\sigma_\epsilon(X_t^\epsilon)}\,dB_{2,t}^\epsilon, \qquad B_0^\epsilon=0,
\]
then $B^\epsilon$ is a $\mathbb{Q}^\epsilon$-Brownian motion as can be seen by L\'{e}vy's characterization. The $\mathbb{Q}^\epsilon$-dynamics of $X$ can then be written as
\[
	dX_t^\epsilon =\kappa_\epsilon(X_t^\epsilon,t;T) \,dt+\sigma_\epsilon(X_t^\epsilon) \, dB_t^\epsilon.
\]

\begin{remark}
	If we consider the problem of the sensitivity of the expected utility stemming from optimizing the long term growth rate, i.e.,
	\[
	\frac{\partial}{\partial\epsilon}\Big\vert_{\epsilon=0}\inf_{\Pi\in\mathcal{X}^\epsilon}\lim_{T\rightarrow \infty}\frac{1}{T}\ln\Bigl\vert \mathbb{E}^\mathbb{P}\bigl[U(\Pi_T)\bigr]\Bigl\vert,
	\] 
	actually all the results  in Section \ref{sec:drift_vol} hold true, only with less assumptions. Following the discussion at the end of Section \ref{sec:dual}, in this case the optimal value can be expressed using the function $v$ in Eq.\eqref{eqn:v_eps} only with $\xi^*$ given in Eq.\eqref{eqn:xi_star} instead of $\hat{\xi}$. In this case we are already in an ergodic regime and no additional change of measure is needed. Thus it it is sufficient to require Assumptions A\ref{assume:SDE_X} -- A\ref{assume:P} as well as A\ref{assume:M} 	for each $\epsilon\in I$  where the two-dimensional Brownian motion $\tilde{W}$ is replaced by $\hat{W}$. Refer to \cite{fleming2002risk} for details.
\end{remark}

\subsection{Drift perturbation of the factor process}
\label{sec:rho}

In this section, we conduct a sensitivity analysis with respect to  the perturbations of $m_\epsilon,$  $b_\epsilon,$ $\varsigma_\epsilon,$ but assume that the volatility functions $\sigma_{1,\epsilon}=\sigma_{1},$ $\sigma_{2,\epsilon}=\sigma_2$ are not perturbed. Under the measure $\mathbb{Q}^\epsilon,$ the perturbed process $X^\epsilon$  has the form
\[
	dX_t^\epsilon =\kappa_\epsilon(X_t^\epsilon,t;T) \,dt+\sigma(X_t^\epsilon)\,dB_t^\epsilon 
\]
so that only the drift term is perturbed. Our goal is to analyze
\[
	\frac{\partial }{\partial \epsilon}w_{\eta,\epsilon}(\chi,T) =\frac{\partial}{\partial\epsilon}\mathbb{E}^{\mathbb{Q}_\epsilon}\Bigl[\frac{1}{\phi_\eta(X_T^\epsilon)}\, e^{\int_0^Tf_\eta(X_s^\epsilon,s;T)\,ds} \Bigr]
\]
under this drift perturbation.
 
Assuming that $\kappa_{\epsilon}$  is continuously differentiable in  $\epsilon$ on $I,$ define
\begin{align}
\label{eqn:hats}
	& \hat{\phi}(\cdot):=\inf_{\epsilon\in I} \phi_{\epsilon}(\cdot)\\
	& \hat{f}(\cdot,t;T):=\sup_{\epsilon\in I}f_\epsilon(\cdot,t;T)\\
	& \hat{g}(\cdot,t;T):=\sup_{\epsilon\in I}\Bigl\vert \frac{1}{\sigma(x)}\frac{\partial}{\partial\epsilon}\kappa_\epsilon(\cdot,t;T)\Bigr\vert.
\end{align}
We consider the following  boundedness assumptions;   $\hat{\phi}(\cdot)>0,$ $\hat{f}(\cdot,t;T)<\infty$ and $\hat{g}(\cdot,t;T)<\infty.$  
If the domain in $(\ell_\epsilon,r_e\samepage)$ in B\ref{bassume:HS_eps} does not depend on $\epsilon,$ then the three functions always satisfy these boundedness condition
by replacing the interval $I$ by a smaller interval if necessary.

\begin{thm}
\label{thm:rho}
	Additionally to B\ref{bassume:perturb} -- \ref{bassume:HS_eps},
	assume that  $\hat{\phi}(\cdot)>0,$ $\hat{f}(\cdot,t;T)<\infty,$ $\hat{g}(\cdot,t;T)<\infty$ and that  $\kappa_{\epsilon}$ is continuously differentiable and $f_\epsilon$ is continuous in $\epsilon$ on $I.$	Suppose the following conditions.
	\begin{itemize} 
		\item[(i)] For each $T\ge0,$ there exists a real number $\epsilon_0=\epsilon_0(T)>0$ such that
		\begin{equation}
				\label{eqn:rho_expo}
					\mathbb{E}^{\mathbb{Q}}\Bigl[e^{\epsilon_0\int_0^T \hat{g}^2(X_s,s;T)\,ds}\Bigr]
				\end{equation}
		is finite.
		\item[(ii)] There exist a real number $v\ge2$ and a  function $h$ with $\lim_{T\to\infty} h(T)=0$  such that for all $T>0$
	\begin{equation} 
\mathbb{E}^{\mathbb{Q}} \left[\Bigl(\int_0^T\hat{g}^{2}(X_s,s;T) \, ds  \Bigr)^{v/2}\right]\leq T^{v}h(T).
\end{equation}   
		\item[(iii)]For each $T\ge0,$ there is a real number $\epsilon_1>0$ such that
		\begin{equation}
				\label{eqn:rho_linear}
					\mathbb{E}^{\mathbb{Q}}\left[ \int_0^T\hat{g}^{v+\epsilon_1}(X_s,s;T) \, ds\right]  
				\end{equation}
		is finite.
		\item[(iv)] The function
		\begin{equation}
				\label{eqn:Gamma}
					\hat{\Gamma}_u(T):=\mathbb{E}^{\mathbb{Q}}\Bigl[\frac{1}{\hat{\phi}^{u}( {X}_T)}\, e^{u\int_0^T\hat{f}(X_s,s;T)\,ds}\Bigr]
				\end{equation}
		is uniformly bounded in $T\ge0$ where $u=\frac{v}{v-1},$ i.e., $\frac{1}{u}+\frac{1}{v}=1,$ for $v$ from (ii). 
	\end{itemize}
	Then, for given $(\chi,T),$ the partial derivative 
	\[
			\frac{\partial }{\partial \epsilon}w_{\eta,\epsilon}(\chi,T)=\frac{\partial }{\partial \epsilon} \mathbb{E}^{\mathbb{Q}_\epsilon}\Bigl[\frac{1}{\phi_\eta(X_T^\epsilon)}\, e^{\int_0^Tf_\eta(X_s^\epsilon,s;T)\,ds} \Bigr]
		\]
	exists and is continuous in $(\eta,\epsilon)$ on $I^2.$  Moreover, for given $\chi,$
	\[
			\frac{1}{T}\frac{\partial}{\partial\epsilon}\Big\vert_{\epsilon=0}w_{0,\epsilon}(\chi,T)=\frac{1}{T}\frac{\partial}{\partial\epsilon}\Big\vert_{\epsilon=0}\mathbb{E}^{\mathbb{Q}_\epsilon} \Bigl[\frac{1}{\phi(X_T^\epsilon)}\, e^{\int_0^Tf(X_s^\epsilon,s;T)\,ds}\Bigr]\rightarrow 0\quad\textnormal{ as }\; T\rightarrow\infty.
		\]
\end{thm}
\noindent The proof of the above theorem  is similar to the proof of Proposition A.1 in \cite{park2015sensitivity}, but for the sake of completeness we provide the proof in Appendix \ref{app:pf_rho}.  
 
\begin{remark}
	One can relax the assumption in the above theorem on the continuous differentiability of $\kappa_{\epsilon}$ by replacing it with local Lipschitz continuity and defining
	\[
			\hat{g}(\cdot,t;T):=\sup_{\epsilon\in I}\Bigl\vert \frac{\kappa_\epsilon(x,t;T)-\kappa(x,t;T)}{\epsilon\sigma(x)}\Bigr\vert.
		\]
	As this introduces cumbersome additional notations, we do not pursue this in the current paper.
\end{remark}

\subsection{Volatility perturbation of the factor process}
\label{sec:vega}

This section discusses the volatility perturbation of the factor process. Consider B\ref{bassume:perturb} -- \ref{bassume:HS_eps} and the perturbed process  
\[
	dX_t^\epsilon =\kappa_\epsilon(X_t^\epsilon,t;T) \,dt+\sigma_\epsilon(X_t^\epsilon)\,dB_t^\epsilon, \qquad X_0^\epsilon=\chi.
\]
Contrary to the previous section, we allow for an additional perturbation of the volatility of the factor process. As this is a mathematically harder problem, we will need stronger conditions.

The main tool of this section is the Lamperti transformation. We assume that $(\epsilon,x)\mapsto\sigma_\epsilon(x)$ is twice continuously differentiable. Fix any $c\in(r,\ell)$ and define
\[
	\ell_\epsilon(\cdot):=\int_{c}^{\cdot}\frac{1}{\sigma_\epsilon(z)}\,dz, \qquad \ell(\cdot):=\ell_0(\cdot).
\]
As $\sigma_\epsilon$ is positive, the function $\ell_\epsilon$ is invertible. Define   two functions $\Phi_\epsilon,$ $F_\epsilon$ and a process $\check{X}^\epsilon$ by 
\[
	\Phi_\epsilon(\cdot)=\phi_\epsilon\bigl(\ell_\epsilon^{-1}(\cdot)\bigr), \qquad F_\epsilon(\cdot,t;T)=f_\epsilon\bigl(\ell_\epsilon^{-1}(\cdot),t;T\bigr), \qquad \check{X}_t^\epsilon:=\ell_\epsilon(X_t^\epsilon),
\]
and let $\Phi:=\Phi_0,$ $F:=F_0$ and $\check{X}:=\check{X}^0.$ The integral begins with a fixed constant $c$ so that the initial value $\check{X}_0^\epsilon=\int_{c}^{\chi}\frac{1}{\sigma_\epsilon(u)}\,du$ is also perturbed if $\chi\neq c.$ The function $v_\epsilon(x,T)$ we want to analyze can be expressed as
\[
	v_\epsilon(\chi,T) =e^{-\lambda_\epsilon T}\phi_\epsilon(\chi) \, \mathbb{E}^{\mathbb{Q}_\epsilon}\Bigl[\frac{1}{\Phi_\epsilon(\check{X}_T^\epsilon)}\, e^{\int_0^TF_\epsilon(\check{X}_s^\epsilon,s;T)\,ds} \Bigr].
\]
Using the It\^{o} formula, it is easy to show that the $\mathbb{Q}^\epsilon$-dynamics of $\check{X}^\epsilon$ is 
\[
	d\check{X}_t^\epsilon =\gamma(\check{X}_t^\epsilon)\,dt + dB_t^\epsilon, \qquad \check{X}_0^\epsilon=\ell_\epsilon(\chi)
\]
where
\[
	\gamma(\cdot):=\frac{\kappa_\epsilon\bigl(\ell_\epsilon^{-1}(\cdot),t;T\bigr)}{\sigma_\epsilon\bigl(\ell_\epsilon^{-1}(\cdot)\bigr)}-\frac{1}{2}\sigma_\epsilon'\bigl(\ell_\epsilon^{-1}(\cdot)\bigr).
\]
Let $U$ be an open neighborhood of $\ell(\chi)$ and define
\[
	\tilde{w}_{\eta,\epsilon}(\check{x},T) :=\mathbb{E}^{\mathbb{Q}_\epsilon}\Bigl[\frac{1}{\Phi_\eta(\check{X}_T^\epsilon)}\, e^{\int_0^TF_\eta(\check{X}_s^\epsilon,s;T)\,ds} \, \Big\vert \, \check{X}_0^\epsilon=\check{x}\Bigr]
\]
for $(\eta,\epsilon,\check{x},T)\in I\times I\times U\times [0,\infty)$ so that  
\[
	v_\epsilon(\chi,T) =e^{-\lambda_\epsilon T}\phi_\epsilon(\chi)\, \tilde{w}_{\epsilon,\epsilon}\bigl(\ell_\epsilon(\chi),T\bigr).
\]

Under these circumstances, we obtain the following theorem. The proof is similar to that of Theorem \ref{eqn:chain_rule}. 
 
\begin{thm}\label{thm:vega_vari}
	Additionally to B\ref{bassume:perturb} -- \ref{bassume:HS_eps}, assume that  $(\epsilon,x)\mapsto\sigma_\epsilon(x)$ is twice continuously differentiable. Suppose condition (i) in Theorem \ref{thm:total_chain} and the following conditions.
	\begin{enumerate}[(i)]
		\item 	The partial derivative $\frac{\partial }{\partial \check{x}}\tilde{w}_{\eta,\epsilon}(\check{x},T)$ exists and is continuous in $(\eta,\epsilon,\check{x})$ on $I\times I\times U.$ Moreover,
		\[
					\lim_{T\to\infty}\frac{1}{T}\frac{\partial }{\partial \check{x}}\Big\vert_{\check{x}=\ell(\chi)}\tilde{w}_{0,0}(\check{x},T)=0.
				\] 
		\item 	The partial derivative $\frac{\partial }{\partial \eta}\tilde{w}_{\eta,\epsilon}(\check{x},T)$ exists and is continuous in $(\eta,\epsilon,\check{x})$ on $I\times I\times U.$ Moreover,
		\[
					\lim_{T\to\infty}\frac{1}{T}\frac{\partial}{\partial \eta}\Big\vert_{\eta=0}\tilde{w}_{\eta,0}\bigl(\ell(\chi),T\bigr)=0.
				\] 
		\item 	The partial derivative $\frac{\partial }{\partial \epsilon}\tilde{w}_{\eta,\epsilon}(\check{x},T)$  exists and is continuous in $(\eta,\epsilon,\check{x})$ on $I\times I\times U.$ Moreover,
		\[
					\lim_{T\to\infty}\frac{1}{T}\frac{\partial }{\partial \epsilon}\Big\vert_{\epsilon=0}\tilde{w}_{0,\epsilon}\bigl(\ell(\chi),T\bigr)=0.
				\]
	\end{enumerate}
	Then $\tilde{w}_{\eta,\epsilon}(x,T)$ (thus	$\ln v_\epsilon(x,T)$) is differentiable in $\epsilon$ on $I$ and  
	\begin{align}
		\label{eqn:decompo}
			&\,\frac{\partial}{\partial\epsilon}\Big\vert_{\epsilon=0}w_{\epsilon,\epsilon}(\chi,T)=\frac{\partial}{\partial\epsilon}\Big\vert_{\epsilon=0}\tilde{w}_{\epsilon,\epsilon}\bigl(\ell_\epsilon(\chi),T\bigr) \nonumber \\
			= &\,\frac{\partial}{\partial\epsilon}\Big\vert_{\epsilon=0}\ell_{\epsilon}(\chi)\cdot\frac{\partial }{\partial \check{x}}\Big\vert_{\check{x}=\ell(\chi)}\tilde{w}_{0,0}(\check{x},T)+\frac{\partial}{\partial \eta}\Big\vert_{\eta=0}\tilde{w}_{\eta,0}\bigl(\ell(\chi),T\bigr)+\frac{\partial }{\partial \epsilon}\Big\vert_{\epsilon=0}\tilde{w}_{0,\epsilon}\bigl(\ell(\chi),T\bigr).
		\end{align}
	Finally,
	\[
			\lim_{T\rightarrow\infty}\frac{1}{T}\frac{\partial}{\partial\epsilon}\Big\vert_{\epsilon=0}\ln v_\epsilon(\chi,T) =- \frac{\partial\lambda_\epsilon }{\partial\epsilon}\Big\vert_{\epsilon=0}. 
		\]
\end{thm}

This theorem has an important implication, namely that the volatility sensitivity of the error term $w$ is a sum of the initial value sensitivity, the functional sensitivity and the drift sensitivity of the error term. Condition (ii) in the above theorem is about the sensitivity with respect to the functional perturbation, which is corresponding to condition (ii) in Theorem \ref{thm:total_chain}. Condition (iii) in the above theorem  is about the sensitivity with respect to the drift corresponding to condition (iii) in Theorem \ref{thm:total_chain}, which can be analyzed in the same way in Section \ref{sec:rho}. In the special case $c=\chi$ we can omit condition (i) in the above theorem since the initial value is not perturbed. Moreover, Eq.\eqref{eqn:decompo} can be written as
\[
	\frac{\partial}{\partial\epsilon}\Big\vert_{\epsilon=0}w_{\epsilon,\epsilon}(\chi,T)=\frac{\partial}{\partial\epsilon}\Big\vert_{\epsilon=0}\tilde{w}_{\epsilon,\epsilon}(T)=\frac{\partial}{\partial \eta}\Big\vert_{\eta=0}\tilde{w}_{\eta,0}(T)+\frac{\partial }{\partial \epsilon}\Big\vert_{\epsilon=0}\tilde{w}_{0,\epsilon}(T).
\]

\section{Conclusion}
\label{sec:conclusion}

In this paper, we conducted a sensitivity analysis of the long-term expected utility of optimal portfolios in an incomplete market given by a factor model. The main purpose was to find the long-term sensitivity, that is, the extent how much the optimal expected utility is affected in the long run for small changes of the underlying factor model. We calculated two kinds of sensitivities; The first is the initial factor sensitivity. For the initial value $\chi=X_0$ of the factor process, we study
the behavior of
\[
	\frac{\partial}{\partial\chi}\sup_{\Pi\in\mathcal{X}}\mathbb{E}^\mathbb{P}\bigl[U(\Pi_T)\bigr]
\]
for large $T.$ The second kind is the drift and volatility sensitivities. For a perturbation parameter $\epsilon,$ consider a perturbed asset price $S^\epsilon$ with $S=S^0$ and the family $\mathcal{X}^\epsilon$ of  wealth processes   of admissible portfolios with the perturbed asset model $S^\epsilon.$  For the long-term sensitivity, we are interested in the behavior of 
\[
	\frac{\partial}{\partial\epsilon}\Big\vert_{\epsilon=0}\sup_{\Pi\in\mathcal{X}^\epsilon}\mathbb{E}^\mathbb{P}\bigl[U(\Pi_T)\bigr]
\]
for large $T.$

To achieve this, we employed several techniques. The primal utility maximization problem was transformed into the dual problem. Then, we  approximated the solution of the dual problem by an HJB equation.  The long-term behavior of the  optimal expected utility can be characterized by a solution pair $(\lambda,\phi)$ of the corresponding ergodic HJB equation, and  we  demonstrated  that this solution pair  determines the long-term sensitivities. The solution $v$ of the dual problem can be decomposed as
\[
	v(\chi,T)=e^{-\lambda T}\phi(\chi)\,\mathbb{E}^{\mathbb{Q}}\Bigl[\frac{1}{\phi (X_T)}\, e^{\int_0^Tf(X_s,s;T)\,ds} \Bigr].
\]
We regarded the expectation in this expression as an error term and then found sufficient conditions under which this error term is negligible. We provided examples of explicit results for several market models such as the Kim--Omberg model for stochastic excess returns and the Heston stochastic volatility model.

\textbf{Acknowledgement.}\\ 
Hyungbin Park was supported by the National Research Foundation of Korea (NRF) grant funded by the Korea government (MSIT) (No. 2018R1C1B5085491 and No. 2017R1A5A1015626).

\bibliographystyle{plainnat}

\bibliography{utility_max}

\let\normalsize\scriptsize
\appendix
\scriptsize

\sectionfont{\scriptsize}
\section{Motivation for the ergodic HJB equation}
\label{app:conn_HJB}

In this section, we derive the ergodic HJB equation and provide the motivation of A\ref{assume:diff_v} -- \ref{assume:P}. These assumptions originate from  the dynamic programming principle. Let $\mathcal{M}$ be the set of all progressively measurable processes $\xi$ such that $\int_0^t\xi_s^2\,ds<\infty$ a.s. for each $t.$ Then 
\begin{align*}
	v(x,T) & =\sup_{Y\in\mathcal{Y}}\mathbb{E}^\mathbb{P}\bigl[Y_T^q\bigr] = \sup_{\xi\in\mathcal{M}}\mathbb{E}^\mathbb{P} \Bigl[e^{-q\int_0^T\theta(X_s)\,dW_{1,s}-\frac{q}{2} \int_0^T\theta^2(X_s)\,ds - q \int_0^T\xi_s\,dW_{2,s} - \frac{q}{2}\int_0^T\xi_s^2\,ds}\Bigr] = \sup_{\xi\in\mathcal{M}}\mathbb{E}^\mathbb{\hat{P}}\Bigl[e^{ \frac{q}{2}(q-1)\int_0^T(\theta^2(X_s)+\xi_s^2)\,ds}\Bigr]
\end{align*}
where 
\begin{equation}
\label{eqn:hat_P_from_P}
	\frac{d\mathbb{\hat{P}}}{d\mathbb{P}}\Big\vert_{\mathcal{F}_T}=\mathcal{E}\Bigl(-q\int_0^{\cdot}\theta(X_s)\,dW_{1,s}-q\int_0^\cdot\xi_s\,dW_{2,s}\Bigr)_T
\end{equation} 
defines a martingale due to A\ref{assume:P}. The $\mathbb{\hat{P}}$-dynamics of $X$ is 
\[
	dX_t=(m(X_t)-q\theta(X_t)\sigma_1(X_t)-q\xi_t\sigma_2(X_t))\,dt+\sigma_1(X_1)\,d\hat{W}_{1,t}+\sigma_2(X_t)\,d\hat{W}_{2,t}
\]
for a $\mathbb{\hat{P}}$-Brownian motion $(\hat{W}_{1,t},\hat{W}_{2,t}).$ We regard the process $X$ as a state variable and $\xi$ as a control variable. The standard argument of the dynamic programming principle says that the value function
\[
	u(x,t):=\sup_{\xi\in\mathcal{M}}\mathbb{E}_{X_t=x}^\mathbb{\hat{P}}\Bigl[e^{ \int_t^Tl(\xi_s,X_s)\,ds}\Bigr]
\]
satisfies
\begin{equation}
\label{eqn:HJB_pre}
	u_t+\frac{1}{2}(\sigma_1^2(x)+\sigma_2^2(x))u_{xx}+\sup_{\xi\in\mathbb{R}}\{h(\xi,x)u_x+l(\xi,x)u\}=0, \qquad u(x,T)=1. 
\end{equation}
The optimal control of Eq.\eqref{eqn:HJB_pre} is given by
\[
	\hat{\xi}(x,t;T)=-\frac{\sigma_2(x)u_x(x,t)}{(1-q)u(x,t)}.
\]
It is convenient to consider  an initial condition at time $0,$ 
\[
	v(x,t)=\sup_{\xi\in\mathcal{M}}\mathbb{E}_{X_0=x}^\mathbb{\hat{P}}\Bigl[e^{ \int_0^tl(\xi_s,X_s)\,ds}\Bigr].
\]
We know that from the Markov property
\[
	v(x,t)=\sup_{\xi\in\mathcal{M}}\mathbb{E}_{X_0=x}^\mathbb{\hat{P}}\Bigl[e^{ \int_0^tl(\xi_s,X_s)\,ds}\Bigr]
	=\sup_{\xi\in\mathcal{M}}\mathbb{E}_{X_{T-t}=x}^\mathbb{\hat{P}}\Bigl[e^{ \int_{T-t}^{T}l(\xi_s,X_s)\,ds}\Bigr]=u(x,T-t).
\]
The function $v(x,t)$ satisfies 
\begin{equation}
\label{eqn:app_v}
	v_t=\frac{1}{2}(\sigma_1^2(x)+\sigma_2^2(x))v_{xx} +\sup_{\zeta\in\mathbb{R}}\{l(\zeta,x)v+h(\zeta,x)v_x\}, \qquad v(0,x)=1.
\end{equation}
The optimal control of Eq.\eqref{eqn:app_v} is given by
\[
	\hat{\zeta}(x,t;T)=-\frac{\sigma_2(x)v_x(x,t)}{(1-q)v(x,t)}
\] 
and it is clear that
\[
	\hat{\xi}(x,t;T)=\hat{\zeta}(x,T-t;T)=-\frac{\sigma_2(x)v_x(x,T-t)}{(1-q)v(x,T-t)},
\] 
which motivates Assumption \ref{assume:struc_opt}  and Eq.\eqref{eqn:opt_xi}.

The \textit{ergodic  HJB equation} is useful to obtain the growth rate $-\lambda$ and to understand the behavior of the optimal function $\hat{\xi}.$ Heuristically, by taking $v(t,x)=e^{-\lambda t}\phi(x)$ in Eq.\eqref{eqn:HJB}, we have 
\begin{equation} 
	-\lambda\phi(x)=\frac{1}{2}(\sigma_1^2(x)+\sigma_2^2(x))\phi_{xx}+\sup_{\zeta\in\mathbb{R}}\{l(\zeta,x)\phi+h(\zeta,x)\phi_x\}.
\end{equation}  
This is a kind of an eigenvalue/eigenfunction problem. The unknown is a pair $(\lambda,\phi)$ and the solution pair is not unique in general.A\ref{assume:long-term} assumes that a specific solution pair $(\lambda,\phi)$ of this ergodic HJB equation approximates the function $v$ defined in Eq.\eqref{eqn:v}, which is also a solution of the original  HJB equation \eqref{eqn:HJB}. Many authors discuss sufficient conditions for this assumption. Refer to Assumption 4.1 in  \cite{knispel2012asymptotics} and Theorem 3.3 in \cite{fleming1995risk}.

\sectionfont{\scriptsize}
\section{Proof of Theorem \ref{thm:rho}}
\label{app:pf_rho}

Proof of Theorem \ref{thm:rho} relies on the following proposition, whose proof is rather long and tedious.
We recall the functions $\hat{\phi},$ $\hat{f}$ and $\hat{g}$ defined in Eq.\eqref{eqn:hats}.
The proof of this proposition is similar to the proof of Proposition A.1 in \cite{park2015sensitivity}, but for the sake of completeness we provide the proof here.  
	
\begin{prop}
\label{prop:rho}   
	Additionally to B\ref{bassume:perturb} -- \ref{bassume:HS_eps}, assume that $\hat{\phi}(\cdot)>0,$ $\hat{f}(\cdot,t;T)<\infty,$ $\hat{g}(\cdot,t;T)<\infty$ and that $\kappa_\epsilon$  is continuously differentiable and $f_\epsilon$ is continuous in $\epsilon$ on $I.$ Fix $T>0$ and suppose the following conditions.
	\begin{itemize}
		\item[(i)] There exists a real number $\epsilon_0>0$ such that
		\begin{equation}
						\mathbb{E}^{\mathbb{Q}}\Bigl[e^{\epsilon_0\int_0^T \hat{g}^2(X_s,s;T)\,ds}\Bigr]
				\end{equation} 
		is finite.	
		\item[(ii)] There exist real numbers $v\ge2$ and $\epsilon_1>0$ such that  	
		\begin{equation}
						\mathbb{E}^{\mathbb{Q}} \int_0^T\hat{g}^{v+\epsilon_1}(X_s,s;T) \, ds  
				\end{equation}
		is finite.
		\item[(iii)] The function 
		\begin{equation}
					\hat{\Gamma}_u(T):=\mathbb{E}^{\mathbb{Q}}\biggl[\frac{1}{\hat{\phi}^{u}( {X}_T)}\, e^{u\int_0^T\hat{f}(X_s,s;T)\,ds}\biggr]
				\end{equation}
		is finite where $u=\frac{v}{v-1},$ i.e., $\frac{1}{u}+\frac{1}{v}=1,$ for $v$ from (ii).
	\end{itemize}
	Then, for given $(\chi,T),$ the partial derivative $\frac{\partial }{\partial \epsilon}w_{\eta,\epsilon}(\chi,T)$ exists and 
	\begin{equation}
		\label{eqn:rho_deriva}
			\frac{\partial }{\partial \epsilon}w_{\eta,\epsilon}(\chi,T) 
			=\frac{\partial }{\partial \epsilon}\mathbb{E}^{\mathbb{Q}_\epsilon}\biggl[\frac{1}{\phi_\eta(X_T^\epsilon)}\, e^{\int_0^Tf_\eta(X_s^\epsilon,s;T)\,ds} \biggr] =\mathbb{E}^{\mathbb{Q}_\epsilon}\biggl[\frac{1}{\phi_\eta(X_T^\epsilon)}\, e^{\int_0^Tf_\eta( X_s^\epsilon,s;T)\,ds}\int_0^T\overline{\ell}_\epsilon(X_s^\epsilon,s;T) \,dB_s^\epsilon\biggr]
		\end{equation}
	where
	\[
			\overline{\ell}_{\epsilon}(x,t;T):=\frac{1}{\sigma(x)}\frac{\partial}{\partial\epsilon}\kappa_\epsilon(x,t;T).
		\]
	Moreover, the derivative is continuous in $(\eta,\epsilon)$ on $I^2$ for given $(\chi,T).$
\end{prop}

\begin{proof}
	As the proof of this proposition is rather intricate, we split up in several steps. We denote $\overline{\ell}(x,t;T):=\overline{\ell}_{0}(x,t;T).$
	\begin{itemize} 
		\item[(I)] We prove Eq.\eqref{eqn:rho_deriva} for $\epsilon=0,$ that is,
		\begin{equation}
				\label{eqn:zero_eps_deriva}
					\frac{\partial}{\partial\epsilon}\Big\vert_{\epsilon=0}
					\mathbb{E}^{\mathbb{Q}_\epsilon}\biggl[\frac{1}{\phi_\eta(X_T^\epsilon)}\, e^{\int_0^Tf_\eta(X_s^\epsilon,s;T)\,ds} \biggr]=\mathbb{E}^{\mathbb{Q}}\biggl[\frac{1}{\phi_\eta(X_T )}\, e^{\int_0^Tf_\eta( X_s,s;T)\,ds}\int_0^T\overline{\ell}(X_s,s;T) \,dB_s\biggr]
				\end{equation}
		This equality will be proven by the following 4 sub-steps.
		\begin{enumerate}[(a)]
			\item First, we show that
			\[
						\frac{\partial}{\partial\epsilon}\Big\vert_{\epsilon=0}\mathbb{E}^{\mathbb{Q}_\epsilon} \biggl[\frac{1}{\phi_\eta(X_T^\epsilon)}\, e^{\int_0^Tf_\eta(X_s^\epsilon,s;T)\,ds}\biggr] 
						=\lim_{\epsilon\rightarrow0}\mathbb{E}^{\mathbb{Q}}\biggl[\frac{1}{\phi_\eta( {X}_T )}\, e^{\int_0^Tf_\eta(X_s.s;T)\,ds}\int_0^TZ_s^\epsilon\,\ell_\epsilon(X_s,s;T)\, dB_s\biggr]
					\]
			for a function $\ell_\epsilon$ and a positive martingale $Z^\epsilon$ defined below.
			\item We prove that the integral $\int_0^T(\ell_\epsilon(X_s,s;T)-\overline{\ell}(X_s,s;T))\,dB_s$ goes to zero in $L^v$ as $\epsilon\to 0.$
			\item We prove that the integral $\int_0^T(Z_s^\epsilon-1)\ell_\epsilon(X_s,s;T)\,dB_s$ goes to zero in $L^v$ as $\epsilon\to 0.$	
			\item We show that steps (b) and (c) imply			
			\[
						\lim_{\epsilon\rightarrow0}\mathbb{E}^{\mathbb{Q}}\biggl[\frac{1}{\phi_\eta(X_T )}\, e^{\int_0^Tf_\eta(X_s,s;T)\,ds}\int_0^T(Z_s^\epsilon\,\ell_\epsilon(X_s,s;T)- \overline{\ell}(X_s,s;T))\,dB_s\biggr] =0,
					\]
			which gives Eq.\eqref{eqn:zero_eps_deriva}.
		\end{enumerate}
	\item[(II)] Using the result of step (I), we prove Eq.\eqref{eqn:rho_deriva} for arbitrary $\epsilon\in I.$
	\item[(III)]  We prove that the derivative is continuous on $I^2,$ which can be obtained  by showing 
	$H_T^\epsilon$ converges to $H_T$ in $L^v$ as $\epsilon\to 0$ 
	where $H_T^\epsilon$ and $H_T$ are defined in Eq.\eqref{eqn:H}.
	 We conduct the following   sub-steps. 
	\begin{enumerate}[(a)]
 	\item First, show that  
 	\[
		 	\epsilon\int_0^T(\overline{\ell}_\epsilon\ell_\epsilon)(X_s,s;T)\,ds\cdot Z_T^\epsilon\to 0
	 	\]
 	in $L^v$ as $\epsilon\to0.$
 	\item We prove that 
 	\[
		 	\int_0^T\overline{\ell}_\epsilon(X_s,s;T)\,dB_s\cdot Z_T^\epsilon\to\int_0^T\overline{\ell}(X_s,s;T)\,dB_s
	 	\]
 	in $L^v$ as $\epsilon\to 0.$ 
	\end{enumerate}
\end{itemize}

	\noindent \textbf{Step}  (I) -- (a).
	We first show Eq.\eqref{eqn:rho_deriva} at $\epsilon=0.$ Define a function $\ell_\epsilon(x,t;T)$ by
	\[
			\ell_\epsilon(x,t;T)=\left\{\begin{array}{ll}
			\frac{\kappa_\epsilon(x,t;T)-\kappa(x,t;T)}{\epsilon\sigma(x)}&\textnormal{ if } \epsilon\neq 0,\\
			\frac{1}{\sigma(x)}\frac{\partial}{\partial\epsilon}\Big\vert_{\epsilon=0}\kappa_\epsilon(x,t;T) &\textnormal{ if } \epsilon= 0,
			\end{array}\right.
		\]
	so that
	\[
			\kappa_\epsilon(x,t;T)=\kappa(x,t;T)+\epsilon\ell_\epsilon(x,t;T)\sigma(x).
		\]
	From the definition of $\overline{\ell}_\epsilon(x,t;T),$ it is clear that $\overline{\ell}(x,t;T)=\overline{\ell}_{0}(x,t;T)=\ell_{0}(x,t;T).$ By the mean-value theorem, we have that
	\[
			\vert \ell_\epsilon(x,t;T) \vert \leq \hat{g}(x,t;T).
		\]

	For $ \vert \epsilon \vert \leq \epsilon_0/2,$ define
	\[
			Z_T^\epsilon:=\frac{d\mathbb{Q}_\epsilon}{d\mathbb{Q}\,}=\mathcal{E}\biggl(\epsilon\int_0^\cdot\ell_{\epsilon}(X_t,t;T)\, dB_t\biggr)_T,
		\]
	then this local martingale process $(Z_t^\epsilon)_{0\leq t\leq T}$ is a martingale since the Novikov condition is satisfied by condition (i). We then have that
	\[
			\mathbb{E}^{\mathbb{Q}_\epsilon} \biggl[\frac{1}{\phi_\eta(X_T^\epsilon)}\, e^{\int_0^Tf_\eta(X_s^\epsilon,s;T)\,ds}\biggr]
			=\mathbb{E}^\mathbb{Q} \biggl[\frac{1}{\phi_\eta(X_T)}\, e^{\int_0^Tf_\eta(X_s,s;T)\,ds}Z_T^\epsilon\biggr].
		\]
	From the equality
	\[
			\frac{Z_T^\epsilon-1}{\epsilon}=\int_0^TZ_s^\epsilon\,\ell_\epsilon(X_s,s;T)\, dB_s
		\]
	derived by the It\^{o} formula, it follows  that
	\begin{align}
		\label{eqn:rho_app}
			\frac{\partial}{\partial\epsilon}\Big\vert_{\epsilon=0}\mathbb{E}^{\mathbb{Q}_\epsilon} \biggl[\frac{1}{\phi_\eta(X_T^\epsilon)}\, e^{\int_0^Tf_\eta(X_s^\epsilon,s;T)\,ds}\biggr]
			& =\frac{\partial}{\partial\epsilon}\Big\vert_{\epsilon=0}\mathbb{E}^{\mathbb{Q}} \biggl[\frac{1}{\phi_\eta(X_T)}\, e^{\int_0^Tf_\eta(X_s,s;T)\,ds}Z_T^\epsilon\biggr]  = \lim_{\epsilon\rightarrow0}\mathbb{E}^{\mathbb{Q}}\biggl[\frac{1}{\phi_\eta( {X}_T )}\, e^{\int_0^Tf_\eta({X}_s,s;T)\,ds}\,\frac{Z_T^\epsilon-1}{\epsilon}\biggr] \nonumber \\
			& =\lim_{\epsilon\rightarrow0}\mathbb{E}^{\mathbb{Q}}\biggl[\frac{1}{\phi_\eta( {X}_T )}\, e^{\int_0^Tf_\eta(X_s.s;T)\,ds}\int_0^TZ_s^\epsilon\,\ell_\epsilon(X_s,s;T)\, dB_s\biggr].
		\end{align}
	
	\noindent \textbf{Step} (I) -- (b).	We show that the integral $\int_0^T(\ell_\epsilon(X_s,s;T)-\overline{\ell}(X_s,s;T))\,dB_s$ goes to zero in $L^v$ as $\epsilon\to 0.$ By the Burkholder--Davis--Gundy inequality and the Jensen inequality,
	\[
			\mathbb{E}^{\mathbb{Q}}\,\biggl\vert\int_0^T(\ell_\epsilon(X_s,s;T)-\overline{\ell}(X_s,s;T))\,dB_s\biggr\vert^v 
			\leq c_v\,\mathbb{E}^{\mathbb{Q}}\biggl\vert\int_0^T(\ell_\epsilon(X_s,s;T)-\overline{\ell}(X_s,s;T))^2\,ds\biggr\vert^{v/2} \leq c_vT^{\frac{v}{2}-1}\,\mathbb{E}^{\mathbb{Q}}\int_0^T \vert \ell_\epsilon(X_s,s;T)-\overline{\ell}(X_s,s;T) \vert^v\,ds
		\]
	for some positive constant $c_v$ in the Burkholder--Davis--Gundy inequality. Because $\vert \ell_\epsilon-\overline{\ell} \vert^v\leq 2^v\bigl( \vert \ell_\epsilon \vert^v+ \vert \overline{\ell} \vert^v\bigr)\leq 2^{v+1}\hat{g}^v$ and  condition (ii) holds, we can apply the Lebesgue dominated convergence theorem, which implies that 
	\[
			\int_0^T\bigl(\ell_\epsilon(X_s,s;T)-\overline{\ell}(X_s,s;T)\bigr)\,dB_s
		\]
	converges to zero in $L^v$ as $\epsilon\rightarrow 0.$  \setlength{\parskip}{6pt}

	\noindent \textbf{Step} (I) -- (c).	We now show that
	\[
			\int_0^T(Z_s^\epsilon-1)\,\ell_\epsilon(X_s,s;T)\,dB_s
		\]
	converges to zero  in $L^v$ as $\epsilon\rightarrow0.$ Choose a sufficiently large positive number $m$ such that
	\[
			\frac{1}{m}+\frac{1}{1+\frac{\epsilon_1}{v}}<1
		\]
	and $mv$ is a positive  integer where $\epsilon_1$ is given by condition (ii). Remember that $v\geq 2.$ It follows again that
	\begin{align*}
			\mathbb{E}^{\mathbb{Q}}\biggl\vert \int_0^T(Z_s^\epsilon-1)\,\ell_\epsilon(X_s,s;T)\,dB_s\biggr\vert^v
			& \leq  \,c_v\,\mathbb{E}^{\mathbb{Q}}\biggl\vert \int_0^T \vert Z_s^\epsilon-1 \vert^2\,\vert\ell_\epsilon\vert^2(X_s,s;T)\,ds \biggr\vert^{v/2}
			\leq  \,c_vT^{\frac{v}{2}-1}\mathbb{E}^{\mathbb{Q}}\int_0^T \vert Z_s^\epsilon-1\vert^v\, \vert\ell_\epsilon\vert^v(X_s,s;T)\,ds \\
			& \leq  \,c_vT^{\frac{v}{2}-1}\biggl(\mathbb{E}^{\mathbb{Q}}\int_0^T \vert Z_s^\epsilon-1 \vert^{mv}\,ds\biggr)^{\frac{1}{m}}
			\biggl(\mathbb{E}^{\mathbb{Q}}\int_0^T \vert\ell_\epsilon \vert^{v+\epsilon_1}(X_s,s;T)\,ds\biggr)^{\frac{1}{1+\frac{\epsilon_1}{v}}}\\
			& \leq\,c_vT^{\frac{v}{2}-1}\biggl(\mathbb{E}^{\mathbb{Q}}\int_0^T \vert Z_s^\epsilon-1 \vert^{mv}\,ds\biggr)^{\frac{1}{m}}
			\biggl(\mathbb{E}^{\mathbb{Q}}\int_0^T\hat{g}^{v+\epsilon_1}(X_s,s;T)\,ds\biggr)^{\frac{1}{1+\frac{\epsilon_1}{v}}}.
		\end{align*}
	The second term is finite by condition (ii). 

	We now prove that the first expectation converges to zero as $\epsilon\to 0.$ Consider 
	\begin{equation}
		\label{eqn:c}
			(Z^\epsilon_t-1)^{mv}=\sum_{i=0}^{mv}\binom{mv}{i} (-1)^{mv-i}(Z_t^\epsilon)^i.
		\end{equation}
	It is enough to show that $\mathbb{E}^{\mathbb{Q}}\int_0^T (Z_t^{\epsilon})^i\,dt$ converges to $T$ as $\epsilon\rightarrow 0$ for $i=1,2,\cdots,mv,$ because 
	\[
			\mathbb{E}^{\mathbb{Q}}\int_0^T(Z_s^\epsilon-1)^{mv}\,ds=\sum_{i=0}^{mv} \binom{mv}{i} (-1)^{mv-i}\mathbb{E}^{\mathbb{Q}}\int_0^T (Z_s^{\epsilon})^i\,dt \quad \longrightarrow \quad T\sum_{i=0}^{mv}\binom{mv}{i} (-1)^{mv-i}=0.
		\]
	To show this, we apply the Lebesgue dominated convergence theorem to $\mathbb{E}^{\mathbb{Q}}\int_0^T (Z_t^{\epsilon})^i\,dt =\int_0^T \mathbb{E}^{\mathbb{Q}}\bigl[(Z_t^{\epsilon})^i\bigr]\,dt$: we prove that $\mathbb{E}^{ \mathbb{Q}}\bigl[(Z_t^{\epsilon})^i\bigr]$ is uniformly bounded for small $\epsilon$ and $0\leq t\leq T$ and that $\mathbb{E}^{\mathbb{Q}}\bigl[(Z_t^{\epsilon})^i\bigr]$ converges to $1$ as $\epsilon$ goes to zero for fixed $t.$ Observe that
	\begin{align}
		\label{eqn:Z}
			\mathbb{E}^{\mathbb{Q}}\bigl[(Z^{\epsilon}_t)^i\bigr] & =\mathbb{E}^{\mathbb{Q}}\exp\Bigl(i\epsilon\int_0^t\ell_\epsilon(X_s)\,dB_s - \frac{i\epsilon^2}{2}\int_0^t \vert\ell_\epsilon \vert^2(X_s)\,ds\Bigr) \nonumber \\
			& = \mathbb{E}^{\mathbb{Q}}\exp\Bigl(i\epsilon\int_0^t\ell_\epsilon(X_s)\,dB_s - i^2\epsilon^2\int_0^t \vert\ell_\epsilon \vert^2(X_s)\,ds\Bigr) \cdot\exp\Bigl(i(i-1/2)\epsilon^2\int_0^t \vert \ell_\epsilon \vert^2(X_s)\,ds\Bigr) \nonumber \\
			& \leq \left(\mathbb{E}^{\mathbb{Q}}\exp\Bigl(2i\epsilon\int_0^t\ell_\epsilon(X_s)\,dB_s - 2i^2\epsilon^2\int_0^t \vert \ell_\epsilon \vert^2(X_s)\,ds\Bigr)\right)^{\frac{1}{2}}  \cdot\biggl(\mathbb{E}^{\mathbb{Q}}\exp\Bigl(i(2i-1)\epsilon^2\int_0^t \vert \ell_\epsilon \vert^2(X_s)\,ds\Bigr)\biggr)^{\frac{1}{2}} \nonumber \\
			& \leq \biggl(\mathbb{E}^{\mathbb{Q}}\exp\Bigl(i(2i-1)\epsilon^2\int_0^t \vert \ell_\epsilon \vert^2(X_s)\,ds\Bigr)\biggr)^{\frac{1}{2}} \leq\biggl(\mathbb{E}^{\mathbb{Q}}\exp\Bigl(i(2i-1)\epsilon^2\int_0^t\hat{g}^2(X_s)\,ds\Bigr)\biggr)^{\frac{1}{2}} \nonumber \\
			&\leq \biggl(\mathbb{E}^{\mathbb{Q}}\exp\Bigl(\epsilon_0\int_0^T\hat{g}^2(X_s)\,ds\Bigr)\biggr)^{\frac{1}{2}},
		\end{align}
	which is finite  by assumption (i) for small $\epsilon.$ Here, for the second inequality, we used that the positive local martingale 
	\[
			\exp\biggl(2i\epsilon\int_0^t\ell_\epsilon(X_s)\,dB_s - 2i^2\epsilon^2\int_0^t \vert \ell_\epsilon \vert^2(X_s)\,ds\biggr)_{0\le t\le T}
		\]
	is a supermartingale. Thus, for small $\epsilon$ and $0\leq t\leq T,$ the term $\mathbb{E}^{\mathbb{Q}}\bigl[(Z^{\epsilon}_t)^i\bigr]$ is uniformly bounded by 
	$(\mathbb{E}^{\mathbb{Q}}\exp\bigl(\epsilon_0\int_0^T\hat{g}^2(X_s)\,ds)\bigr)^{\frac{1}{2}}.$

	Now we prove that $\mathbb{E}^{\mathbb{Q}}\bigl[(Z^{\epsilon}_t)^i\bigr]$ converges to $1$ as $\epsilon$ goes to zero for fixed $t.$ We will apply the Lebesgue dominated convergent theorem to
	\[
			\exp\Bigl(i(2i-1)\epsilon^2\int_0^t\hat{g}^2(X_s)\,ds\Bigr)
		\]
	as $\epsilon$ goes to zero. Using the last inequality in Eq.\eqref{eqn:Z}, this is dominated by
	\[
			\exp\Bigl(\epsilon_0\int_0^t\hat{g}^2(X_s)\,ds\Bigr),
		\]
	whose expectation is finite, thus we know that
	\[
			\mathbb{E}^{\mathbb{Q}}\exp\Bigl(i(2i-1)\epsilon^2\int_0^t\hat{g}^2(X_s)\,ds\Bigr)
		\]
	converges to $1$ as $\epsilon$ goes to zero. 
	\begin{equation}
		\label{eqn:d}
			1  =\mathbb{E}^{\mathbb{Q}}\Bigl[\liminf_{\epsilon\rightarrow 0}(Z^{\epsilon}_t)^i\Bigr]
			\leq \liminf_{\epsilon\rightarrow 0}\mathbb{E}^{\mathbb{Q}}\bigl[(Z^{\epsilon}_t)^i\bigr]
			\leq \limsup_{\epsilon\rightarrow0}\mathbb{E}^{\mathbb{Q}}\bigl[(Z^{\epsilon}_t)^i\bigr]  \leq\lim_{\epsilon\rightarrow0}\mathbb{E}^{\mathbb{Q}}\exp\Bigl(i(2i-1)\epsilon^2\int_0^t\hat{g}^2(X_s)\,ds\Bigr) =1.
		\end{equation}
	This gives the desired result. \setlength{\parskip}{6pt}	

	\noindent \textbf{Step} (I) -- (d). From Eq.\eqref{eqn:rho_app}, in order to show Eq.\eqref{eqn:zero_eps_deriva}, it suffices to prove that
	\[
			\lim_{\epsilon\rightarrow0}\mathbb{E}^{\mathbb{Q}}\biggl[\frac{1}{\phi_\eta(X_T )}\, e^{\int_0^Tf_\eta(X_s,s;T)\,ds}\int_0^T(Z_s^\epsilon\,\ell_\epsilon(X_s,s;T)- \overline{\ell}(X_s,s;T))\,dB_s\biggr] =0.
		\]
	From the condition (iii) that 
	\[
			\hat{\Gamma}_u(T)=\mathbb{E}^{\mathbb{Q}}\biggl[\frac{1}{\hat{\phi}^{u}( {X}_T )}\, e^{u\int_0^T\hat{f}(X_s,s;T)\,ds}\biggr]
		\]
	is finite for $u$ with   $1/u+1/v=1$, by the H\"older inequality, it is enough to show  
	\[
			\int_0^T(Z_s^\epsilon\,\ell_\epsilon(X_s,s;T)-\overline{\ell}(X_s,s;T))\,dB_s\rightarrow 0
		\]
	in $L^v$ as $\epsilon\rightarrow0.$	Observe that
	\begin{equation}
			\int_0^T(Z_s^\epsilon\,\ell_\epsilon(X_s,s;T)- \overline{\ell}(X_s,s;T))\,dB_s = \int_0^T(Z_s^\epsilon-1)\ell_\epsilon(X_s,s;T)\,dB_s+\int_0^T(\ell_\epsilon(X_s,s;T)-\overline{\ell}(X_s,s;T))\,dB_s. 
		\end{equation}
	Steps (b) and (c) above imply that the two terms on the right-hand side converge to zero as $\epsilon\to 0.$ \setlength{\parskip}{6pt}

	\noindent \textbf{Step} (II). We now prove Eq.\eqref{eqn:rho_deriva} for any $\epsilon\in I.$ Fix $\epsilon\in I$ and choose a small open interval $J$ so that $\epsilon+J\subseteq I.$ We introduce another variable $h$ to rewrite the derivative
	\[
			\frac{\partial}{\partial\epsilon}\mathbb{E}^{\mathbb{Q}_\epsilon} \biggl[\frac{1}{\phi_\eta(X_T^\epsilon)}\, e^{\int_0^Tf_\eta(X_s^\epsilon,s;T)\,ds}\biggr]
			=\frac{\partial}{\partial h}\Big\vert_{h=0} \mathbb{E}^{\mathbb{Q}_{\epsilon+h}} \biggl[\frac{1}{\phi_\eta(X_T^{\epsilon+h})}\, e^{\int_0^Tf_\eta(X_s^{\epsilon+h},s;T)\,ds}\biggr].
		\]
	We can regard $h$ as a perturbation parameter. It is easy to show that the perturbed functions $m_{\epsilon+h},$ $\sigma_{1,{\epsilon+h}},$ $\sigma_{2,{\epsilon+h}},$ $b_{\epsilon+h},$ $v_{\epsilon+h}$ with perturbation parameter  $h\in J$ satisfy the hypothesis of this proposition. For example,
	\[
			\sup_{h\in J}\Bigl\vert\frac{1}{\sigma(x)}\cdot\frac{\partial \kappa_{\epsilon+h}(x)}{\partial h}\Bigr\vert \leq \sup_{\epsilon \in I}\Bigl\vert\frac{1}{\sigma(x)}\cdot\frac{\partial \kappa_{\epsilon}(x)}{\partial \epsilon}\Bigr\vert \leq \hat{g}(x).
		\]
	Thus, by applying step (I) to the perturbation parameter $h$, we have
	\[
			\frac{\partial}{\partial h}\Big\vert_{h=0} \mathbb{E}^{\mathbb{Q}_{\epsilon+h}} \biggl[\frac{1}{\phi_\eta(X_T^{\epsilon+h})}\, e^{\int_0^Tf_\eta(X_s^{\epsilon+h},s;T)\,ds}\biggr]=\mathbb{E}^{\mathbb{Q}_\epsilon}\biggl[\frac{1}{\phi_\eta(X_T^\epsilon)}\, e^{\int_0^Tf_\eta( X_s^\epsilon,s;T)\,ds}\int_0^T\overline{\ell}_\epsilon(X_s^\epsilon,s;T) \,dB_s^\epsilon\biggr], 
		\]
	where
	\[
			\overline{\ell}_\epsilon(x,t;T)=\frac{1}{\sigma(x)}\frac{\partial}{\partial h}\Big\vert_{h=0}\kappa_{\epsilon+h}(x,t;T)=\frac{1}{\sigma(x)}\frac{\partial}{\partial\epsilon}\kappa_\epsilon(x,t;T).
		\]
	This gives  Eq.\eqref{eqn:rho_deriva} for any $\epsilon\in I.$ \setlength{\parskip}{6pt}

	\noindent \textbf{Step} (III). We show that the derivative
	\[
			\frac{\partial }{\partial \epsilon} \mathbb{E}^{\mathbb{Q}_\epsilon}\biggl[\frac{1}{\phi_\eta(X_T^\epsilon)}\, e^{\int_0^Tf_\eta(X_s^\epsilon,s;T)\,ds} \biggr]
		\]
	is jointly continuous in $(\eta,\epsilon)$  on $I^2.$ 	Using the same argument as in Step (II), it suffices 
	to show the continuity  at $(\eta,\epsilon)=(0,0).$ 
	We know that
	\begin{align*}
			\frac{\partial }{\partial \epsilon} \mathbb{E}^{\mathbb{Q}_\epsilon}\biggl[\frac{1}{\phi_\eta(X_T^\epsilon)}\, e^{\int_0^Tf_\eta(X_s^\epsilon,s;T)\,ds} \biggr]
			& = \mathbb{E}^{\mathbb{Q}_\epsilon}\biggl[\frac{1}{\phi_\eta (X_T^\epsilon)}\, e^{\int_0^Tf_\eta	(X_s^\epsilon,s;T)\,ds}\int_0^T\overline{\ell}_\epsilon(X_s^\epsilon,s;T) \,dB_s^\epsilon\biggr]\\
			& = \mathbb{E}^{\mathbb{Q}}\biggl[\frac{1}{\phi_\eta(X_T)}\, e^{\int_0^Tf_\eta( X_s,s;T)\,ds} \Bigl(\int_0^T\overline{\ell}_\epsilon(X_s,s;T)\,dB_s-\epsilon\int_0^T(\overline{\ell}_\epsilon\ell_\epsilon)(X_s,s;T)\,ds\Bigr) Z_T^\epsilon\biggr]
		\end{align*}


	For convenience, we define
	\begin{equation}
		\label{eqn:H}
		  H_T^{\epsilon}:=\Bigl(\int_0^T\overline{\ell}_\epsilon(X_s,s;T)\,dB_s-\epsilon\int_0^T(\overline{\ell}_\epsilon\ell_\epsilon)(X_s,s;T)\,ds\Bigr) Z_T^\epsilon; \qquad \qquad H_T:=H_T^0.			
		\end{equation}
	Thus we want to prove that  as $(\eta,\epsilon)\to(0,0),$
	\[
			\mathbb{E}^{\mathbb{Q}}\biggl[\frac{1}{\phi_\eta(X_T)}\, e^{\int_0^Tf_\eta( X_s,s;T)\,ds}H_T^\epsilon\biggr]\to\mathbb{E}^{\mathbb{Q}}\biggl[\frac{1}{\phi(X_T)}\, e^{\int_0^Tf( X_s,s;T)\,ds}H_T\biggr].
		\]
	Condition (iii) implies by the Lebesgue dominated convergence theorem  thanks to the uniform boundedness of $1/\phi_\eta$ and $f_\eta$ over $\eta\in I$ that 
	\[
			\frac{1}{\phi_\eta(X_T)}\, e^{\int_0^Tf_\eta( X_s,s;T)\,ds}\to \frac{1}{\phi(X_T)}\, e^{\int_0^Tf( X_s,s;T)\,ds}
		\]
	in $L^u$ as $\eta\to 0.$ It suffices to prove that $H_T^\epsilon$ converges to $H_T$ in $L^v$ as $\epsilon\to 0.$ This can be achieved by the following two steps.

	\noindent \textbf{Step} (III) -- (a). We show that  
	\[
			\epsilon\int_0^T(\overline{\ell}_\epsilon\ell_\epsilon)(X_s,s;T)\,ds\cdot Z_T^\epsilon\to 0
		\]
	in $L^v$ as $\epsilon\to0.$ This is obtained from 
	\[
			\mathbb{E}^\mathbb{Q}\biggl\vert \int_0^T(\overline{\ell}_\epsilon\ell_\epsilon)(X_s,s;T)\,ds\cdot Z_T^\epsilon\biggr\vert^v
			 \leq \mathbb{E}^\mathbb{Q}\biggl[\Bigl(\int_0^T\hat{g}^2(X_s,s;T)\,ds\Bigr)^v\cdot (Z_T^{\epsilon})^v\biggr] \leq \biggl(\mathbb{E}^\mathbb{Q}\Bigl\vert \int_0^T\hat{g}^2(X_s,s;T)\,ds\Bigr\vert^{2v}\biggr)^{1/2} \Bigl( \mathbb{E}^\mathbb{Q}\bigl[(Z_T^{\epsilon\,})^{2v}\bigr]\Bigr)^{1/2}.
		\]
	The expectation $\mathbb{E}^\mathbb{Q}\bigl\vert \int_0^T\hat{g}^2(X_s,s;T)\,ds \bigr\vert^{2v}$ on the right-hand side is finite from condition (ii) and the expectation $\mathbb{E}^\mathbb{Q}\bigl[(Z_T^{\epsilon})^2v\bigr]$ is uniformly bounded on $I$ by the constant $\bigl(\mathbb{E}^{\mathbb{Q}}\exp(\epsilon_0\int_0^Tg^2(X_s)\,ds)\bigr)^{\frac{1}{2}}$ using the same argument we used to derive Eq.\eqref{eqn:Z}. \setlength{\parskip}{6pt}

	\noindent \textbf{Step} (III) -- (b). We prove that 
	\[
			\int_0^T\overline{\ell}_\epsilon(X_s,s;T)\,dB_s\cdot Z_T^\epsilon\to\int_0^T\overline{\ell}(X_s,s;T)\,dB_s
		\]
	in $L^v$ as $\epsilon\to 0.$ Choose a sufficiently large positive number $m$ such that 
	\[
		\frac{1}{m}+\frac{1}{1+\frac{\epsilon_1}{v}}<1
		\]
	 and $mv$ is a positive  integer where $\epsilon_1$ is given by condition (ii). It is enough to show that as $\epsilon\to0$
	\begin{equation}
		\label{eqn:a}
			\int_0^T\overline{\ell}_\epsilon(X_s,s;T)\,dB_s\to\int_0^T\overline{\ell}(X_s,s;T)\,dB_s\;\textnormal{ in } L^{v+\epsilon_1} 
		\end{equation}
	and
	\begin{equation}
		\label{eqn:b}
			Z_T^\epsilon\to 1\;\textnormal{ in } L^{mv}.
		\end{equation}
	Eq.\eqref{eqn:a} is obtained from condition (ii). Eq.\eqref{eqn:b} is from Eq.\eqref{eqn:c} and the fact that $\lim_{\epsilon\rightarrow0}\mathbb{E}^{\mathbb{Q}}[(Z_t^{\epsilon})^i]=1$ for $0\leq i\leq mv$ shown in Eq.\eqref{eqn:d}.
\end{proof}

We now shift our attention to Theorem \ref{thm:rho}. The proof is as follows.
\begin{proof73} 
	By Proposition \ref{prop:rho}, it suffices to show that
	\[
			\lim_{T\rightarrow\infty}\frac{1}{T}\,\mathbb{E}^\mathbb{Q}\biggl[\frac{1}{\phi(X_T )}\, e^{\int_0^Tf( X_s,s;T)\,ds}\int_0^T \overline{\ell}(X_s,s;T)\, dB_s\biggr]=0.
		\]
	By the H\"older inequality, the  Burkholder-Davis-Gundy inequality and the Jensen inequality, we know that 
	\begin{align}
			\frac{1}{T} \mathbb{E}^\mathbb{Q}\,\Bigl\vert \frac{1}{\phi(X_T )}\, e^{\int_0^Tf( X_s,s;T)\,ds}\int_0^T \overline{\ell}(X_s,s;T)\, dB_s\Bigr\vert
			& \leq \frac{1}{T}\hat{\Gamma}_u(T)^{\frac{1}{u}} \Bigl(\mathbb{E}^\mathbb{Q}\Bigl\vert \int_0^T \overline{\ell}(X_s,s;T)\,dB_s\Bigr\vert^{v} \Bigr)^{\frac{1}{v}}\leq \frac{c'}{T}\,\hat{\Gamma}_u(T)^{\frac{1}{u}} \Bigl(\mathbb{E}^\mathbb{Q}\Bigl( \int_0^T \overline{\ell}^{2}(X_s,s;T)\, ds\Bigr)^{\frac{v}{2}} \Bigr)^{\frac{1}{v}} \\	
		 	& \leq \frac{c'}{T}\hat{\Gamma}_u(T)^{\frac{1}{u}} \Bigl(\mathbb{E}^\mathbb{Q}\Bigl( \int_0^T \hat{g}^{2}(X_s,s;T)\, ds\Bigr)^{\frac{v}{2}} \Bigr)^{\frac{1}{v}} \leq  c'\hat{\Gamma}_u(T)^{\frac{1}{u}}  h(T)^{\frac{1}{v}}  
		\end{align}
	for the positive constant $c'$ in the  Burkholder-Davis-Gundy inequality. For the last inequality, we used (ii) in Theorem \ref{thm:rho}.
	As $\lim_{T\rightarrow \infty}h(T)=0$ and $\hat{\Gamma}_u(T)$ is uniformly bounded in $T,$  
	we obtain the desired result.
\end{proof73}

\sectionfont{\scriptsize}
\section{A note on condition (ii) in Theorem \ref{thm:total_chain}}
\label{app:condi_2} 

This section discusses a method to analyze the derivative $\frac{\partial }{\partial \eta}w_{\eta,\epsilon}(x,T)$ which is useful to check condition (ii) in Theorem \ref{thm:total_chain}. Appendices \ref{app:KO} and \ref{app:Heston} that discuss specific examples will rely on the following proposition.

\begin{prop}\label{prop:condi_2}
	Assume that $\phi_\eta$ and $f_\eta$ are continuously differentiable in $\eta$ on $I.$ Fix $T>0$ and assume the following conditions;
	\begin{enumerate}[(i)]
		\item There exists a function $g(\,\cdot,\cdot\,;T)$ such that $\int_0^Tg(X_s,s;T)\,ds<\infty$ a.s.  and
		\[
					\Bigl\vert \frac{\partial}{\partial\eta}f_\eta(x,t;T)\Bigr\vert \leq g(x,t;T)
				\]
		for all $\eta\in I,$  $x\in (\ell,r)$ and $0\le t\le T.$  
		\item There exists a random variable $G_T$ such that $\mathbb{E}^\mathbb{Q}[G_T^u]<\infty$ for some $u>1$
		and such that
		\[
					 \Bigl\vert \frac{\partial \phi_\eta}{\partial \eta}\Bigr\vert \frac{1}{\phi_\eta^2(X_T)}\, e^{\int_0^Tf_\eta(X_s,s;T)\,ds}+\frac{1}{\phi_\eta(X_T)}\, e^{\int_0^Tf_\eta(X_s,s;T)\,ds}\int_0^T\Bigl\vert \frac{\partial}{\partial\eta}f_\eta(X_s,s;T)\Bigr\vert \,ds \leq G_T
				\]
		for all $\eta\in I.$  
	\end{enumerate}    
	Then 
	\[
			\frac{\partial }{\partial \eta}w_{\eta,\epsilon}(x,T) =\mathbb{E}^{\mathbb{Q}_\epsilon}\biggl[\frac{\partial }{\partial \eta}\Bigl(\frac{1}{\phi_\eta(X_T^\epsilon)}\, e^{\int_0^Tf_\eta(X_s^\epsilon,s;T)\,ds}\Bigr)\biggr]
		\]
	and $\frac{\partial }{\partial \eta}w_{\eta,\epsilon}(x,T) $ is continuous in $(\eta,\epsilon)$ on $I^2.$ 
\end{prop}

\begin{proof} 
	By direct calculation, it follows that
	\begin{align*}
			\frac{\partial}{\partial\eta}\Bigl(\frac{1}{\phi_\eta(X_T)} e^{\int_0^Tf_\eta(X_s,s;T)\,ds}\Bigr)
			& = \frac{\partial \phi_\eta}{\partial \eta}\frac{1}{\phi_\eta^2(X_T)}\, e^{\int_0^Tf_\eta(X_s,s;T)\,ds}+\frac{1}{\phi_\eta(X_T)}\, e^{\int_0^Tf_\eta(X_s,s;T)\,ds}\frac{\partial}{\partial\eta}\int_0^Tf_\eta(X_s,s;T)\,ds  \\ 
			& =\frac{\partial \phi_\eta}{\partial \eta}\frac{1}{\phi_\eta^2(X_T)}\, e^{\int_0^Tf_\eta(X_s,s;T)\,ds}+\frac{1}{\phi_\eta(X_T)}\, e^{\int_0^Tf_\eta(X_s,s;T)\,ds}\int_0^T\frac{\partial}{\partial\eta}f_\eta(X_s,s;T)\,ds.
		\end{align*} 
	Condition (i) was used for the last equality in order to interchange the differentiation and integration using the Leibniz integral rule.
	Observe that
	\[
			w_{\eta,\epsilon}(x,T) =\mathbb{E}^{\mathbb{Q}_\epsilon}\biggl[\frac{1}{\phi_\eta(X_T^\epsilon)}\, e^{\int_0^Tf_\eta(X_s^\epsilon,s;T)\,ds} \biggr] =\mathbb{E}^{\mathbb{Q}}\biggl[\frac{1}{\phi_\eta(X_T)}\, e^{\int_0^Tf_\eta(X_s,s;T)\,ds} Z_T^\epsilon\biggr].
		\]
	From (ii), the Leibniz integral rule states that $\frac{\partial }{\partial \eta}w_{\eta,\epsilon}(x,T)$ exists and
	\[
			\frac{\partial }{\partial \eta}w_{\eta,\epsilon}(x,T) =\mathbb{E}^{\mathbb{Q}}\biggl[\frac{\partial }{\partial \eta}\Bigl(\frac{1}{\phi_\eta(X_T)}\, e^{\int_0^Tf_\eta(X_s,s;T)\,ds}\Bigr) Z_T^\epsilon\biggr].
		\]
	The continuity on $I^2$ can be proven as follows. Using the same argument as in Step (II) of the proof of Proposition \ref{prop:rho}, it suffices to show continuity at the origin $(\eta,\epsilon)=(0,0).$ Choose a sufficiently large even integer $v$ and a sufficiently small $u>1$ such that $1/u+1/v=1.$  Define
	\[
			A_T^\eta:=  \frac{\partial }{\partial \eta}\Bigl(\frac{1}{\phi_\eta(X_T)}\, e^{\int_0^Tf_\eta(X_s,s;T)\,ds}\Bigr); \qquad \qquad A_T:=A_T^0
		\]
	and we claim that $\mathbb{E}^\mathbb{Q}[A_T^\eta Z_T^\epsilon]\to\mathbb{E}^\mathbb{Q}[A_TZ_T]$ as $(\eta,\epsilon)\to(0,0).$ Using the inequalities 
	\begin{align*}
			\Bigl\vert \mathbb{E}^\mathbb{Q}[A_T^\eta Z_T^\epsilon]-\mathbb{E}^\mathbb{Q}[A_TZ_T]\Bigr\vert
			&\leq \Bigl\vert \mathbb{E}^\mathbb{Q}[A_T^\eta (Z_T^\epsilon-Z_T)]\Bigr\vert + \Bigl\vert\mathbb{E}^\mathbb{Q}[(A_T^\eta-A_T)Z_T]\Bigr\vert \\
			&\leq \Bigl\vert (\mathbb{E}^\mathbb{Q} \bigl\vert A_T^\eta\bigr\vert^u)^{1/u}(\mathbb{E}^\mathbb{Q}\vert Z_T^\epsilon-Z_T\vert^v)^{1/v}\Bigr\vert + \Bigl\vert (\mathbb{E}^\mathbb{Q} \bigl\vert A_T^\eta-A_T\bigr\vert^u)^{1/u}(\mathbb{E}^\mathbb{Q}\vert Z_T\vert^v)^{1/v}\Bigr\vert
		\end{align*}
	and since $\vert A_T^\eta \vert \leq G_T$ and $\mathbb{E}^\mathbb{Q}[G_T^u]<\infty$, it is enough to show that $Z_T^\epsilon\to Z_T$ in $L^v$ as $\epsilon\to0.$ This was proven using Eq.\eqref{eqn:c} and the fact that $\lim_{\epsilon\rightarrow0}\mathbb{E}^{\mathbb{Q}}[(Z_T^{\epsilon})^i]=1$ for $0\leq i\leq v$ which  was shown in Eq.\eqref{eqn:d}. Finally, Girsanov's theorem gives that
	\[
			\frac{\partial }{\partial \eta}w_{\eta,\epsilon}(x,T)  =\mathbb{E}^{\mathbb{Q}}\biggl[\frac{\partial }{\partial \eta}\Bigl(\frac{1}{\phi_\eta(X_T)}\, e^{\int_0^Tf_\eta(X_s,s;T)\,ds}\Bigr) Z_T^\epsilon\biggr]=\mathbb{E}^{\mathbb{Q}_\epsilon}\biggl[\frac{\partial }{\partial \eta}\Bigl(\frac{1}{\phi_\eta(X_T^\epsilon)}\, e^{\int_0^Tf_\eta(X_s^\epsilon,s;T)\,ds}\Bigr)\biggr].
		\]
\end{proof}

\sectionfont{\scriptsize}
\section{The Kim--Omberg model}
\label{app:KO}

This appendix discusses the details of the Kim--Omberg model presented at Section \ref{sec:KO_model}, and shows the assumptions made in the main part of the paper are satisfied in this model. Assumptions A\ref{assume:SDE_X} -- \ref{eqn:NFLVR} are well-known to be satisfied for the Kim--Omberg model. We recall the model in Eq.\eqref{eqn:KO_model} and investigate the corresponding objects 
\[
	v(t,T), \; \hat{\xi}, \; (\lambda,\phi), \; \xi^*, \; f, \;\kappa, \;  \mathbb{Q}.
\]
The function $l(\xi,x)$ and $h(\xi,x)$ in Eq.\eqref{eqn:l_h} are
\[
	l(\xi,x)=-\frac{q}{2}(1-q)\Bigl(\frac{\mu^2x^2}{\varsigma^2}+\xi^2\Bigr), \qquad h(\xi,x)=k\overline{m} - \Bigl(k+\frac{q\mu\sigma_1}{\varsigma}\Bigr)x-q\sigma_2\xi.
\]
The HJB equation \eqref{eqn:HJB}  reads in this case
\[
	v_t=\frac{1}{2}\sigma^2v_{xx}+\sup_{\xi\in\mathbb{R}}\{l(\xi,x)v+h(\xi,x)v_x\} =\frac{1}{2}\sigma^2v_{xx} + \biggl(k\overline{m} - \Bigl(k+\frac{q\mu\sigma_1}{\varsigma}\Bigr)x\biggl)v_x-\frac{q}{2}(1-q)\frac{\mu^2x^2}{\varsigma^2}v+    \frac{q\sigma_2^2}{2(1-q)}\frac{v_x^2}{v}  
\]
with $v(x,0)=1.$ Here, we used that the  supremum  of the above HJB equation  is achieved at
\begin{equation}
\label{eqn:KO_zeta}
	\xi=-\frac{\sigma_2}{1-q}\frac{v_x}{v}.
\end{equation}
The solution to this HJB equation corresponds to the function $v$ in Eq.\eqref{eqn:v} (c.f. \cite[Lemma 3]{battauz2015kim}) and can be expressed as
\[
	v(x,t)=e^{\Lambda(t)-\frac{1}{2}\beta(t)x^2-\gamma(t)x},
\]
where the coefficients solve the following system of differential equations:
\begin{align}
\label{eqn:three}
	&\beta'(t)=-\alpha_2\beta^2(t)-2\alpha_1\beta(t) +q(1-q)\frac{\mu^2}{\varsigma^2}, \qquad \beta(0)=0, \nonumber \\
	&\gamma'(t)=-(\alpha_1+\alpha_2\beta(t))\gamma(t)+\alpha_3\beta(t), \qquad \qquad \, \gamma(0)=0, \nonumber \\
	&\Lambda'(t)=\frac{1}{2}\alpha_2\gamma^2(t)-\alpha_3\gamma(t)-\frac{1}{2}\sigma^2\beta(t,) \qquad \quad \,\,\,\,\Lambda(0)=0.
\end{align}
with
\begin{equation}
\label{eqn:alpha}
	\alpha_1=k+\frac{q\mu\sigma_1}{\varsigma}, \qquad \alpha_2=\sigma_1^2+\frac{\sigma_2^2}{1-q}, \qquad \alpha_3=k\overline{m}, \qquad \alpha_4=\sqrt{\alpha_1^2+q(1-q)\alpha_2\mu^2/\varsigma^2}. 
\end{equation} 
Thus assumption A\ref{assume:diff_v} holds. The first equation is the standard Riccati equation with solution
\begin{equation}
\label{eqn:KO_beta}
	\beta(t) =\frac{q(1-q)\frac{\mu^2}{\varsigma^2}(1-e^{-2\alpha_4t})}{\alpha_4+\alpha_1 +(\alpha_4-\alpha_1)e^{-2\alpha_4t}}.
\end{equation}
Given $\beta,$ the second equation of Eq.\eqref{eqn:three} is a first-order ODE which can be easily solved. The solution is
\begin{equation}
\label{eqn:KO_gamma}
	\gamma(t)=\frac{\alpha_3}{\mu(t)}\int_0^t \beta(s)\mu(s)\,ds
\end{equation} 
where
\[
	\mu(t)=e^{\int_0^t(\alpha_1+\alpha_2\beta(s))\,ds}.
\]
The optimal control in Eq.\eqref{eqn:opt_xi}
\begin{equation}
\label{eqn:xi_OU}
	\hat{\xi}(x,t;T)=-\frac{\sigma_2v_x(T-t,x)}{(1-q)v(T-t,x)} = \frac{\sigma_2 }{1-q}\Bigl(\beta(T-t)x+\gamma(T-t)\Bigr)
\end{equation}
is obtained. With this optimizer, assumption A\ref{assume:struc_opt} is satisfied by \cite[Eq.(26)]{battauz2015kim}.

Now we shift our attention to the ergodic  HJB equation \eqref{eqn:EBE}. Direct calculation shows that 
\[
	\phi(x)=e^{-\frac{1}{2}Bx^2-Cx}
\]
with the coefficients
\begin{equation} 
	B=\frac{\alpha_4-\alpha_1}{\alpha_2}, \qquad \qquad C=\frac{\alpha_3(\alpha_4-\alpha_1)}{\alpha_2\alpha_4},
\end{equation} 
is a solution to the ergodic HJB equation \eqref{eqn:EBE}. It is easy to show that $\beta(t)\to B,$ $\gamma(t)\to C$ and  $\frac{\Lambda(t)}{t}\to -\lambda$ as $t\to\infty,$ thus assumption A\ref{assume:long-term}  holds. The optimal control  $\xi^*$ is given by
\begin{equation}
\label{eqn:xi_star_OU}
	\xi^*(x) =- \frac{\sigma_2 \phi_x(x)}{(1-q)\phi(x)}=\frac{\sigma_2  }{1-q}\bigl(Bx+C\bigr).
\end{equation}

For the rest of this section, we  show that assumptions A\ref{assume:P} -- A\ref{assume:Q} are satisfied.

\begin{prop}  \label{prop:KO_A7}
	For the Kim--Omberg model  A\ref{assume:P} holds, that is,  the local martingale
	\begin{equation}
			\biggl(	\mathcal{E}\Bigl(-\frac{q\mu}{\varsigma}\int_0^\cdot X_s\,dW_{1,s}-q\int_0^\cdot\hat{\xi}(X_s,s;T)\,dW_{2,s}\Bigr)_t\biggr)_{ 0\leq t\leq T } 
		\end{equation}  
	is a true martingale under the measure ${\mathbb{P}}$.
\end{prop}

\begin{proof}
	In order to show this is a true martingale, we use Theorem 8.1 in \cite{klebaner2014stochastic}. Recall that
	\begin{align}  
			dX_t = k(\overline{m}-X_t)\,dt+\sigma_1 \,dW_{1,t}+\sigma_2 \,dW_{2,t}, \qquad X_0=\chi. 
		\end{align}
	Using the notions in \cite{klebaner2014stochastic}, we have
	\begin{align*}
			a_t(x) & = k(\overline{m}-x)\\
			b_t(x) & = (\sigma_1,\sigma_2),\\
			\sigma_t(x) & = \Bigl(-\frac{q\mu x}{\varsigma},-q\hat{\xi}(x,t;T)\Bigr),
		\end{align*}
	so that
	\begin{align*}
			\Vert\sigma_t(x)\Vert^2 & = q^2\Bigl(\frac{\mu^2x^2}{\varsigma^2}+\hat{\xi}^2(x,t;T)\Bigr),\\
			L_t(x) & = 2k(\overline{m}-x)x+\sigma^2,\\
			\mathfrak{L}_t(x) & = -2x\Bigl(-k\overline{m}+\bigl(k+\frac{q\mu\sigma_1}{\varsigma}\bigr)x+q\sigma_2\hat{\xi}(x,t;T)\Bigr)+\sigma^2.
		\end{align*}
	Using Eq.\eqref{eqn:xi_OU} and the fact that $\beta(T-t)$ and $\gamma(T-t)$ are bounded functions in $t$ on $[0,T],$ one can find a positive $r>\chi=X_0$ such that 
	\[
		\Vert\sigma_t(x)\Vert^2+L_t(x)+\mathfrak{L}_t(x)\leq r(1+x^2).
		\] 
	This implies that the assumptions of Theorem 8.1 in \cite{klebaner2014stochastic} are met, and thus we obtain the desired result.
\end{proof}

Now the measure $\hat{\mathbb{P}}$ is well-defined by Eq.\eqref{eqn:hat_P_from_P_RN} and the $\hat{\mathbb{P}}$-dynamics of $X$ is
\begin{equation}
\label{eqn:hat_X_KO}
	dX_t=(k\overline{m}-(k+\frac{q\mu\sigma_1}{\varsigma})X_t -q\sigma_2\hat{\xi}(X_t,t;T) )\,dt+\sigma_1 \,d\hat{W}_{1,t}+\sigma_2 \,d\hat{W}_{2,t}.
\end{equation} 

\begin{prop}\label{prop:mart_condi_KO}
	For the Kim--Omberg model  A\ref{assume:P_hat} holds, that is, the local martingale
	\[
		\biggl(\mathcal{E}\Bigl(q\int_0^\cdot \hat{\xi} (X_s,s;T)-\xi^*(X_s)\,d\hat{W}_{2,s}\Bigr)_t\biggr)_{0\leq t\leq T}
		\]
	is a true martingale under the measure $\hat{\mathbb{P}}$.
\end{prop}

\begin{proof} 
	In order to show this is a true martingale, we use Theorem 8.1 in \cite{klebaner2014stochastic}. The proof is similar to the proof of Proposition \ref{prop:KO_A7}, thus we only state the corresponding functions,
	\begin{align*}
			a_t(x) & = k\overline{m}-(k+\frac{q\mu\sigma_1}{\varsigma})x -q\sigma_2\hat{\xi}(x,t;T), \\
			b_t(x) & = (\sigma_1,\sigma_2),\\
			\sigma_t(x) & = \bigl(0,q(\hat{\xi}(x,t;T)-\xi^*(x)\bigr),
		\end{align*}
	and it is straightforward to verify that the assumptions of Theorem 8.1 in \cite{klebaner2014stochastic} are met.
\end{proof}

Now the measure $\tilde{\mathbb{P}}$ is well-defined by Eq.\eqref{eqn:tilde_P} and the $\tilde{\mathbb{P}}$-dynamics of $X$ is 
\begin{align} 
	dX_t & = \Bigl(k\overline{m}-(k+\frac{q\mu\sigma_1}{\varsigma})X_t -q\sigma_2\xi^*(X_t) \Bigr)\,dt+\sigma_1 \,d\tilde{W}_{1,t}+\sigma_2 \,d\tilde{W}_{2,t}\\
	&=\biggl(k\overline{m}-\frac{q\sigma_2^2C}{1-q}-\Bigl(k+\frac{q\mu\sigma_1}{\varsigma}+\frac{q\sigma_2^2B}{1-q}\Bigr)X_t \biggr)\,dt+\sigma_1 \,d\tilde{W}_{1,t}+\sigma_2 \,d\tilde{W}_{2,t},  \qquad X_0=\chi. 
\end{align}

\begin{prop}
	For the Kim--Omberg model,  A\ref{assume:M}  holds, that is, the process
	\[
	M =\biggl(\mathcal{E}\Bigl(-\int_0^\cdot(BX_s+C)\sigma_{1}\,d\tilde{W}_{1,s}-\int_0^\cdot(BX_s+C)\sigma_{2}\,d\tilde{W}_{2,s} \Bigr)_t\biggr)_{0\le t\le T}
	\]
	is a martingale under the measure  $\tilde{\mathbb{P}}$.	 
\end{prop}

\begin{proof}
	In order to show this is a true martingale, we use Theorem 8.1 in \cite{klebaner2014stochastic}. The proof is similar to the proof of Proposition \ref{prop:KO_A7}, thus we only state the corresponding functions,
	\begin{align*}
			a_t(x) & = k\overline{m}-\frac{q\sigma_2^2C}{1-q}-\Bigl(k+\frac{q\mu\sigma_1}{\varsigma} +\frac{q\sigma_2^2B}{1-q}\Bigr)x, \\
			b_t(x) & = (\sigma_1,\sigma_2),\\
			\sigma_t(x) & = \bigl(-\sigma_1(Bx+C),-\sigma_2(Bx+C)\bigr),
	\end{align*}
	and it is straightforward to verify that  the assumptions of Theorem 8.1 in \cite{klebaner2014stochastic} are met.
\end{proof}

Now the measure $\overline{\mathbb{P}}$ is well-defined by Eq.\eqref{eqn:overline_P} and the $\overline{\mathbb{P}}$-dynamics of $X$ is
\begin{equation}
\label{eqn:OU_bar_P_dynamics}
	dX_t = \Biggl( k\overline{m}-\Bigl( \sigma_1^2+\frac{\sigma_2^2 }{1-q}\Bigr)C -\biggl( k+\frac{q\mu\sigma_1}{\varsigma}+\Bigl(\sigma_1^2+\frac{\sigma_2^2 }{1-q}\Bigr)B\biggr)X_t\Biggr)\,dt + \sigma_{1}\,d\overline{W}_{1,t} + \sigma_{2}\,d\overline{W}_{2,t} ,
\end{equation} 
which is again the OU process with re-parametrization.

\begin{prop}\label{prop:KO_mart_condi}
	For the Kim--Omberg model A\ref{assume:Q} holds, that is, the local martingale
	\[
			\biggl(\mathcal{E}\Bigl(q\int_0^\cdot\xi^*(X_s)-\hat{\xi} (X_s,s;T)\,d\overline{W}_{2,s}\Bigr)_t\biggr)_{0\leq t\leq T}
		\]
	is a true martingale under the measure $\overline{\mathbb{P}}$.
\end{prop}

\begin{proof}
	In order to show this is a true martingale, we use Theorem 8.1 in \cite{klebaner2014stochastic}. The proof is similar to the proof of Proposition \ref{prop:KO_A7}, thus we only state the corresponding functions,
	\begin{align*}
			a_t(x) & = k\overline{m}-\Bigl(\sigma_1^2+\frac{\sigma_2^2 }{1-q}\Bigr)C -\biggl(k+\frac{q\mu\sigma_1}{\varsigma}+\Bigl(\sigma_1^2+\frac{\sigma_2^2 }{1-q}\Bigr)B\biggr)x,\\
			b_t(x) & = (\sigma_1,\sigma_2),\\
			\sigma_t(x) & = \Bigl(0,q\bigl(\xi^*(x)-\hat{\xi} (x,t;T)\bigr)\Bigr),
		\end{align*}
	and it is straightforward to verify that the assumptions of Theorem 8.1 in \cite{klebaner2014stochastic} are met.	 
\end{proof}

Now the measure $\mathbb{Q}$  is well-defined by Eq.\eqref{eqn:bar_P_Q} and the $\mathbb{Q}$-dynamics of $X$ is 
\begin{equation}
\label{eqn:X_under_Q}
	dX_t =\biggl(k\overline{m}-C\sigma^2-\frac{q\sigma_2^2 }{1-q}  \gamma(T-t) -\Bigl(k+\frac{q\mu\sigma_1}{\varsigma}+B\sigma^2+\frac{q\sigma_2^2 }{1-q}\beta(T-t)\Bigr)X_t \biggr)\,dt+\sigma\,dB_{t}
\end{equation}
for $0\leq t\leq T$. The functions $f$ and $\kappa$ in Eq.\eqref{eqn:kappa} are
\begin{equation}
\label{eqn:f_OU}
	f(x,t;T) =-\frac{q\sigma_2^2}{2(1-q)}\Bigl(\bigl(B-\beta(T-t)\bigr)x+\bigl(C-\gamma(T-t)\bigr)\Bigr)^2
\end{equation}
and  
\begin{equation}
\label{eqn:OU_kappa}
	\kappa(x,t;T) = k\overline{m}-C\sigma^2-\frac{q\sigma_2^2 }{1-q}  \gamma(T-t) -\Bigl(k+\frac{q\mu\sigma_1}{\varsigma}+B\sigma^2+\frac{q\sigma_2^2 }{1-q}\beta(T-t)\Bigr)x.
\end{equation}

\subsectionfont{\scriptsize}
\subsection{Integrability condition}

In the following we prove integrability conditions, which will be needed in the analysis in the next sections.

\begin{lemma} \label{lem:our_lemma}
	Let $\theta,$ $\sigma$ be two positive constants and  let  $W$ be a Brownian motion. Define $Z_t=\sigma e^{-\theta t}\int_0^te^{\theta s}\,dW_s$ for $t\ge0,$ which is the solution of the SDE
	\[
			dZ_t=-\theta  Z_t \,dt+\sigma  dW_t, \qquad Z_0=0.
		\] 
	For any $\alpha>0$ and  $\delta<\frac{\alpha\theta}{\sigma^2},$ the expectation $\mathbb{E}[e^{\delta e^{-\alpha T}\int_0^T e^{\alpha s}Z_s^2\,ds}]$ is uniformly bounded for $T\ge0.$ 
\end{lemma}

\begin{proof}
	If $\delta\leq 0,$ then the boundedness is trivial since the exponent is negative. Assume that $0<\delta<\frac{\alpha\theta}{\sigma^2}.$ Using the change of variable $u=e^{\alpha s},$ we get
	\begin{equation}
			\delta e^{-\alpha T}\int_0^T e^{\alpha s}Z_s^2\,ds = \frac{1}{e^{\alpha T}-1}\int_1^{e^{\alpha T}}\frac{\delta}{\alpha} (1-e^{-\alpha T})Z_{(\ln u)/\alpha}^2\,du.
		\end{equation}
	From  Jensen's inequality it follows that
	\begin{equation}
			e^{\delta e^{-\alpha T}\int_0^T e^{\alpha s}Z_s^2\,ds} 	\leq \frac{1}{e^{\alpha T}-1}\int_1^{e^{\alpha T}}e^{\frac{\delta}{\alpha} (1-e^{-\alpha T})Z_{(\ln u)/\alpha}^2}\,du \leq \frac{1}{e^{\alpha T}-1}\int_0^{T}\alpha e^{\alpha s} e^{\frac{\delta}{\alpha}  Z_{s}^2}\,ds.	
		\end{equation} 
	The random variable $Z_s$ is  normally distributed with mean $0$ and variance $\frac{\sigma^2}{2\theta}(1-e^{-2\theta t}).$ Thus, for $0<\delta<\frac{\alpha\theta}{\sigma^2},$ the expectation  
	$\mathbb{E}[e^{\frac{\delta}{\alpha}  Z_{s}^2}]$ is bounded on $0\leq s<\infty.$ Let $C$ be a positive number such that $\mathbb{E}[e^{\frac{\delta}{\alpha}  Z_{s}^2}]\leq C$ for all $0\leq s<\infty.$ 
	It follows that
	\[
			\mathbb{E}[e^{ \delta e^{-\alpha T}\int_0^T e^{\alpha s}Z_s^2\,ds}]	\leq  \frac{1}{e^{\alpha T}-1}\int_0^{T}\alpha e^{\alpha s} \mathbb{E}[e^{\frac{\delta}{\alpha}  Z_{s}^2}]\,ds	\leq  \frac{C\alpha  }{e^{\alpha T}-1}\int_0^{T} e^{\alpha s} \,ds =C,
	\]
	which gives the desired result.
\end{proof}

We introduce the shorthand 
\[
\zeta_t=\zeta(X_t,t;T) =\xi^*(X_t)-\hat{\xi}(X_t,t;T)
\]
to avoid a notationally heavy expression. From \eqref{eqn:OU_bar_P_dynamics}, the $\overline{\mathbb{P}}$-dynamics of $X$ satisfies
\begin{equation} 
	dX_t = \Bigl(\frac{\alpha_1\alpha_3}{\alpha_4}-\alpha_4\,X_t\Bigr)\,dt + \sigma_{1}\,d\overline{W}_{1,t}+\sigma_{2}\,d\overline{W}_{2,t} 
\end{equation}  
which is a re-parametrized OU process.

\begin{lemma}
\label{lem:OU_bdd}
	For any  
	\[
			\delta<\frac{(1-q)^2\alpha_2^2}{\sigma^2\sigma_2^2}\frac{(\alpha_4+\alpha_1)^2}{ (\alpha_4-\alpha_1)^2},
		\]
	the expectation
	\[
			\mathbb{E}^\mathbb{\overline{P}}[e^{\delta\int_0^T\zeta^2(X_s,s;T)\,ds}]
		\]
	is uniformly bounded in $T\ge0.$
\end{lemma}

\begin{proof} Define  $a:=\alpha_1\alpha_3/\alpha_4^2,$  and a process $\overline{W}:=\frac{\sigma_{1}}{\sigma}\overline{W}_{1}+\frac{\sigma_{2}}{\sigma}\overline{W}_{2}$ so that the process $X$ satisfies
	\[
			dX_t=\alpha_4(a-X_t)\,dt+\sigma\,d\overline{W}_t, \qquad X_0=\chi.
		\]
	The solution of this SDE is
	\begin{equation}
			X_t =\chi e^{-\alpha_4 t}+a(1-e^{-\alpha_4 t})+  Z_t
		\end{equation} 
	where $Z_t=\sigma e^{-\alpha_4 t}\int_0^te^{\alpha_4 s}\,d\overline{W}_s.$	 
	From Eq.\eqref{eqn:xi_OU} and \eqref{eqn:xi_star_OU}, it can be shown that 
	\begin{equation}
		\zeta(x,t;T) =\xi^*(x)-\hat{\xi}(x,t;T) =\frac{\sigma_2  }{1-q}\Bigl((B-\beta(T-t))x+C-\gamma(T-t)\Bigr)  
		\end{equation}
	and it is easy to show that
	\begin{equation}
		\label{eqn:conv_rate_beta_gamma}
			\vert B-\beta(t) \vert \leq \frac{2\alpha_4(\alpha_4-\alpha_1)}{\alpha_2(\alpha_4+\alpha_1)}e^{-2\alpha_4 t}, \qquad \vert C-\gamma(t) \vert \leq c_0e^{-2\alpha_4t}
		\end{equation}
	for some positive constant $c_0$. For the second inequality, we observe that
	\[
			\lim_{t\to\infty}\frac{\gamma(t)-C}{e^{-2\alpha_4 t}} =	\lim_{t\to\infty}\frac{{\alpha_3}\int_0^t \beta(s)\mu(s)\,ds-C\mu(t)}{\mu(t)e^{-2\alpha_4 t}} = \lim_{t\to\infty}\frac{(\alpha_3-C\alpha_2)(\beta(t)-B)}{(\alpha_1+\alpha_2 \beta(t)-2\alpha_4)e^{-2\alpha_4 t}} 
		\]
	and the limit converges to a nonzero constant. Here, we used $\alpha_3B-C(\alpha_1+\alpha_2B)=0,$ L'H\^{o}pital's rule and Eq.\eqref{eqn:KO_beta}. Then
	\begin{equation}
			\zeta^2(x,t;T) \leq  c_1^2e^{-4\alpha_4 (T-t)}x^2+\textnormal{(const)}\,e^{-4\alpha_4 (T-t)}x+\textnormal{(const)}\,e^{-4\alpha_4 (T-t)}  
		\end{equation}
	where
	\begin{equation}
		\label{eqn:KO_c_1}
			c_1:=\frac{2\sigma_2\alpha_4(\alpha_4-\alpha_1)}{(1-q)\alpha_2(\alpha_4+\alpha_1)}.
		\end{equation}  	
	The large-time behavior  of the expectation $\mathbb{E}^\mathbb{\overline{P}}[e^{\delta\int_0^T\zeta^2(X_s,s;T)\,ds}]$ depends only on the highest-order term $c_1^2e^{-4\alpha_4(T-t)}X_t^2.$ Using that   $X_t\leq  Z_t + x+a$, it suffices to 	prove that for such a $\delta$ the expectation
	\[
			\mathbb{E}^{\overline{\mathbb{P}}} e^{\delta c_1^2 e^{-4\alpha_4 T}\int_0^T e^{4\alpha_4s}Z_s^2\,ds}
		\]
	is uniformly bounded in $T\geq0.$ Lemma \ref{lem:our_lemma}  gives that this expectation is uniformly bounded in $T\ge0$ if	$\delta c_1^2<\frac{4\alpha_4^2}{\sigma^2},$ which gives the desired result. 
\end{proof}


\begin{lemma}\label{lem:measure_change_KO}
	There are  positive numbers $c$ and $r>1$ such that for any $T\geq0$ and any nonnegative path functiona $h$  
	\[
			\mathbb{E}^\mathbb{Q}[h(X_{\mathbf{\cdot}\wedge  T})]\leq c \bigl(\mathbb{E}^{\overline{\mathbb{P}}}[h^r(X_{\cdot\wedge  T})]\bigr)^{1/r}.
		\]
\end{lemma}

\noindent We emphasize that the positive constants $c$ and $r$ do not depend on the time $T\ge0$ and the  nonnegative functional $h.$
\begin{proof}
	One can first find a positive $\delta$ such that
	\[
			\mathbb{E}^{\overline{\mathbb{P}}}e^{\frac{1}{2}\delta q^2\int_0^T\zeta_s^2\,ds}
		\]
	is uniformly bounded in $T\ge0$ by using Lemma \ref{lem:OU_bdd}. Choose $r_1>1$ and $r_2>1$ so that  $\delta=r_1(r_2-1),$ and define $r>1$ by $\frac{1}{r}+\frac{1}{r_1}+\frac{1}{r_2}=1.$ Then  
	\begin{align*} 
			\mathbb{E}^{\mathbb{Q}} [h(X_{\cdot\wedge  T})] & =\mathbb{E}^{\overline{\mathbb{P}}}\Bigl[h(X_{\cdot\wedge  T})e^{q\int_0^T\zeta_s\,d\overline{W}_{2,s}-\frac{q^2}{2}\int_0^T\zeta_s^2\,ds} \Bigr]  \\
			& \leq \bigl(\mathbb{E}^{\overline{\mathbb{P}}}[h^{r}(X_{\cdot\wedge  T})]\bigr)^{\frac{1}{r}} \biggl(\mathbb{E}^{\overline{\mathbb{P}}}\Bigl[e^{\frac{1}{2}r_1(r_2-1)q^2\int_0^T\zeta_s^2\,ds}\Bigr]\biggr)^{\frac{1}{r_1}} \biggl(\mathbb{E}^{\overline{\mathbb{P}}}\Bigl[e^{r_2q\int_0^T\zeta_s\,d\overline{W}_{2,s}-\frac{1}{2}r_2^2q^2\int_0^T\zeta_s^2\,ds}\Bigr]\biggr)^{\frac{1}{r_2}}.
		\end{align*}
	The last term is a positive local martingale so that the expectation is less than or equal to $1.$ It follows that
	\begin{equation}
			\mathbb{E}^{\mathbb{Q}} [h(X_{\cdot\wedge  T})] \leq \bigl(\mathbb{E}^{\overline{\mathbb{P}}}[h^{r}(X_{\cdot\wedge  T})]\bigr)^{\frac{1}{r}} \bigl(\mathbb{E}^{\overline{\mathbb{P}}}[e^{\frac{1}{2}\delta q^2\int_0^T\zeta_s^2\,ds}]\bigr)^{\frac{1}{r_1}}
		\end{equation}
	The second term $\mathbb{E}^{\overline{\mathbb{P}}}[e^{\frac{1}{2}\delta q^2\int_0^T\zeta_s^2\,ds}]$ is uniformly bounded in $T\ge0$ by the choice of $\delta.$ This gives the desired result.
\end{proof}

\begin{lemma}\label{lem:KO_finite_power}
	For any $\delta>0,$ the expectation 
	\[
			\mathbb{E}^\mathbb{Q}\bigl[\vert X_T\vert^\delta\bigr]
		\]
	is uniformly bounded in $(x,T)$ on $(\chi-1,\chi+1)\times[0,\infty).$ 
\end{lemma}

\begin{proof} 
	From Lemma \ref{lem:measure_change_KO}, there  are  positive numbers $c$ and $r>1$ such that for any $T\geq0$ and  
	\begin{equation}
		\label{eqn:KO_eq}
			\mathbb{E}^\mathbb{Q}\bigl[\vert X_T\vert^\delta\bigr] \leq c \Bigl(\mathbb{E}^{\overline{\mathbb{P}}}\bigl[\vert X_T\vert^{r\delta} \bigr]\Bigr)^{ {1}/{r}}.
		\end{equation}
	The right-hand side is uniformly bounded in $(x,T)$ on $(\chi-1,\chi+1)\times[0,\infty)$  since $X$ is an OU process under the measure $\overline{\mathbb{P}}.$	
\end{proof}

\begin{lemma}\label{eqn:KO_exponent_u}
	There are a number $u>1$ and an open neighborhood $I_\chi$ of $\chi$ such that 
	\[
			\Gamma_u(x,T):=\mathbb{E}^{\mathbb{Q}}\Bigl[\frac{1}{\phi^{u}({X}_T)}\, e^{u\int_0^Tf(X_s,s;T)\,ds}\Bigr]
		\]
	is uniformly bounded on $I_\chi\times [0,\infty).$
\end{lemma}

\begin{proof}
	Since the function  $f$ is nonpositive  as one can see in Eq.\eqref{eqn:f_OU}, it suffices to show that there is a number $u>1$ such that
	\[
			\mathbb{E}^{\mathbb{Q}}\Bigl[\frac{1}{\phi^{u}({X}_T)}\Bigr]=\mathbb{E}^{\mathbb{Q}}\Bigl[\frac{1}{\phi^{u}({X}_T)} \, \Big\vert \,  X_0=x\Bigr]
		\]
	is uniformly bounded in $(x,T)$ on $(\chi-1,\chi+1)\times [0,\infty).$ Define
	\[
			\beta^\mathbb{Q}(t):=k+\frac{q\mu\sigma_1}{\varsigma}+B\sigma^2+\frac{q\sigma_2^2 }{1-q}\beta(t), \qquad \gamma^\mathbb{Q}(t):=k\overline{m}-C\sigma^2-\frac{q\sigma_2^2 }{1-q}  \gamma(t),
		\]
	then the $\mathbb{Q}$-dynamics of $X$ is 
	\begin{equation}  
			dX_t = \bigl( \gamma^\mathbb{Q}(T-t) -	\beta^\mathbb{Q}(T-t)X_t \bigr)\,dt+\sigma \,dB_{t}, \qquad X_0=x
		\end{equation}
	for $0\leq t\leq T$. Solving this SDE, it follows that
	\[
			X_T=xe^{-\int_0^T\beta^\mathbb{Q}(T-s)\,ds}+e^{-\int_0^T\beta^\mathbb{Q}(T-u)\,du}\int_0^T \gamma^\mathbb{Q}(T-s)e^{\int_{0}^s\beta^\mathbb{Q}(T-u)\,du }\,ds+\sigma e^{-\int_0^T\beta^\mathbb{Q}(T-u)\,du} \int_0^Te^{\int_{0}^s\beta^\mathbb{Q}(T-u)\,du }\,dB_s.
		\]
	The random variable $X_T$ is normally distributed with mean
	\[
			m_T = xe^{-\int_0^T\beta^\mathbb{Q}(T-s)\,ds}+e^{-\int_0^T\beta^\mathbb{Q}(T-s)\,ds}\int_0^T \gamma^\mathbb{Q}(T-s)e^{\int_{0}^s\beta^\mathbb{Q}(T-u)\,du }\,ds = xe^{-\int_0^T\beta^\mathbb{Q}(s)\,ds}+\int_0^T \gamma^\mathbb{Q}(s)e^{-\int_{0}^s\beta^\mathbb{Q}(u)\,du }\,ds 
		\]
	and variance
	\begin{align*}
			v_T^2 =\sigma^2 e^{-2\int_0^T\beta^\mathbb{Q}(T-u)\,du} \int_0^Te^{2\int_{0}^s\beta^\mathbb{Q}(T-u)\,du }\,ds = \sigma^2 \int_0^Te^{-2\int_{0}^{s}\beta^\mathbb{Q}(u)\,du }\,ds. 
		\end{align*}  
	In addition, it is easy to check the limits exist, i.e.,
	\[
			m_\infty:=\lim_{T\rightarrow \infty}m_T=\int_0^\infty \gamma^\mathbb{Q}(s)e^{-\int_{0}^s\beta^\mathbb{Q}(u)\,du }\,ds, \qquad v_\infty^2:=\lim_{T\rightarrow \infty}v_T^2=\sigma^2 \int_0^\infty e^{-2\int_{0}^{s}\beta^\mathbb{Q}(u)\,du }\,ds.
		\]
	The $\mathbb{Q}$-density function of $X_T$ is
	\[
			\frac{1}{(2\pi v_T^2)^{1/2}}e^{-\frac{1}{2}\frac{(x-m_T)^2}{v_T^2}},
		\]
	thus
	\begin{equation}
		\label{eqn:KO_Q_density}
			\mathbb{E}^{\mathbb{Q}}\Bigl[\frac{1}{\phi^{u}({X}_T)}\Bigr] = \mathbb{E}^{\mathbb{Q}}\Bigl[e^{\frac{1}{2}uBX_T^2+uCX_T}\Bigr] = \frac{1}{(2\pi v_T^2)^{1/2}}\int_{-\infty}^\infty e^{\frac{1}{2}uBz^2+uCz-\frac{1}{2}\frac{(z-m_T)^2}{v_T^2}}\,dz.
		\end{equation}	
	Observe that $\beta^\mathbb{Q}(t)\geq k+\frac{q\mu\sigma_1}{\varsigma}+B\sigma^2.$ We have
	\[
			v_T^2\leq v_\infty^2=\sigma^2 \int_0^\infty e^{-2\int_{0}^{s}\beta^\mathbb{Q}(u)\,du }\,ds \leq \sigma^2 \int_0^\infty e^{-2(k+\frac{q\mu\sigma_1}{\varsigma}+B\sigma^2)s }\,ds=\frac{\sigma^2}{2(k+\frac{q\mu\sigma_1}{\varsigma}+B\sigma^2)} .
		\]
	The integral in Eq.\eqref{eqn:KO_Q_density} satisfies
	\begin{equation}
		\label{eqn:density_KO_estimate}
			\int_{-\infty}^\infty e^{\frac{1}{2}uBz^2+uCz-\frac{1}{2}\frac{(z-m_T)^2}{v_T^2}}\,dz \leq  \int_{-\infty}^\infty e^{\frac{1}{2}uBz^2+uCz-\frac{1}{\sigma^2}(k+\frac{q\mu\sigma_1}{\varsigma}+B\sigma^2)(z-m_T)^2}\,dz.
		\end{equation}
	Using the condition $k+\frac{q\mu\sigma_1}{\varsigma}+\frac{B\sigma^2}{2}>0,$ one can choose a small $u>1$ such that the right-hand side is uniformly bounded in $(x,T)$ on $(\chi-1,\chi+1)\times [0,\infty).$
\end{proof}

\subsectionfont{\scriptsize}
\subsection{Sensitivity with respect to the initial volatility}
\label{app:KO_initial}

The purpose of this section is to prove the following proposition, which yields the first statement of Theorem \ref{thm:KO_asymp}.

\begin{prop}\label{prop:delta_KO}
	For the  Kim--Omberg model presented in Eq.\eqref{eqn:KO_model}, the long-term sensitivity with respect to the initial value of the volatility is
	\[
			\lim_{T\rightarrow\infty}\frac{\partial }{\partial \chi}\ln v(\chi,T)=-B\chi-C.
		\]
\end{prop}

\begin{proof} 
	By Theorem \ref{thm:delta}, it suffices to prove that the expectation $\mathbb{E}^{\mathbb{Q}}\bigl[\frac{1}{\phi(X_T)}\, e^{\int_0^Tf(X_s,s;T)\,ds} \, \big\vert \, X_0=x\bigr]$ is continuously differentiable in $x,$ and 
	\[
			\frac{\partial }{\partial x}\mathbb{E}^{\mathbb{Q}}\Bigl[\frac{1}{\phi(X_T)}\, e^{\int_0^Tf(X_s,s;T)\,ds} \, \Big\vert \, X_0=x\Bigl]
		\]
	converges to zero as $T\to\infty.$ To prove this, we apply  Proposition  \ref{prop:delta}. Condition (i) of this proposition was proved in Lemma \ref{eqn:KO_exponent_u}. For (ii), we fix any $v>1.$ By Lemma \ref{lem:KO_finite_power}, it follows that 
	\[
				\mathbb{E}^{\mathbb{Q}} \biggl\vert\frac{\phi'(X_T)}{\phi(X_T)}\biggr\vert^{v}= \mathbb{E}^{\mathbb{Q}} \vert BX_T+C \vert^{v} 
		\]
	is uniformly bounded in $(x,T)$ on $(\chi-1,\chi+1)\times[0,\infty).$ 
	To show (iii), we calculate the first variation process $Y$ of $X$ given  Eq.\eqref{eqn:X_under_Q}. Then   $Y_t=Y_{t;T}$ satisfies  
	\[
			dY_t=-\Bigl(k+\frac{q\mu\sigma_1}{\varsigma}+B\sigma^2+\frac{q\sigma_2^2}{1-q}\beta(T-t)\Bigr)Y_t\,dt, \qquad Y_0=1, \qquad 0\leq t\leq T,
		\]
	which is a deterministic process. It follows that
	\[
			Y_{t;T}=e^{-(k+\frac{q\mu\sigma_1}{\varsigma}+B\sigma^2)t-\frac{q\sigma_2^2}{1-q}\int_0^t\beta(T-s)\,ds}.
		\]
	By direct calculation, for any fixed $w>1,$  it is clear that 
	\[
			\lim_{T\to\infty}\mathbb{E}^\mathbb{Q} \vert Y_{T;T} \vert^w = \lim_{T\to\infty}e^{-w(k+\frac{q\mu\sigma_1}{\mu}+B\sigma^2)T-w\frac{q\sigma_2^2}{1-q}\int_0^T\beta(T-s)\,ds}=0
		\]
	since $k+\frac{q\mu\sigma_1}{\mu}+B\sigma^2>0$ and $\beta(\cdot)>0.$  
	
	We now consider (iv). By using Eq.\eqref{eqn:conv_rate_beta_gamma}, it can be easily shown that there are positive constants $c_1$ and $c_2$ such that  
	\[
			\bigl \vert f_x(x,t;T) \bigr\vert =\frac{q\sigma_2^2}{1-q}\biggl\vert \Bigl(\bigl(B-\beta(T-t)\bigr)x+\bigl(C-\gamma(T-t)\bigr)\Bigr)\Bigl(B-\beta(T-t)\Bigr)\biggr\vert \leq c_1e^{-c_2(T-t)}(\vert x \vert+1).
		\]
	By using $Y_{t;T}\leq e^{-\nu t}$ where $\nu:=k+\frac{q\mu\sigma_1}{\varsigma}+B\sigma^2,$ we obtain that for any $m>1$
	\begin{equation}
			\mathbb{E}^{\mathbb{Q}}\biggl[\biggl(\int_0^T \vert f_x(X_s,s;T)Y_{s;T} \vert \,ds\biggr)^{m}  \biggr] \leq c_1^m T^{m-1}e^{-c_2mT}\int_0^T  e^{(c_2-\nu) ms}\mathbb{E}^{\mathbb{Q}}\bigl[(\vert X_s\vert +1)^m\bigr]\,ds 
		\end{equation} 
	by Jensen's inequality. Using Lemma \ref{lem:KO_finite_power}, we observe that for each $m>1,$ the expectation $\mathbb{E}^{\mathbb{Q}}[( \vert X_s \vert +1)^m]$ is uniformly  bounded in  $s\geq0$ by a positive constant $C_m.$ Thus,
	\begin{equation}
		\label{KO_delta_condi_4}
			\mathbb{E}^{\mathbb{Q}}\biggl[\biggl(\int_0^T \vert f_x(X_s,s;T)Y_{s;T} \vert \,ds\biggr)^{m} \biggr] \leq  \frac{c_1^m C_m}{(c_2-\nu)m}   T^{m-1}  \Bigl(e^{-\nu mT}-e^{-c_2mT}\Bigr)\to 0
		\end{equation}
	as $T\to\infty.$ Finally, conditions (ii), (iii), (iv) in Proposition \ref{prop:delta} hold true for arbitrary $v,w,m>1,$ and (i) holds for some $u>1,$ so we obtain the desired result.	
\end{proof}

\subsectionfont{\scriptsize}
\subsection{Sensitivities with respect to \texorpdfstring{$k,$ $\overline{m},$ $\mu,$ $\varsigma$}{kmms} and \texorpdfstring{$\rho$}{r}}
\label{sec:sen_KO_k_b}

We compute the long-term sensitivity with respect to the perturbation of $k.$ Those  with respect to the parameters $\overline{m},$ $\mu,$ $\varsigma$ and $\rho$ can be calculated in a similar way because all these parameters affect the functionals $\phi,$ $f$ and the drift of $X$ but not the volatility of $X$ as seen in the $\mathbb{Q}$-dynamics  of $X$
\begin{equation} 
	dX_t = \biggl(k\overline{m}-C\sigma^2-\frac{q(1-\rho^2)\sigma^2 }{1-q}  \gamma(T-t) -\Bigl(k+\frac{q\mu\rho\sigma}{\varsigma}+B\sigma^2+\frac{q(1-\rho^2)\sigma^2}{1-q}\beta(T-t)\Bigr)X_t \biggr)\,dt+\sigma\,dB_{t}
\end{equation}
for $0\leq t\leq T$. The five functions in B\ref {bassume:perturb} and B\ref{bassume:HS_eps} are 
\begin{align*}
	m_\epsilon(x)=(k+\epsilon) (\overline{m} - x), \qquad  \sigma_{1,\epsilon}(x)=\sigma_1, \qquad \sigma_{2,\epsilon}(x)=\sigma_2, \qquad b_\epsilon(x)=\mu x, \qquad \varsigma_\epsilon(x)=\varsigma
\end{align*}  
and it is easy to check that they satisfy assumptions B\ref {bassume:perturb} and B\ref{bassume:HS_eps}. Observe that 
\[
	\frac{\partial }{\partial \epsilon}\Big\vert_{\epsilon=0}\ln v_\epsilon(\chi,T)=\frac{\partial }{\partial k}\ln v_0(\chi,T)=\frac{\partial }{\partial k}\ln v(\chi,T),
	\]
thus for the rest of this section we use $\frac{\partial }{\partial k}$ instead of $\frac{\partial }{\partial \epsilon}\vert_{\epsilon=0}.$

\begin{lemma}\label{lem:KO_finite_inte_power}
	Let $\alpha>0$ and $\ell>0.$ The expectation 
	\[
			\mathbb{E}^{\mathbb{Q}}\biggl[\biggl(	\int_0^Te^{-\alpha(T-s)}X_s^2\,ds\biggr)^\ell\biggr]
		\]
	is uniformly bounded in $T$ on $[0,\infty).$
\end{lemma}

\begin{proof} 
	By Lemma \ref{lem:measure_change_KO}, there are  positive numbers $c$ and $r,$ independent of $T,$ such that   
	\[
	\mathbb{E}^{\mathbb{Q}}\biggl[\biggr(\int_0^Te^{-\alpha(T-s)}X_s^2\,ds\biggr)^\ell\biggr] \leq c\,\Biggl( \mathbb{E}^{\overline{\mathbb{P}}}\biggl[\biggr(	\int_0^Te^{-\alpha(T-s)}X_s^2\,ds\biggr)^{r\ell}\biggr]\Biggr)^{1/r}.
	\]
	From Lemma \ref{lem:our_lemma}, we know that
	\[
			\mathbb{E}^{\overline{\mathbb{P}}}\Bigl[ e^{\delta\int_0^Te^{-\alpha(T-s)}X_s^2\,ds}\Bigr]
		\]
	is uniformly bounded in $T$ for sufficiently small $\delta>0.$ Choose $n\in\mathbb{N}$ such that $r\ell\leq n.$	Using the inequality $\frac{x^n}{n!} \leq e^x  $ for $x>0,$
	we have
	\[
			 \frac{\delta^n}{n!}\mathbb{E}^{\overline{\mathbb{P}}}\biggl[\biggl( \int_0^Te^{-\alpha(T-s)}X_s^2\,ds\biggr)^{r\ell}\biggr]\leq \frac{\delta^n}{n!}\mathbb{E}^{\overline{\mathbb{P}}}\biggl[\biggl( 	\int_0^Te^{-\alpha(T-s)}X_s^2\,ds\biggr)^n\biggr]\leq\mathbb{E}^{\overline{\mathbb{P}}}\Bigl[e^{\delta\int_0^Te^{-\alpha(T-s)}X_s^2\,ds}\Bigr].
		\]
	Thus,
	\[
			\mathbb{E}^{\overline{\mathbb{P}}}\biggl[\biggl(\int_0^Te^{-\alpha(T-s)}X_s^2\,ds\biggr)^{r\ell }\biggr]
		\]
	is also uniformly bounded in $T$ on $[0,\infty),$ which gives the desired result.	
\end{proof}
 
\begin{prop}\label{prop:KO_k}
	For the  Kim--Omberg model presented in Eq.\eqref{eqn:KO_model}, the long-term sensitivity with respect to the parameter $k$ is
	\[
			\lim_{T\rightarrow\infty}\frac{1}{T}\frac{\partial }{\partial k}\ln v(\chi,T)=-\frac{\partial\lambda }{\partial k}.
	\]
\end{prop}

\begin{proof}
	To prove this equality, we use Theorem \ref{thm:total_chain}. Condition (i) in Theorem \ref{thm:total_chain} is satisfied trivially. We prove (iii) in Theorem \ref{thm:total_chain} first because some techniques used for (iii) are also used in the proof of (ii). For condition (iii) in Theorem \ref{thm:total_chain}, we apply Theorem \ref{thm:rho}. It can be easily checked that
	\[
			\Bigl\vert \frac{\partial}{\partial k}\kappa (x,t;T)\Bigr\vert \leq  c(\vert x \vert+1)
	\]
	for a positive constant $c$ independent of $t,$ $T$ and $x.$ By choosing sufficiently large $c,$ we can achieve that $\hat{g}(x,t;T) \leq c( \vert x \vert +1)$ holds true for $\hat{g}$ defined in Eq.\eqref{eqn:hats}.

	Then, (i) in Theorem \ref{thm:rho} can be proven as follows. Since $X$ is an OU process under the measure $\overline{\mathbb{P}},$ for each $T>0$ one can choose a positive $\delta=\delta(T)$ such that
	\[
			\mathbb{E}^{\overline{\mathbb{P}}}\Bigl[e^{\delta \int_0^T X_s^2\,ds}\Bigr]
		\]
	is finite. For the positive constant $r$ in Lemma \ref{lem:measure_change_KO}, we define $\epsilon_0=\frac{\delta}{2c^2r},$ then
	\begin{align*} 
			\mathbb{E}^{\mathbb{Q}}\Bigl[e^{\epsilon_0\int_0^T \hat{g}^2(X_s,s;T)\,ds}\Bigr] &\leq \mathbb{E}^{\mathbb{Q}}\Bigl[e^{\epsilon_0c^2\int_0^T ( \vert X_s \vert +1)^2\,ds}\Bigr]  \leq c'\biggl(\mathbb{E}^{\overline{\mathbb{P}}}\Bigl[e^{\epsilon_0c^2r\int_0^T ( \vert X_s \vert +1)^2\,ds}\Bigr] \biggr)^{1/r} \leq c'\biggl(\mathbb{E}^{\overline{\mathbb{P}}}\Bigl[e^{2\epsilon_0c^2r\int_0^T (X_s^2+1)\,ds}\Bigr] \biggr)^{1/r}\\
			&=c'e^{2\epsilon_0c^2T}\biggr(\mathbb{E}^{\overline{\mathbb{P}}}\Bigl[e^{2\epsilon_0c^2r\int_0^T X_s^2\,ds}\Bigr] \biggr)^{1/r} =c'e^{2\epsilon_0c^2T}\biggl(\mathbb{E}^{\overline{\mathbb{P}}}\Bigl[e^{\delta\int_0^T X_s^2\,ds}\Bigr] \biggr)^{1/r} 
		\end{align*}
	where $c'$ is the positive  constant in Lemma \ref{lem:measure_change_KO}. This gives (i) in Theorem \ref{thm:rho}.

	For (ii) in Theorem \ref{thm:rho}, we observe that for any $v\geq 2$
	\begin{align*}
			\mathbb{E}^\mathbb{Q}\biggl[\Bigl(\int_0^T \hat{g}^2(X_s,s;T)\,ds\Bigr)^{v/2}\biggr] & \leq c^v\,\mathbb{E}^\mathbb{Q}\biggl[\Bigl(\int_0^T  ( \vert X_s \vert +1)^2\,ds\Bigr)^{v/2}\biggr] \leq c^vT^{v/2}\,\Biggl(\mathbb{E}^\mathbb{Q}\biggl[\Bigl(\frac{1}{T}\int_0^T  ( \vert X_s \vert+1)^2\,ds\Bigr)^{v/2}\biggr]\Biggr)\\
			&\leq c^vT^{v/2}\,\biggl(\mathbb{E}^\mathbb{Q}\Bigl[ \frac{1}{T}\int_0^T  ( \vert X_s \vert +1)^{v}\,ds\Bigr]\biggr) = c^v T^{{v}/{2}-1}\,\biggl(  \int_0^T  \mathbb{E}^\mathbb{Q}[( \vert X_s \vert +1)^{v}]\,ds\biggr).
		\end{align*}
	By Lemma \ref{lem:KO_finite_power}, the expectation $\mathbb{E}^\mathbb{Q}[( \vert X_s \vert +1)^{v}]$ is uniformly bounded in $s$ by a positive constant, say $C.$ Then 
	\begin{equation}
		\label{eqn:KO_Q_estimate}
			\mathbb{E}^\mathbb{Q}\biggl[\Bigl(\int_0^T \hat{g}^2(X_s,s;T)\,ds\Bigr)^{v/2}\biggr] \leq   c^v T^{{v}/{2}-1}\,\biggl(  \int_0^T  \mathbb{E}^\mathbb{Q}[( \vert X_s \vert +1)^{v}]\,ds\biggr) \leq  c^vCT^{v/2}.
		\end{equation}
	Since the constants $c$ and $C$ do not depend on $T,$ we obtain the desired result. For (iii) in Theorem \ref{thm:rho}, we observe that for $\epsilon_1=1$ 
	\begin{equation}
		\label{eqn:KO_drfit_sen}
			\mathbb{E}^{\mathbb{Q}}\biggl[ \int_0^T\hat{g}^{v+\epsilon_1}(X_s,s;T) \, ds\biggr] \leq   c^{v+1} \int_0^T\mathbb{E}^{\mathbb{Q}}\Bigl[( \vert X_s \vert +1)^{v+1}\Bigr] \, ds,
		\end{equation} 
	and the right-hand side is  finite for each $T\ge0$ because the expectation $\mathbb{E}^{\mathbb{Q}}\bigl[( \vert X_s \vert +1)^{v+1}\bigr]$ is uniformly bounded in $s$ by Lemma \ref{lem:KO_finite_power}.

	For (iv) Theorem \ref{thm:rho}, we want to show that for $u$ with $1/u+1/v=1$ the expectation
	\[
			\mathbb{E}^{\mathbb{Q}}\Bigl[\frac{1}{\hat{\phi}^{u}( {X}_T)}\, e^{u\int_0^T\hat{f}(X_s,s;T)\,ds}\Bigr]
		\]
	is uniformly bounded in $T$ on $[0,\infty).$ However, observe that we proved that (ii) and (iii) in Theorem \ref{thm:rho} hold  true for arbitrary $v\ge2.$ Thus, it is enough to show that such $u>1$ exists. We use the notations $B(k)$ and $C(k)$ to emphasize the dependence of $k$ on the constants $B$ and $C,$ respectively. From Eq.\eqref{eqn:KO_Q_density} and Eq.\eqref{eqn:density_KO_estimate}, we know for a small $u_0>1$ the expectation
	\begin{equation}
		\label{eqn:KO_u_0}
			\mathbb{E}^{\mathbb{Q}}\Bigl[e^{\frac{1}{2}u_0B(k)X_T^2+u_0C(k)X_T}\Bigr]
		\end{equation} 
	is uniformly bounded in $T$ on $[0,\infty).$ Since the two maps $k\mapsto B(k)$ and $k\mapsto C(k)$ are continuous and $\frac{u_0+1}{2}>1,$ by choosing a smaller interval $I$ if necessary, it follows that 
	\[
			\sup_{\epsilon\in I}B(k+\epsilon) \leq \frac{u_0+1}{2}B(k), \qquad \sup_{\epsilon\in I}C(k+\epsilon) \leq \frac{u_0+1}{2}C(k).
		\]
	Then 
	\begin{equation}
		\label{eqn:KO_hat_phi}
			\hat{\phi}(x)=\inf_{\epsilon\in I}e^{-\frac{1}{2}B(k+\epsilon)x^2-C(k+\epsilon)x} \geq e^{-\frac{1}{2}\frac{u_0+1}{2}B(k)x^2-\frac{u_0+1}{2}C(k)x}.
		\end{equation}
	Define 
	\begin{equation}
		\label{eqn:KO_u_hat}
			\hat{u}:=\frac{2u_0}{u_0+1}>1,
		\end{equation}
	then we have
	\begin{equation}
		\label{eqn:hat_phi_hat_u}
			\mathbb{E}^{\mathbb{Q}}\Bigl[\frac{1}{\hat{\phi}^{\hat{u}}( {X}_T)}\, e^{\hat{u}\int_0^T\hat{f}(X_s,s;T)\,ds}\Bigr]\leq \mathbb{E}^{\mathbb{Q}}\Bigl[\frac{1}{\hat{\phi}^{\hat{u}}( {X}_T)}\Bigr]\leq    \mathbb{E}^{\mathbb{Q}}\Bigl[e^{\frac{1}{2}u_0B(k)X_T^2+u_0C(k)X_T}\Bigr]
		\end{equation}
	where for the first inequality we used $\hat{f}\le 0.$ Since the right-hand side is  uniformly bounded in $T$ on $[0,\infty),$ we obtain the desired result. We have now shown all conditions in  Theorem \ref{thm:rho} and thus condition (iii) in Theorem \ref{thm:total_chain} holds true.

	For condition (ii) in Theorem \ref{thm:total_chain}, we first calculate the   partial derivative with respect to the variable  $k$ in $\phi$ and $f$ but not in $X=(X_t)_{t\ge0}.$ To be precise, we  use notation $\phi(x;k)$ and $f(x,t;T;k)$ to emphasize the dependence of $k.$ We want to analyze
	\[
			w_{\eta,\epsilon}(\chi,T)=\mathbb{E}^{\mathbb{Q}_\epsilon}\Bigl[\frac{1}{\phi(X_T^\epsilon;k+\eta)}\, e^{\int_0^Tf(X_s^\epsilon,s;T;k+\eta)\,ds} \Bigr]
		\]
	where the $\mathbb{Q}^\epsilon$-dynamics of $X_t^\epsilon$ satisfies Eq.\eqref{eqn:X_under_Q} with $k$ replaced by $k+\epsilon.$ The equality
	\[
			 \frac{\partial }{\partial \eta} \mathbb{E}^\mathbb{Q}\Bigl[\frac{1}{\phi(X_T^\epsilon;k+\eta)}e^{\int_0^Tf(X_s^\epsilon,s;T;k+\eta)\,ds}\Bigr]= \mathbb{E}^\mathbb{Q}\biggl[\frac{\partial }{\partial \eta}\Bigl(\frac{1}{\phi(X_T^\epsilon;k+\eta)}e^{\int_0^Tf(X_s^\epsilon,s;T;k+\eta)\,ds}\Bigr)\biggr]
		 \]
	and the continuity of this partial derivative in $(\eta,\epsilon)$ on $I^2$ are obtained from Proposition \ref{prop:condi_2} with $g(x,t;T)$ and $G_T$ given below. Observe that
	\[
			\frac{\partial f}{\partial k}(x,t;T;k) = -\frac{q\sigma_2^2}{1-q}\Bigl(\bigl(B-\beta(T-t)\bigr)x+\bigl(C-\gamma(T-t)\bigr)\Bigr)\biggl(\Bigl(\frac{\partial B}{\partial k}-\frac{\partial\beta }{\partial k}(T-t)\Bigr)x+\Bigl(\frac{\partial C}{\partial k}-\frac{\partial\gamma}{\partial k}(T-t)\Bigr)\biggr).
		\]
	We use the notations $\beta(T-t;k),$ $\gamma(T-t;k)$ to emphasize the dependence of $k.$ For a given small open interval $I,$ since $B(k+\eta),C(k+\eta),\frac{\partial B}{\partial k}(k+\eta),\frac{\partial C}{\partial k}(k+\eta)$ are continuous in $\eta$  on $\overline{I}$ and $\beta(T-t;k+\eta),\gamma(T-t;k+\eta),\frac{\partial\beta}{\partial k}(T-t;k+\eta),\frac{\partial\gamma}{\partial k}(T-t;k+\eta)$ are continuous in $(\eta,t)$ on $\overline{I}\times[0,T],$ one can find a positive constant $b_1$  such that for all $(\eta,t)\in \overline{I}\times [0,T]$
	\[
			\biggl\vert\frac{\partial f}{\partial \eta}(x,t;T;k+\eta)\biggr\vert \leq  b_1(x^2+1) =:g(t,x;T).
		\]
	With this function $g,$ condition (i) in Proposition \ref{prop:condi_2} is trivially satisfied. For condition (ii) in Proposition \ref{prop:condi_2}, choose two positive constants $b_2$ and $c_2$ such that for all  $\eta\in \overline{I}$
	\[
			\frac{1}{2}\biggl\vert \frac{\partial B}{\partial  \eta}(k+\eta)\biggr\vert \leq b_2, \qquad \biggl\vert \frac{\partial C}{\partial \eta}(k+\eta)\biggr\vert \leq c_2.
		\]
	Using the function $\hat{\phi}$ in Eq.\eqref{eqn:KO_hat_phi}, we define 
	\[
			G_T:=\frac{1}{\hat{\phi}(X_T)}\Bigl(b_2X_T^2+c_2 \vert X_T \vert \Bigr) + \frac{1}{\hat{\phi}(X_T)}\int_0^Tb_1(X_s^2+1)\,ds.
		\]	
	Then for all  $(\eta,t)\in \overline{I}\times [0,T]$ it follows that
	\begin{equation}
			\frac{1}{\phi^2(X_T;k+\eta)}\biggl\vert \frac{\partial \phi }{\partial \eta}(X_T;k+\eta)\biggr\vert  +\frac{1}{\phi(X_T;k+\eta)}\, \int_0^T\biggl\vert \frac{\partial f}{\partial \eta}(X_s,s;T;k+\eta)\biggr\vert \,ds\leq G_T 
		\end{equation}
	by using that $\hat{\phi}(x)=\inf_{\eta\in I}\phi(x;k+\eta)$ and 
	\[
			\Bigl\vert \frac{\partial \phi }{\partial \eta}(x;k+\eta)\Bigr\vert = \Bigl\vert \frac{1}{2}x^2\frac{\partial B}{\partial  \eta}(k+\eta)+x\frac{\partial C}{\partial \eta}(k+\eta)\Bigr\vert \phi(x;k+\eta)\leq  \Bigl(b_2x^2+c_2 \vert x \vert \Bigr)\phi(x;k+\eta).
		\]
	Recall $u_0>1$ from Eq.\eqref{eqn:KO_u_0} and $\hat{\phi}(x)\ge e^{-\frac{1}{2}\frac{u_0+1}{2}B(k)x^2-\frac{u_0+1}{2}C(k)x}$ and $\hat{u}=\frac{2u_0}{u_0+1}>1$ in Eq.\eqref{eqn:KO_u_hat}. We claim 
	$\mathbb{E}^\mathbb{Q}[G_T^{u_1}]<\infty$ for $u_1=\frac{3u_0+1}{2(u_0+1)}>1,$ which implies condition (ii) in Proposition \ref{prop:condi_2}. Let $\hat{v}$ be such that $\frac{1}{\hat{u}/u_1}+\frac{1}{\hat{v}}=1$, then
	\begin{equation}
			\mathbb{E}^\mathbb{Q}\Bigl[\frac{1}{\hat{\phi}^{u_1}(X_T)}\Bigl(b_2X_T^2+c_2 \vert X_T \vert \Bigr)^{u_1}\Bigr]\leq\biggl(  \mathbb{E}^\mathbb{Q}\Bigl[\frac{1}{\hat{\phi}^{\hat{u}}(X_T)}\Bigr] \biggr)^{{u_1}/{\hat{u}}}\biggl(\mathbb{E}^\mathbb{Q}\Bigl[\Bigl(b_2X_T^2+c_2 \vert X_T \vert \Bigr)^{u_1\hat{v}}\Bigr]\biggr)^{1/\hat{v}}.
		\end{equation}
	The two expectations on the right-hand side are finite by  Eq.\eqref{eqn:hat_phi_hat_u} and Lemma \ref{lem:KO_finite_power}. In a similar way, we have
	\begin{align*}
			\mathbb{E}^\mathbb{Q}\biggl[\frac{1}{\hat{\phi}^{u_1}(X_T)}\Bigl(b_1\int_0^T(X_s^2+1)\,ds\Bigr)^{u_1}\biggr]
			&\leq\biggl(  \mathbb{E}^\mathbb{Q}\Bigl[\frac{1}{\hat{\phi}^{\hat{u}}(X_T)}\Bigr] \biggr)^{{u_1}/{\hat{u}}}\biggl(\mathbb{E}^\mathbb{Q}\Bigl[\Bigl(b_1\int_0^T(X_s^2+1)\,ds\Bigr)^{u_1\hat{v}}\Bigr]\biggr)^{1/\hat{v}}\\
			&\leq T^{u_1-\frac{1}{\hat{v}}}\biggl(  \mathbb{E}^\mathbb{Q}\Bigl[\frac{1}{\hat{\phi}^{\hat{u}}(X_T)}\Bigr] \biggr)^{{u_1}/{\hat{u}}} \biggl(\mathbb{E}^\mathbb{Q}\Bigl[b_1^{u_1\hat{v}}\int_0^T(X_s^2+1)^{u_1\hat{v}}\,ds\Bigr]\biggr)^{1/\hat{v}}\\
			&\leq T^{u_1-\frac{1}{\hat{v}}}\biggl(  \mathbb{E}^\mathbb{Q}\Bigl[\frac{1}{\hat{\phi}^{\hat{u}}(X_T)}\Bigr] \biggr)^{{u_1}/{\hat{u}}}\biggl(b_1^{u_1\hat{v}}\int_0^T\mathbb{E}^\mathbb{Q}\bigl[(X_s^2+1)^{u_1\hat{v}}\bigr]\,ds\biggr)^{1/\hat{v}}.
		\end{align*}
	Since $\mathbb{E}^\mathbb{Q}[(X_s^2+1)^{u_1\hat{v}}]$ is uniformly bounded in $s$ on $[0,\infty)$ by\ref{lem:KO_finite_power}, the right-hand side is finite. Hence, $\mathbb{E}^\mathbb{Q}[G_T^{u_1}]<\infty.$

	The convergence 
	\[
			\lim_{T\rightarrow \infty}\frac{1}{T}\frac{\partial }{\partial \eta}\Big\vert_{\eta=0}\mathbb{E}^\mathbb{Q}  \Bigl[\frac{1}{\phi(X_T;k+\eta)}e^{\int_0^Tf(X_s,s;T;k+\eta)\,ds}\Bigr] = 0
		\]
	can be shown as follows. The partial derivative with respect to $\eta$ satisfies
	\begin{align*}
			& \phantom{==} \biggl\vert \frac{\partial }{\partial \eta}\Big\vert_{\eta=0}\Bigl(\frac{1}{\phi(X_T;k+\eta)}e^{\int_0^Tf(X_s,s;T;k+\eta)\,ds}\Bigr)\biggr\vert \\	
			& \leq e^{\frac{1}{2}BX_T^2+CX_T+\int_0^Tf(X_s,s;T;k)\,ds} \Bigl\vert \frac{1}{2}X_T^2\frac{\partial B}{\partial k}+X_T\frac{\partial C}{\partial k} \Bigr\vert + e^{\frac{1}{2}BX_T^2+CX_T+\int_0^Tf(X_s,s;T;k)\,ds}\biggl\vert \int_0^T\frac{\partial f}{\partial \eta}\Big\vert_{\eta=0}(X_s,s;T;k+\eta)\,ds\biggr\vert\\
			& \leq e^{\frac{1}{2}BX_T^2+CX_T} \Bigl\vert \frac{1}{2}X_T^2\frac{\partial B}{\partial k}+X_T\frac{\partial C}{\partial k}\Bigr\vert +e^{\frac{1}{2}BX_T^2+CX_T}\biggl\vert \int_0^T\frac{\partial f}{\partial \eta}\Big\vert_{\eta=0}(X_s,s;T;k+\eta)\,ds\biggr\vert.
		\end{align*}
	By the triangle inequality and the H\"older inequality, for $u_0$ in Eq.\eqref{eqn:KO_u_0} and $v_0$ satisfying $1/u_0+1/v_0=1$ it follows that 
	\begin{align*}
			\mathbb{E}^\mathbb{Q}\biggl\vert \frac{\partial }{\partial \eta}\Big\vert_{\eta=0}\Bigl(\frac{1}{\phi(X_T;k+\eta)}e^{\int_0^Tf(X_s,s;T;k+\eta)\,ds}\Bigr)\biggr\vert 
			& \leq \Bigl(\mathbb{E}^\mathbb{Q}e^{\frac{1}{2}u_0BX_T^2+u_0CX_T}\Bigr)^{1/u_0}\Bigl(\mathbb{E}^\mathbb{Q}\Bigl\vert\frac{1}{2}X_T^2\frac{\partial B}{\partial k}+X_T\frac{\partial C}{\partial k}\Bigr\vert^{v_0}\Bigr)^{1/v_0}\\
			& \phantom{=:} +\Bigl(\mathbb{E}^\mathbb{Q}e^{\frac{1}{2}u_0BX_T^2+u_0CX_T}\Bigr)^{1/u_0}\biggl(\mathbb{E}^\mathbb{Q}\Bigl\vert\int_0^T\frac{\partial f}{\partial \eta}\Big\vert_{\eta=0}(X_s,s;T;k+\eta)\,ds\Bigr\vert^{v_0}\biggr)^{1/v_0}.
	\end{align*}
	By the choice of $u_0,$ the expectation  $\mathbb{E}^\mathbb{Q}e^{\frac{1}{2}u_0BX_T^2+u_0CX_T}$ is  uniformly bounded in $T.$ The expectation $\mathbb{E}^\mathbb{Q} \vert \frac{1}{2}X_T^2\frac{\partial B}{\partial k}+X_T\frac{\partial C}{\partial k} \vert^{v_0}$ is also uniformly bounded in $T$ by  Lemma \ref{lem:KO_finite_power}. Now, we show that the expectation $\mathbb{E}^\mathbb{Q} \vert \int_0^T\frac{\partial f}{\partial \eta}\vert_{\eta=0}(X_s,s;T;k+\eta)\,ds \vert^{v_0}$ is uniformly bounded in $T$. By direct calculation, one can choose positive constants $c$ and $d,$ which are independent of $s$ and $T$  but are dependent of $k,$ such that
	\[
			\biggl\vert \frac{\partial f }{\partial \eta}\Big\vert_{\eta=0}(x,s;T;k+\eta)\biggr\vert \leq de^{-c(T-s)}(x^2+1).
		\] 
	Using the change of variable $u=e^{c s},$ observe that
	\begin{align}
		\label{eqn:KO_estimate_change_variable}
			\mathbb{E}^\mathbb{Q}\Bigl[\Bigl(\int_0^Te^{cs}(X_s^2+1)\,ds\Bigr)^{v_0}\Bigr] & =  \mathbb{E}^\mathbb{Q}\Bigl[\Bigl(\int_1^{e^{cT}}\frac{1}{c}(X_{(\ln u)/c}^2+1)\,du\Bigr)^{v_0}\Bigr] = \frac{(e^{cT}-1)^{v_0}}{c^{v_0}} \mathbb{E}^\mathbb{Q}\Bigl[\Bigl(\frac{1}{e^{cT}-1}\int_1^{e^{cT}} (X_{(\ln u)/c}^2+1)\,du\Bigr)^{v_0}\Bigr] \nonumber\\
			& \leq \frac{(e^{cT}-1)^{v_0}}{c^{v_0}} \mathbb{E}^\mathbb{Q}\Bigl[\frac{1}{e^{cT}-1}\int_1^{e^{cT}} (X_{(\ln u)/c}^2+1)^{v_0}\,du\Bigr] \nonumber\\
			& = \frac{(e^{cT}-1)^{v_0-1}}{c^{v_0}}\int_1^{e^{cT}} \mathbb{E}^\mathbb{Q}\Bigl[\bigl(X_{(\ln u)/c}^2+1\bigr)^{v_0}\Bigr]\,du. 
		\end{align}
	By Lemma \ref{lem:KO_finite_power}, there is a positive constant $C$ such that $\mathbb{E}^\mathbb{Q}[(X_{(\ln u)/c}^2+1)^{v_0}]\leq C$ for all $u\geq1.$ Thus,
	\begin{align}
		\label{eqn:KO_estimate_change_variable_2}
			\mathbb{E}^\mathbb{Q}\biggl\vert \int_0^T\frac{\partial f}{\partial \eta}\Big\vert_{\eta=0}(X_s,s;T;k+\eta)\,ds\biggr\vert^{v_0} &\leq d^{v_0}e^{-c v_0 T}\mathbb{E}^\mathbb{Q}\Bigl[\Bigl(\int_0^Te^{c s} (X_s^2+1) \,ds \Bigr)^{v_0} \Bigr] \nonumber\\
			& \leq d^{v_0}e^{-c v_0 T}\frac{(e^{cT}-1)^{v_0-1}}{c^{v_0}} \int_1^{e^{cT}} \mathbb{E}^\mathbb{Q}\Bigl[\bigl(X_{(\ln u)/c}^2+1\bigr)^{v_0}\Bigr]\,du \nonumber \\
			&\leq d^{v_0}e^{-c v_0 T}\frac{(e^{cT}-1)^{v_0-1}}{c^{v_0}} \int_1^{e^{cT}} C\,du \leq \frac{C d^{v_0}}{c^{v_0}},
		\end{align}
	which gives the desired result.
\end{proof}

\subsectionfont{\scriptsize}
\subsection{Sensitivity with respect to \texorpdfstring{$\sigma$}{s}}

We evaluate the long-term sensitivity with respect to the perturbations of $\sigma.$ 

\begin{prop}\label{prop:KO_sigma_1}
	Under the  Kim--Omberg model in Eq.\eqref{eqn:KO_model},  the long-term sensitivity with respect to the parameter $\sigma$ is
	\[
			\lim_{T\rightarrow\infty}\frac{1}{T}\frac{\partial }{\partial \sigma}\ln v(\chi,T)=-\frac{\partial \lambda}{\partial \sigma}.
		\]
\end{prop}

\begin{proof}
	In the decomposition
	\[
			v(\chi,T)=e^{-\lambda T}\phi(\chi)\mathbb{E}^{\mathbb{Q}}\Bigl[\frac{1}{\phi({X}_T)}\, e^{\int_0^Tf(X_s,s;T)\,ds}\Bigr],
		\]
	we analyze the expectation term $\mathbb{E}^{\mathbb{Q}}[\frac{1}{\phi({X}_T)}\, e^{\int_0^Tf(X_s,s;T)\,ds}]$ by using the method in Section \ref{sec:vega}. Consider the Lamperti transformation
	\[
			\ell(x)=\int_\chi^x\frac{1}{\sigma}\,du=\frac{x-\chi}{\sigma}.
		\]
		and define
	\[
			\check{X}_t:=\ell(X_t)=\frac{X_t-\chi}{\sigma}
		\]
	as well as
	\begin{align*}
			F(\check{x}) &= -\frac{q\sigma_2^2}{2(1-q)}\Bigl(\bigl(B-\beta(T-t)\bigr)(\sigma \check{x}+\chi)+\bigl(C-\gamma(T-t)\bigr)\Bigr)^2,\\
			\Phi(\check{x}) &= e^{-\frac{1}{2}B\sigma^2\check{x}^2-(B\chi+C)\sigma \check{x}-\frac{1}{2}B\chi^2-C\chi}.
		\end{align*}
	Then $\check{X}$ satisfies the SDE
	\[
			d\check{X}_t=\biggl(\frac{1}{\sigma} \bigl( k\overline{m}-C\sigma^2-\frac{q\sigma_2^2 }{1-q}  \gamma(T-t)\bigr) -\Bigl(k+\frac{q\mu\sigma_1}{\varsigma}+B\sigma^2+\frac{q\sigma_2^2 }{1-q}\beta(T-t)\Bigr)\bigl(\check{X}_t+\frac{\chi}{\sigma}\bigr) \biggr)\,dt+dB_{t}.
		\]
 
	We want to analyze
	\[
			\frac{\partial}{\partial\sigma}\mathbb{E}^{\mathbb{Q}}\Bigl[\frac{1}{\Phi(\check{X}_T)}\, e^{\int_0^TF(\check{X}_s,s;T)\,ds}\Bigr].
		\]
	The perturbation parameter $\sigma$ is only involved with   the functional and the drift term of $\check{X}$, but not with the volatility term of $\check{X}.$ Thus, we can apply the same method used in Proposition \ref{prop:KO_k} to show 
	\[
			\lim_{T\rightarrow \infty}\frac{1}{T} \frac{\partial}{\partial\sigma}\mathbb{E}^{\mathbb{Q}}\Bigl[\frac{1}{\Phi(\check{X}_T)}\, e^{\int_0^TF(\check{X}_s,s;T)\,ds}\Bigr]=0.
		\]
	This gives the desired result.
\end{proof}

\sectionfont{\scriptsize}
\section{The Heston model}
\label{app:Heston}
  
This appendix investigates the Heston model presented in Section \ref{sec:Heston_model} and shows the assumptions made in the main part of the paper are satisfied in this model. Assumptions A\ref{assume:SDE_X} -- \ref{eqn:NFLVR} are well-known to be satisfied for the Heston model.

We first find the HJB equation and the ergodic  HJB equation. The functions $l$ and $h$ in Eq.\eqref{eqn:l_h} are
\begin{align*}
	l(\xi,x) & := -\frac{q}{2}(1-q) \Bigl(\frac{\mu^2x}{\varsigma^2}+\xi^2 \Bigr)\\ 
	h(\xi,x) & := k\overline{m} - \Bigl(k+\frac{q\mu \sigma_1}{\varsigma}\Bigr) x -q\xi\sigma_2\sqrt{x}.
\end{align*} 
The corresponding HJB equation \eqref{eqn:HJB} is
\begin{align*}
	v_t&=\frac{1}{2}\sigma^2xv_{xx} +\sup_{\xi\in\mathbb{R}}\left\{-\frac{q}{2}(1-q) \Bigl(\frac{\mu^2x}{\varsigma^2}+\xi^2 \Bigr)v + \biggl(k\overline{m}-\Bigr(k+\frac{q\mu \sigma_1}{\varsigma}\Bigr)x - q \xi \sigma_2 \sqrt{x}\biggr) v_x\right\}\\
	& =\frac{1}{2}\sigma^2xv_{xx} -\frac{q}{2}(1-q)\frac{\mu^2 }{\varsigma^2}xv + \biggl(k\overline{m}-\Bigl(k+\frac{q\mu \sigma_1}{\varsigma}\Bigr) x \biggr) v_x+\frac{q\sigma_2^2 x}{2(1-q)}\frac{v_x^2}{v} 
\end{align*}
with $v(x,0)=1$. Here, we used that the  supremum  of the above HJB equation is achieved at
\[
	\xi=-\frac{\sigma_2\sqrt{x}}{(1-q)} \frac{v_x(x,t)}{v(x,t)}.
\]
The solution to the HJB equation is $v(x,t)= e^{-\gamma(t)-\beta(t)x}$ with
\begin{align}
\label{eqn:alpha_beta_Hes}
\beta(t) & = q(1-q)\frac{\mu^2}{\varsigma^2}\frac{\sinh(\beta_2 t/2)}{\beta_2\cosh(\beta_2 t/2)+ \beta_1 \sinh(\beta_2 t/2)}, \nonumber\\
\gamma(t) & = k\overline{m}\int_0^tB(s)  \,ds,
\end{align}
where
\[
	\beta_1:=k+\frac{q\mu \sigma_1}{\varsigma}, \qquad  \beta_2:=\sqrt{\beta_1^2+ \frac{q ((1-q)\sigma_1^2+  \sigma_2^2)\mu^2}{\varsigma^2} }.
\]
Thus assumption A\ref{assume:diff_v} holds. The optimal control $\hat{\xi}$ is 
\begin{equation}
\label{eqn:Heston_xi_hat}
	\hat{\xi}(x,t;T)  =\frac{\sigma_2}{1-q} \beta(T-t)\sqrt{x}
\end{equation} 
With this optimizer, assumption A\ref{assume:struc_opt} is satisfied.

Now we shift our attention to the ergodic  HJB equation \eqref{eqn:EBE}. By direct calculation, we can see that the solution to the ergodic HJB equation 
\[
	- \lambda \phi=\frac{1}{2}\sigma^2x\phi_{xx} -\frac{q}{2}(1-q)\mu^2 x\phi + \bigl(k\overline{m}-(q\mu \sigma_1+k) x \bigr)\phi_x+\frac{q\sigma_2^2 x}{2(1-q)}\frac{\phi_x^2}{\phi}
\]
is given by $\phi(x)=e^{-Bx}$ with
\[
	B=\frac{\beta_2-\beta_1 }{\sigma_1^2+\frac{ \sigma_2^2   }{  1-q }}.
\]
It is easy to show that $\beta(t)\to B$ and  $\frac{\gamma(t)}{t}\to-\lambda$ as $t\to\infty,$ thus assumption A\ref{assume:long-term} holds. The ergodic optimal control  $\xi^*$ is
\[
	\xi^*(x)= \frac{\sigma_2}{1-q}B\sqrt{x}.
\]
 
For the rest of this section, we show that  assumptions A\ref{assume:P} -- A\ref{assume:Q} are satisfied.

\begin{prop}  \label{prop:marting_Hes}
	For the Heston model  A\ref{assume:P} holds, that is, the local martingale
	\begin{equation}
			\biggl(	\mathcal{E}\Bigl(-\frac{q\mu}{\varsigma}\int_0^\cdot \sqrt{X_s}\,dW_{1,s}-q\int_0^\cdot\hat{\xi}(X_s,s;T)\,dW_{2,s}\Bigr)_t\biggr)_{0\leq t\leq T} 
		\end{equation}  
	is a true martingale under the measure ${\mathbb{P}}$.
\end{prop}

\begin{proof}
	In order to show this is a true martingale, we use Theorem 8.1 in \cite{klebaner2014stochastic}. Recall that
	\[
			dX_t=k(\overline{m}-X_t)\,dt+\sigma_1 \sqrt{X_t} \,dW_{1,t}+\sigma_2 \sqrt{X_t} \,dW_{2,t}, \qquad X_0=\chi.
		\]
	Using the notions in \cite{klebaner2014stochastic}, we have
	\begin{align*}
			a_t(x) & = k(\overline{m}-x),\\
			b_t(x) & = \bigl(\sigma_1\sqrt{x},\sigma_2\sqrt{x}\bigr),\\
			\sigma_t(x) & = \Bigl(-\frac{q\mu}{\varsigma}\sqrt{x},-q\hat{\xi}(x,t;T)\Bigr),
		\end{align*}
	so that
	\begin{align*}
			\Vert \sigma_t(x) \Vert^2 & = q^2\Bigl(\frac{\mu^2x}{\varsigma^2}+\hat{\xi}^2(x,t;T)\Bigr),\\
			L_t(x) & = 2k(\overline{m}-x)x+\sigma^2x,\\
			\mathfrak{L}_t(x) & = 2x\biggl(k\overline{m}-\Bigl(k+\frac{q\mu\sigma_1}{\varsigma}\Bigr){x}-q\sigma_2\hat{\xi}(x,t;T)\sqrt{x}\biggr)+\sigma^2x.
	\end{align*}
	Using the expression of $\hat{\xi}$ in Eq.\eqref{eqn:Heston_xi_hat}, one can find a positive $r>\chi=X_0$ such that 
	\[
			\Vert \sigma_t(x) \Vert^2 + L_t(x)+\mathfrak{L}_t(x)\leq r\bigl(1+x^2\bigr).
		\] 
	This implies that the assumptions of Theorem 8.1 in \cite{klebaner2014stochastic} are met, and thus we obtain the desired result.
\end{proof}

Now the measure $\hat{\mathbb{P}}$ is well-defined by Eq.\eqref{eqn:hat_P_from_P_RN} and the $\hat{\mathbb{P}}$-dynamics of $X$ is
\begin{equation}
\label{eqn:hat_X_Hes}
	dX_t = \biggl(k\overline{m} - \Bigl(k+\frac{q\mu \sigma_1}{\varsigma}\Bigr) X_t-q \sigma_2\hat{\xi}(X_t,t;T)\sqrt{X_t}\biggr)\,dt +\sigma_1 \sqrt{X_t} \,d\hat{W}_{1,t}+\sigma_2 \sqrt{X_t} \,d\hat{W}_{2,t} , \qquad X_0=\chi.
\end{equation}

\begin{prop}
\label{prop:mart_condi_Hes}
	For the Heston model  A\ref{assume:P_hat} holds, that is, the local martingale
	\[
			\biggl(\mathcal{E}\Bigl(q\int_0^\cdot \hat{\xi} (X_s,s;T)-\xi^*(X_s)\,d\hat{W}_{2,s}\Bigr)_t\biggr)_{0\leq t\leq T}
		\]
	is a true martingale under the measure $\hat{\mathbb{P}}$.
\end{prop}

\begin{proof} 
	In order to show this is a true martingale, we use Theorem 8.1 in \cite{klebaner2014stochastic}. The proof is similar to the proof of Proposition \ref{prop:marting_Hes}, thus we only state the corresponding functions,
	\begin{align*}
			a_t(x) & = k\overline{m}-\Bigl(k+\frac{q\mu \sigma_1}{\varsigma}\Bigr) x - q\sigma_2\hat{\xi}(x,t;T)\sqrt{x},\\
			b_t(x) & = \bigl(\sigma_1\sqrt{x},\sigma_2\sqrt{x}\bigr),\\
			\sigma_t(x) &= \Bigl(0,q\bigl(\hat{\xi}(x,t;T)-\xi^*(x)\bigr)\Bigr),
		\end{align*}
	and it is straightforward to verify that the assumptions of Theorem 8.1 in \cite{klebaner2014stochastic} are met.
\end{proof}

Now the measure $\tilde{\mathbb{P}}$ is well-defined by Eq.\eqref{eqn:tilde_P}  and the $\tilde{\mathbb{P}}$-dynamics of $X$ is 
\[
	dX_t = \biggl(k\overline{m} - \Bigl(k+\frac{q\mu \sigma_1}{\varsigma}\Bigr) X_t-q\sigma_2{\xi}^*(X_t)\sqrt{X_t}\biggr)\,dt + \sigma_1 \sqrt{X_t} \,d\tilde{W}_{1,t}+\sigma_2 \sqrt{X_t} \,d\tilde{W}_{2,t} , \qquad X_0=\chi.
\]

\begin{prop}
	For the Heston model,  A\ref{assume:M}  holds, that is, the process
	\[
			M =\biggl(\mathcal{E}\Bigl(-\int_0^\cdot\sigma_{1}B\sqrt{X_s}\,d\tilde{W}_{1,s}-\int_0^\cdot\sigma_{2}B\sqrt{X_s}\,d\tilde{W}_{2,s} \Bigr)_t\biggr)_{0\le t\le T } 
	\]
	is a martingale under the measure $\tilde{\mathbb{P}}$.	 
\end{prop}

\begin{proof}
	In order to show this is a true martingale, we use Theorem 8.1 in \cite{klebaner2014stochastic}.The proof is similar to the proof of Proposition \ref{prop:marting_Hes}, thus we only state the corresponding functions, 
	\begin{align*}
			a_t(x) &= k\overline{m}-\Bigl(k+\frac{q\mu \sigma_1}{\varsigma}\Bigr) x-q\sigma_2{\xi}^*(x)\sqrt{x}, \\
			b_t(x) & = \bigl(\sigma_1\sqrt{x},\sigma_2\sqrt{x}\bigr)\\
			\sigma_t(x) & = \bigl(-\sigma_1B\sqrt{x},-\sigma_2B\sqrt{x}\bigr),
	\end{align*}
	and it is straightforward to verify that  the assumptions of Theorem 8.1 in \cite{klebaner2014stochastic} are met.
\end{proof}

Now the measure $\overline{\mathbb{P}}$ is well-defined by Eq.\eqref{eqn:overline_P} and the $\overline{\mathbb{P}}$-dynamics of $X$ is 
\begin{align*}
	dX_t & = \biggl(k\overline{m} - \Bigl(k+\frac{q\mu \sigma_1}{\varsigma}+\sigma^2B\Bigr) X_t-q\sigma_2{\xi}^*(X_t)\sqrt{X_t}\biggr)\,dt+\sigma_1 \sqrt{X_t} \,d\overline{W}_{1,t}+\sigma_2 \sqrt{X_t} \,d\overline{W}_{2,t}  \\
	& = \Biggl(k\overline{m}-\biggl(k+\frac{q\mu \sigma_1}{\varsigma}+\Bigl(\sigma_1^2+\frac{\sigma_2^2}{1-q}\Bigr)B\biggr)X_t \Biggr)\,dt+\sigma_1 \sqrt{X_t} \,d\overline{W}_{1,t}+\sigma_2 \sqrt{X_t} \,d\overline{W}_{2,t}  
\end{align*}
which is again the CIR process with re-parametrization. 

\begin{prop}\label{prop:Hes_mart_condi}
	For the Heston model A\ref{assume:Q} holds, that is, the local martingale
	\[
			\biggl(\mathcal{E}\Bigl(q\int_0^\cdot\xi^*(X_s)-\hat{\xi} (X_s,s;T)\,d\overline{W}_{2,s}\Bigr)_t\biggr)_{0\leq t\leq T}
		\]
	is a true martingale under the measure $\overline{\mathbb{P}}$.
\end{prop}

\begin{proof}
	In order to show this is a true martingale, we use Theorem 8.1 in \cite{klebaner2014stochastic}. The proof is similar to the proof of Proposition \ref{prop:marting_Hes}, thus we only state the corresponding functions,
	\begin{align*}
			a_t(x) & = k\overline{m}-\biggl(k+\frac{q\mu \sigma_1}{\varsigma}+\Bigl(\sigma_1^2+\frac{\sigma_2^2}{1-q}\Bigr)B\biggr)x, \\
			b_t(x) & = \bigl(\sigma_1\sqrt{x},\sigma_2\sqrt{x}\bigr),\\
			\sigma_t(x) & = \Bigl(0,q\bigl(\xi^*(x)-\hat{\xi} (x,t;T)\bigr)\Bigr),
	\end{align*}
	and it is straightforward to verify that  the assumptions of Theorem 8.1 in \cite{klebaner2014stochastic} are met.	 
\end{proof}

Now the measure $\mathbb{Q}$ is well-defined by Eq.\eqref{eqn:bar_P_Q}. The functions $f$ and $\kappa$ in Eq.\eqref{eqn:kappa} are
\begin{equation}
\label{eqn:f_Heston}
	f(x,t;T) = -\frac{q\sigma_2^2x}{2(1-q)}\bigl(B-\beta(T-t)\bigr)^2
\end{equation}
and 
\[
	\kappa(x,t;T) = k\overline{m} - \Bigl(k+\frac{q\mu \sigma_1}{\varsigma}+\sigma^2B\Bigr)x - q\sigma_2 \hat{\xi}(x,t;T) \sqrt{x} = k\overline{m}-\Bigl(k+\frac{q\mu \sigma_1}{\varsigma} + \sigma^2B + \frac{q\sigma_2^2}{1-q} \beta(T-t)\Bigr)x.
\]
Finally, the $\mathbb{Q}$-dynamics of $X$ is
\begin{equation}
\label{eqn:Q_dynamics_Heston}
	dX_t = \kappa(X_t,t;T)\,dt+\sigma_{1}\sqrt{X_t}\,dB_{1,t}+\sigma_{2}\sqrt{X_t}\,dB_{2,t}, \qquad 0\leq t\leq T.
\end{equation}

\subsectionfont{\scriptsize}
\subsection{Integrability condition}

In the following we prove integrability conditions, which will be needed in the analysis in the next sections.

\begin{lemma}
\label{lem:U_Hes}
	Under the measure $\mathbb{Q},$ consider two processes $U$ and $L$ defined as the solutions of SDEs
	\begin{align*} 
		\label{eqn:U_hes}
			dU_t & = \bigl(k\overline{m}- \upsilon_U U_t\bigr)\,dt + \sigma \sqrt{U_t}\,dB_{t}, \qquad U_0=x,\\
			dL_t & = \bigl(k\overline{m}- \upsilon_L L_t\bigr)\,dt+\sigma \sqrt{L_t}\,dB_{t}, \qquad L_0=x,
		\end{align*}
	where $\upsilon_U:=k+\frac{q\mu\sigma_1}{\varsigma} +\sigma^2B$ and $\upsilon_L:=k+\frac{q\mu\sigma_1}{\varsigma} +(\sigma_1^2+\frac{\sigma_2^2}{1-q}) B.$ Then
	\[
			\mathbb{Q}\bigl[L_t\leq X_t\leq U_t \textnormal{ for all } 0\leq t\leq T\bigr]=1.
		\]
\end{lemma}

\begin{proof}
	Under the measure $\mathbb{Q},$ the process $X$ satisfies
	\begin{equation} 
			dX_t = \biggl(k\overline{m}-\Bigl(\upsilon_U+\frac{q\sigma_2^2}{1-q} \beta(T-t)\Bigr)X_t\biggr)\,dt+\sigma\sqrt{X_t}\,dB_{t}, \qquad \leq t\leq T.
		\end{equation} 
	Using $0<\beta(\cdot)<B$, we have
	\[
			k\overline{m} -\upsilon_Lx\leq k\overline{m}-\Bigl(\upsilon_U+\frac{q\sigma_2^2}{1-q} \beta(T-t) \Bigr)x \leq k\overline{m}- \upsilon_Ux.
		\]
	Proposition 5.2.18 in \cite{karatzas1998brownian} gives
	\[
			\mathbb{Q}\bigl[L_t\leq X_t\leq U_t \textnormal{ for all } 0\leq t\leq T\bigr]=1.
		\]
\end{proof}

\begin{lemma}
\label{lem:large_mean_rev_rate_Heston}
	There are a number $u>1$ and an open neighborhood $I_\chi$ of $\chi$ such that 
	\[
		\Gamma_u(x,T):=\mathbb{E}^{\mathbb{Q}}\Bigl[\frac{1}{\phi^{u}({X}_T)}\, e^{u\int_0^Tf(X_s,s;T)\,ds} \, \Big\vert \,  X_0=x\Bigr]
	\]
	is uniformly bounded on $I_\chi\times [0,\infty).$
\end{lemma}

\begin{proof}
	Since the function  $f$  is nonpositive  as one can see in Eq.\eqref{eqn:f_Heston}, it suffices to show that there is a number $u>1$  such that
	\[
			\mathbb{E}^{\mathbb{Q}}\Bigl[\frac{1}{\phi^{u}({X}_T)}\Bigr]=\mathbb{E}^{\mathbb{Q}}\Bigl[\frac{1}{\phi^{u}({X}_T)} \, \Big\vert \, X_0=x\Bigr]
		\]
	is uniformly bounded in $(x,T)$ on $(\frac{\chi}{2},\frac{3\chi}{2})\times [0,\infty).$ Recall that the process $U$ in Lemma \ref{lem:U_Hes} satisfies 
	\[
			\mathbb{Q}\bigl[X_t\leq U_t \textnormal{ for all } 0\leq t\leq T\bigr]=1.
		\] 
	Then for $u>1$
	\[
			\mathbb{E}^{\mathbb{Q}}\Bigl[\frac{1}{\phi^{u}({X}_T)}\Bigr]=\mathbb{E}^{\mathbb{Q}}[e^{uBX_T} \, \vert \, X_0=x]\leq \mathbb{E}^{\mathbb{Q}}[e^{uBU_T} \, \vert \, U_0=x].
		\]
	Since $U$ is the a CIR process, it is known that the moment generating function is
	\[
			\mathbb{E}^{\mathbb{Q}}[e^{uBU_T} , \vert \, U_0=x]=\Bigl(\frac{h_T}{h_T-uB}\Bigr)^{\frac{2k\overline{m}}{\sigma^2}}\exp\Bigl(\frac{uBe^{-\upsilon_U T}h_Tx}{h_T-uB}\Bigr) 
		\]
	where  
	\[
			h_T=\frac{2\upsilon_U}{\sigma^2(1-e^{-\upsilon_U T})}.
		\]
	Using $k+\frac{q\mu\sigma_1}{\varsigma} +\sigma^2B=\upsilon_U,$ observe that$2B+\frac{2k}{\sigma^2}+\frac{2q\mu\sigma_1}{\sigma^2\varsigma}=\frac{2\upsilon_U}{\sigma^2}<h_T.$ From this explicit expression, it is easy to check that for 
	\begin{equation}
		\label{eqn:Hes_u_bound}
			1<u< 2+\frac{2}{\sigma^2B}\Bigl(k+\frac{q\mu\sigma_1}{\varsigma }\Bigr),
		\end{equation}
	the expectation $\mathbb{E}^{\mathbb{Q}}[e^{uBU_T} \, \vert \, U_0=x]$ is uniformly bounded in $(x,T)$ on $\bigl(\frac{\chi}{2},\frac{3\chi}{2}\bigr)\times [0,\infty).$ This completes the proof.
\end{proof}

\subsectionfont{\scriptsize}
\subsection{Sensitivity with respect to the initial volatility}

\begin{prop}
	Under the  Heston model, the long-term sensitivity with respect to the initial value of the volatility is
	\[
			\lim_{T\rightarrow\infty}\frac{\partial }{\partial \chi}\ln v(\chi,T)=\frac{\phi'(\chi)}{\phi(\chi)}=-B.
		\]
\end{prop}

\begin{proof}
	By Theorem \ref{thm:delta}, it suffices to prove that the expectation $\mathbb{E}^{\mathbb{Q}}\bigl[\frac{1}{\phi(X_T)}\, e^{\int_0^Tf(X_s,s;T)\,ds} \, \big\vert \, X_0=x\bigr]$ is continuously differentiable in $x,$ and 
	\begin{equation}
		\label{eqn:Hes_delta_go_zero}
			\frac{\partial }{\partial x}\mathbb{E}^{\mathbb{Q}}\Bigl[\frac{1}{\phi(X_T)}\, e^{\int_0^Tf(X_s,s;T)\,ds} \, \Big\vert \, X_0=x\Bigl]
		\end{equation} 
	converges to zero as $T\to\infty.$ To prove this, we apply  Proposition  \ref{prop:delta}. Condition (i) of this proposition  was proved in Lemma \ref{lem:large_mean_rev_rate_Heston}. For (ii), observe that  $\frac{\phi'(X_T)}{\phi(X_T)}=-B$ is a constant, thus this condition holds trivially for any $v>1.$ 

	We now prove that (iii) holds: for any $w>1$ the expectation $\mathbb{E}^\mathbb{Q}\bigl[\vert Y_{T;T}\vert^w \, \vert \, X_0=x\bigr]$ is uniformly bounded in $x$ on  $\bigl(\frac{\chi}{2},\frac{3\chi}{2}\bigr)$ and converges to zero as $T\to\infty.$ From Eq.\eqref{eqn:Q_dynamics_Heston}, the process $Y_t=Y_{t;T}$ satisfies
	\[
			dY_t=-\Bigl(k+\frac{q\mu\sigma_1}{\varsigma}+\sigma^2B+\frac{q\sigma_2^2}{1-q}\beta(T-t)\Bigr)Y_t \,dt+ \frac{\sigma}{2\sqrt{X_t}}Y_t\,dB_{t}, \qquad 0\leq t\leq T.
		\]
	By the It\^{o} formula, we get 
	\[
			X_t^{-\frac{1}{2}}Y_t=x^{-\frac{1}{2}}e^{ -\frac{1}{2}(k+\frac{q\mu\sigma_1}{\varsigma}+\sigma^2B )t+\int_0^t ((-\frac{1}{2}k\overline{m}+\frac{1}{8}\sigma^2)\frac{1}{X_t}- \frac{1}{2} \frac{q\sigma_2^2}{1-q}\beta(T-s))\,ds}.
		\]
	Since $2k\overline{m}\geq \sigma^2$ and $\beta(\cdot)\ge0,$ the integrand in the exponent on the right-hand side is negative. It follows that
	\[
			X_t^{-\frac{1}{2}}Y_t\leq x^{-\frac{1}{2}}e^{-\frac{1}{2}(k+\frac{q\mu\sigma_1}{\varsigma}+\sigma^2B )t}.
		\]
	Then for any $w>1,$ 
	\begin{equation}
		\label{eqn:Heston_first_vari}
			\mathbb{E}^\mathbb{Q} \vert Y_t \vert^w=\mathbb{E}^\mathbb{Q} \Bigl \vert X_t^{\frac{1}{2}}X_t^{-\frac{1}{2}}Y_t \Bigr\vert^w\leq x^{-\frac{1}{2}w}e^{-\frac{1}{2}w(k+\frac{q\mu\sigma_1}{\varsigma}+\sigma^2B )t}\,\mathbb{E}^\mathbb{Q}\Bigl[X_t^{\frac{1}{2}w}\Bigr].
		\end{equation} 
	To obtain (iii), we consider the expectation $\mathbb{E}^\mathbb{Q}\bigl[X_T^{\frac{1}{2}w}\bigr]=\mathbb{E}^\mathbb{Q}\bigl[X_T^{\frac{1}{2}w} \, \big\vert \, X_0=x\bigr].$ Recall that the process $U$ in Lemma \ref{lem:U_Hes} satisfies $\mathbb{Q}\bigl[X_t\leq U_t \textnormal{ for all } 0\leq t\leq T\bigr]=1.$ Thus, 
	\begin{equation}
		\label{eqn:Heston_comparison}
			\mathbb{E}^\mathbb{Q}\Bigl[X_T^{\frac{1}{2}w} \, \Big\vert \,  X_0=x\Bigr] \leq \mathbb{E}^\mathbb{Q}\Bigl[U_T^{\frac{1}{2}w} \, \Big\vert \, X_0=x\Bigr].
			\end{equation} 
	On the other hand, since $U$ is a CIR process, for any $w>1,$ the expectation on the right hand side is uniformly bounded in $(x,T)$ on $\bigl(\frac{\chi}{2},\frac{3\chi}{2}\bigr) \times [0,\infty).$ Eq.\eqref{eqn:Heston_first_vari} implies that  the expectation $\mathbb{E}^\mathbb{Q}\vert Y_{T}\vert^w$ is uniformly bounded in $x$ on $\bigl(\frac{\chi}{2},\frac{3\chi}{2}\bigr)$ and converges to zero as $T\to\infty.$

	We now show that (iv) holds for any $m>1.$ It is easy to show that  there is a positive constant $c$ such that
	\[
			\vert f_x(x,t;T) \vert =\frac{q\sigma_2^2}{2(1-q)} \bigl\vert B-\beta(T-t) \bigr\vert^2\leq c e^{-\beta_2(T-t)}.
		\]
	For convenience, we define $\delta:=\frac{1}{2}(k+\frac{q\mu\sigma_1}{\varsigma}+\sigma^2B).$ By Eq.\eqref{eqn:Heston_first_vari} and Eq.\eqref{eqn:Heston_comparison}, it follows that
	\[
			\mathbb{E}^\mathbb{Q}\vert Y_{t;T}\vert^m \leq b_mx^{-\frac{1}{2}m}e^{-\delta mt}
		\]
	for a positive constant $b_m$ which dominates $	\mathbb{E}^\mathbb{Q}\bigl[X_t^{\frac{1}{2}m} \, \big\vert \, X_0=x\bigr]$ on $\bigl(\frac{\chi}{2},\frac{3\chi}{2}\bigr)\times[0,\infty).$ By the Jensen inequality, we have
	\[
			\mathbb{E}^{\mathbb{Q}}\biggl[\Bigl(\int_0^T \bigl\vert f_x(X_s,s;T)Y_{s;T} \bigr\vert \,ds\Bigr)^{m} \biggr] \leq \frac{c^mb_mx^{-\frac{1}{2}m}}{\beta_2-\delta} T^{m-1} \bigl(e^{-\delta mT}-e^{-\beta_2 mT}\bigr). 
		\]
	The right-hand side is uniformly bounded in $x$ on $(\frac{\chi}{2},\frac{3\chi}{2})$ for each $T\ge0$ and converges to zero as $T\to\infty$ for each $x\in \bigl(\frac{\chi}{2},\frac{3\chi}{2}\bigr).$ This proves (iv). Finally, conditions (ii), (iii), (iv) in Proposition \ref{prop:delta} hold true for arbitrary $v,w,m>1,$ and (i) holds for some $u>1,$ so we obtain the desired result.
\end{proof}

\subsectionfont{\scriptsize}
\subsection{Sensitivities with respect to \texorpdfstring{$k,$ $\overline{m},$ $\mu,$ $\varsigma$}{kmms} and \texorpdfstring{$\rho$}{r}}
\label{sec:sen_Hes_k_b}
 
We calculate the sensitivity with respect to the parameter $k.$ Those with respect to the parameters  $\overline{m},$ $\mu,$ $\varsigma$ and $\rho$ can be calculated in a similar way. The five functions in B\ref {bassume:perturb} and B\ref{bassume:HS_eps} are 
\[
	m_\epsilon(x)=(k+\epsilon) (\overline{m} - x), \qquad \sigma_{1,\epsilon}(x)=\sigma_1\sqrt{x}, \qquad \sigma_{2,\epsilon}(x)=\sigma_2\sqrt{x}, \qquad b_\epsilon(x)=\mu x, \qquad \varsigma_\epsilon(x)=\varsigma\sqrt{x}
\]
and it is easy to check that they satisfy assumptions B\ref {bassume:perturb} and B\ref{bassume:HS_eps}. Observe that 
\[
	\frac{\partial }{\partial \epsilon}\Big\vert_{\epsilon=0}\ln v_\epsilon(\chi,T)=\frac{\partial }{\partial k}\ln v_0(\chi,T)=\frac{\partial }{\partial k}\ln v(\chi,T),
\]
thus for the rest of this section we use $\frac{\partial }{\partial k}$ instead of $\frac{\partial }{\partial \epsilon}\vert_{\epsilon=0}.$

\begin{prop} \label{prop:Hes_k}
	Under the Heston model, the long-term sensitivity with respect to the parameter $k$ is
	\[
			\lim_{T\rightarrow\infty}\frac{1}{T}\frac{\partial }{\partial k}\ln v(\chi,T)=-\frac{\partial\lambda}{\partial k}.
		\]
\end{prop}

\begin{proof}
	To prove this equality, we use Theorem \ref{thm:total_chain}. Condition (i) in Theorem \ref{thm:total_chain} is satisfied trivially. We prove (iii) in Theorem \ref{thm:total_chain} first because some techniques used for (iii) are also used in the proof of (ii). For condition (iii) in Theorem \ref{thm:total_chain}, we apply Theorem \ref{thm:rho}. It can be easily checked that 
	\[
			\Big\vert \frac{1}{\sigma\sqrt{x}}\frac{\partial}{\partial k}\kappa (x,t;T)\Big\vert \leq  c\Bigl(\sqrt{x}+\frac{1}{\sqrt{x}}\Bigr), \qquad x>0
		\]
	for a positive constant $c$ independent of $t,$ $T$ and $x>0.$ By choosing sufficiently large $c,$ we can achieve that $\hat{g}^2(x,t;T) \leq c(x+1/x)$ holds true for $\hat{g}$ defined in Eq.\eqref{eqn:hats}.

	Then, (i) in Theorem \ref{thm:rho} can be proven as follows. Recall the processes $U$ and $L$ from Lemma \ref{lem:U_Hes}. Since $\mathbb{Q}[L_t\leq X_t\leq U_t \textnormal{ for all } 0\leq t\le T]=1,$ we have  
	\[
			\mathbb{E}^{\mathbb{Q}}\bigl[e^{\epsilon_0\int_0^T \hat{g}^2(X_s,s;T)\,ds}\bigr] \leq \mathbb{E}^{\mathbb{Q}}\bigl[e^{\epsilon_0c\int_0^T (X_s+1/X_s)\,ds}\bigr] \leq \mathbb{E}^{\mathbb{Q}}\bigl[e^{\epsilon_0c \int_0^T (U_s+1/L_s)\,ds} \bigr] \leq \bigl(\mathbb{E}^{\mathbb{Q}}\bigl[e^{2\epsilon_0c\int_0^T U_s\,ds}\bigr]\bigr)^{\frac{1}{2}} \bigl(\mathbb{E}^{\mathbb{Q}} \bigl[e^{2\epsilon_0c \int_0^T 1/L_s\,ds}\bigr]\bigr)^{\frac{1}{2}}.
		\]
	Since $U$ is a CIR process, for given $T\ge0$ one can find $\epsilon_0>0$ such that $\mathbb{E}^{\mathbb{Q}}\bigl[e^{2\epsilon_0c\int_0^T U_s\,ds}\bigr]$  is finite. In addition, since $L$ is also a CIR process satisfying the Feller condition, applying  Proposition D.2 in  \cite{park2015sensitivity}, one can find $\epsilon_0$ such that $\mathbb{E}^{\mathbb{Q}}\bigl[e^{2\epsilon_0c\int_0^T 1/L_s\,ds}\bigr]$ is finite. This gives (i) in Theorem \ref{thm:rho}.
 
	 Now we prove (ii) in Theorem \ref{thm:rho} with $v=2.$ It suffices to show that there is a positive constant $c_0$ such that for all $T\ge0$
	\[
			\mathbb{E}^\mathbb{Q}\biggl[\int_0^T \Bigl(X_s+\frac{1}{X_s}\Bigr)\,ds\biggr]\leq c_0 T.
		\]
	Using the processes $U$ and $L$ in Lemma \ref{lem:U_Hes}, observe that
	\[
			\mathbb{E}^\mathbb{Q}\biggl[\int_0^T \Bigl(X_s+\frac{1}{X_s}\Bigr)\,ds\biggr] \leq  \mathbb{E}^\mathbb{Q}\biggl[\int_0^T \Bigl(U_s+\frac{1}{L_s}\Bigr)\,ds\biggr] = \int_0^T \mathbb{E}^\mathbb{Q}\bigl[U_s\bigr]+\mathbb{E}^\mathbb{Q}\bigl[L_s^{-1}\bigr]\,ds.
		\]
	Since $U$ and $L$ are  CIR processes satisfying the Feller condition, there is a positive constant $c_0$ such that for all $s\ge0$ 
	\[
			\mathbb{E}^\mathbb{Q}\bigl[U_s\bigr]+\mathbb{E}^\mathbb{Q}\bigl[L_s^{-1}\bigr]\leq c_0.
		\]
	This gives the desired result.

	For (iii) in Theorem \ref{thm:rho}, we observe that for $v=2$ and $\epsilon_1=1$ 
	\[
			\mathbb{E}^{\mathbb{Q}}\biggl[ \int_0^T\hat{g}^{v+\epsilon_1}(X_s,s;T) \, ds\biggr] \leq  c^{3/2} \int_0^T\mathbb{E}^{\mathbb{Q}}\biggl[\Bigl(X_s+\frac{1}{X_s}\Bigr)^{3/2}\biggr] \, ds\leq c' \int_0^T\mathbb{E}^{\mathbb{Q}}\bigl[U_s^{3/2}\bigr]+\mathbb{E}^{\mathbb{Q}}\bigl[L_s^{-3/2}\bigr] \, ds.
		\]
	Since $U$ is a CIR process, it is well known that $\mathbb{E}^{\mathbb{Q}}\bigl[U_s^{3/2}\bigr]$ is uniformly bounded in $s$ on $[0,\infty).$ In addition, for a CIR process $L$ satisfying the Feller condition, we have
	\[
			\sup_{0\leq s\leq T}\mathbb{E}^{\mathbb{Q}} \bigl[L_s^{-3/2}\bigr]<\infty
		\]
	by Eq.(3.1) in \cite{dereich2011euler}. This gives (iii) in Theorem \ref{thm:rho}.

	For (iv) in Theorem \ref{thm:rho}, we want to show that for $u=2$ the expectation
	\[
			\mathbb{E}^{\mathbb{Q}}\Bigl[\frac{1}{\hat{\phi}^{u}( {X}_T)}\, e^{u\int_0^T\hat{f}(X_s,s;T)\,ds}\Bigr]
		\]
	is uniformly bounded in $T$ on $[0,\infty).$ We use notation $B(k)$ to emphasize the dependence of $k$ on the constant $B.$ From Eq.\eqref{eqn:Hes_u_bound}, we know that for $u_0$ with $2<u_0<2+\frac{2}{\sigma^2B}(k+\frac{q\mu\sigma_1}{\varsigma})$ the expectation $\mathbb{E}^{\mathbb{Q}}[e^{u_0B(k)X_T}]$ is uniformly bounded in $T$ on $[0,\infty).$ Since the maps $k\mapsto B(k)$ is continuous and $\frac{u_0}{2}>1,$ by choosing a smaller interval $I$ if necessary, it follows that 
	\[
			\sup_{\epsilon\in I}B(k+\epsilon) \leq  \frac{u_0}{2} B(k).
		\]
	Then
	\begin{equation}
		\label{eqn:Hes_hat_phi}
			\hat{\phi}(x)=\inf_{\epsilon\in I}e^{-B(k+\epsilon)x} \geq e^{-\frac{u_0}{2}B(k)x}.
		\end{equation}
	Thus
	\begin{equation}
		\label{eqn:Hes_hat_phi_bound}
			\mathbb{E}^{\mathbb{Q}}\Bigl[\frac{1}{\hat{\phi}^{2}( {X}_T)}\, e^{2\int_0^T\hat{f}(X_s,s;T)\,ds}\Bigr]\leq \mathbb{E}^{\mathbb{Q}}\Bigl[\frac{1}{\hat{\phi}^{2}( {X}_T)}\Bigr] \leq   \mathbb{E}^{\mathbb{Q}}\bigl[e^{u_0B(k)X_T}\bigr],
		\end{equation}
	where for the first inequality we used $\hat{f}\le 0.$ Since the right-hand side is  uniformly bounded in $T$ on $[0,\infty),$ we obtain the desired result. We have now shown all conditions in  Theorem \ref{thm:rho} and thus condition (iii) in Theorem \ref{thm:total_chain} holds true. For condition (ii) in Theorem \ref{thm:total_chain}, we first calculate the partial derivative with respect to the variable  $k$ in $\phi$ and $f$ but not in $X=(X_t)_{t\ge0}.$ To be precise, we use the notation $\phi(x;k)$ and $f(x,t;T;k)$ to emphasize the dependence of $k.$ We want to analyze 
	\[
			w_{\eta,\epsilon}(\chi,T)=\mathbb{E}^{\mathbb{Q}_\epsilon}\Bigl[\frac{1}{\phi(X_T^\epsilon;k+\eta)}\, e^{\int_0^Tf(X_s^\epsilon,s;T;k+\eta)\,ds} \Bigr]
		\]
	where the $\mathbb{Q}^\epsilon$-dynamics of $X_t^\epsilon$ satisfies Eq.\eqref{eqn:Q_dynamics_Heston} with $k$ replaced by $k+\epsilon.$ The equality
	\[
			\frac{\partial }{\partial \eta} \mathbb{E}^\mathbb{Q}\Bigl[\frac{1}{\phi(X_T^\epsilon;k+\eta)}e^{\int_0^Tf(X_s^\epsilon,s;T;k+\eta)\,ds}\Bigr]= \mathbb{E}^\mathbb{Q}\biggl[\frac{\partial }{\partial \eta}\Bigl(l\frac{1}{\phi(X_T^\epsilon;k+\eta)}e^{\int_0^Tf(X_s^\epsilon,s;T;k+\eta)\,ds}\Bigr)\biggr]
		\]
	and the continuity of this partial derivative in $(\eta,\epsilon)$ on $I^2$ are obtained from Proposition \ref{prop:condi_2} with $g(x,t;T)$ and $G_T$  given below. Observe that
	\begin{equation}  
			\frac{\partial f}{\partial k}(x,t;T;k)= -\frac{q\sigma_2^2}{1-q}(B-\beta(T-t))\Bigl(\frac{\partial B}{\partial k}-\frac{\partial\beta}{\partial k}(T-t)\Bigr)x.
		\end{equation}
	We use the notation $\beta(T-t;k)$ to emphasize the dependence of $k.$ For a given small open interval $I,$ since $B(k+\eta),\frac{\partial B}{\partial k}(k+\eta)$ are continuous in $\eta$  on $\overline{I}$ and $\beta(T-t;k+\eta),\frac{\partial\beta}{\partial k}(T-t;k+\eta)$ are continuous in $(\eta,t)$ on $\overline{I}\times[0,T],$ one can find a positive constant $b_1$ such that
	for all $(\eta,t)\in \overline{I}\times [0,T]$
	\[
			\biggl\vert \frac{\partial f}{\partial \eta}(x,t;T;k+\eta)\biggr\vert \leq  b_1x =:g(t,x;T)\quad \textnormal{for}\enspace x>0.
		\]
	With this function $g,$ condition (i) in Proposition \ref{prop:condi_2} is trivially satisfied.

	For condition (ii) in Proposition \ref{prop:condi_2}, choose a positive constant $b_2$ such that for all $\eta\in \overline{I}$
	\[
			\Bigl\vert \frac{\partial B}{\partial  \eta}(k+\eta)\Bigr\vert \leq b_2.
		\]
	Using the function $\hat{\phi}$ in Eq.\eqref{eqn:Hes_hat_phi}, we define 
	\[
			G_T:=\frac{b_2X_T}{\hat{\phi}(X_T)} +\frac{1}{\hat{\phi}(X_T)}\int_0^Tb_1X_s\,ds.
		\]
	Then for all $(\eta,t)\in \overline{I}\times [0,T]$ it follows that
	\begin{equation}
			\frac{1}{\phi^2(X_T;k+\eta)}\biggl\vert\frac{\partial \phi }{\partial \eta}(X_T;k+\eta)\biggr\vert  +\frac{1}{\phi(X_T;k+\eta)}\, \int_0^T\biggl\vert \frac{\partial f}{\partial \eta}(X_s,s;T;k+\eta)\biggr\vert \,ds\leq G_T
		\end{equation}
	by using that $\hat{\phi}(x)=\inf_{\eta\in I}\phi(x;k+\eta)$ and
	\[
		\Bigl\vert \frac{\partial \phi }{\partial \eta}(x;k+\eta) \Bigr\vert = \Bigl\vert \frac{\partial B}{\partial  \eta}(k+\eta)\Bigr\vert x\phi(x;k+\eta)\leq  b_2x\phi(x;k+\eta)
		\]
	for $x>0.$ We claim $\mathbb{E}^\mathbb{Q}\bigl[G_T^{3/2}\bigr]<\infty,$ which implies condition (ii) in Proposition \ref{prop:condi_2}. Using that $\frac{1}{4/3}+\frac{1}{4}=1$, it follows that
	\[
			\mathbb{E}^\mathbb{Q}\Bigl[\frac{1}{\hat{\phi}^{3/2}(X_T)}\bigl(b_2X_T\bigr)^{3/2}\Bigr] \leq \biggl( \mathbb{E}^\mathbb{Q}\Bigl[\frac{1}{\hat{\phi}^{2}(X_T)}\Bigr] \biggr)^{3/4} \biggl(\mathbb{E}^\mathbb{Q}\Bigl[\bigl(b_2X_T\bigr)^{6}\Bigr]\biggr)^{1/4}.
		\]
	The first expectation on the right-hand side is finite by  Eq.\eqref{eqn:Hes_hat_phi_bound}. For the second  expectation observe that $\mathbb{E}^\mathbb{Q}[X_T^{6}]\leq \mathbb{E}^\mathbb{Q}[U_T^{6}]$ for the process $U$ in Lemma \ref{lem:U_Hes}. Since $U$ is a CIR process, the expectation $\mathbb{E}^\mathbb{Q}[U_T^{6}]$ is finite. In a similar way, we have
	\begin{align*}
			\mathbb{E}^\mathbb{Q}\biggl[\frac{1}{\hat{\phi}^{3/2}(X_T)}\Bigl(b_1\int_0^TX_s\,ds\Bigr)^{3/2}\biggr] & \leq \biggl( \mathbb{E}^\mathbb{Q}\Bigl[\frac{1}{\hat{\phi}^{2}(X_T)}\Bigr] \biggr)^{3/4} \biggl(\mathbb{E}^\mathbb{Q}\biggl[\Bigl(b_1\int_0^TX_s\,ds\Bigr)^{6}\biggr]\biggr)^{1/4} \\
			& \leq b_1^{3/2}T^{3/2} \biggl( \mathbb{E}^\mathbb{Q}\Bigl[\frac{1}{\hat{\phi}^{2}(X_T)}\Bigr] \biggr)^{3/4}\biggl(\mathbb{E}^\mathbb{Q}\biggl[\Bigl(\frac{1}{T}\int_0^TX_s\,ds\Bigr)^6\biggr]\biggr)^{1/4}  \\
			& \leq  b_1^{3/2}T^{5/4} \biggl( \mathbb{E}^\mathbb{Q}\Bigl[\frac{1}{\hat{\phi}^{2}(X_T)}\Bigr] \biggr)^{3/4}\biggl(\mathbb{E}^\mathbb{Q}\Bigl[\int_0^TX_s^6\,ds\Bigr]\biggr)^{1/4} \\
			& \leq b_1^{3/2}T^{5/4} \biggl( \mathbb{E}^\mathbb{Q}\Bigl[\frac{1}{\hat{\phi}^{2}(X_T)}\Bigr] \biggr)^{3/4}\biggl(\int_0^T\mathbb{E}^\mathbb{Q}[X_s^6]\,ds\biggr)^{1/4}.
		\end{align*}
	Since $\mathbb{E}^\mathbb{Q}[X_s^{6}]\leq \mathbb{E}^\mathbb{Q}[U_s^{6}]$ and the expectation $\mathbb{E}^\mathbb{Q}[U_s^{6}]$ is uniformly bounded in $s$ on $[0,T],$ the right-hand side is  finite. Hence $\mathbb{E}^\mathbb{Q}[G_T^{3/2}]<\infty.$  The convergence 
	\[
			\lim_{T\rightarrow \infty}\frac{1}{T}\frac{\partial }{\partial \eta}\Big\vert_{\eta=0}\mathbb{E}^\mathbb{Q}  \Bigl[\frac{1}{\phi(X_T;k+\eta)}e^{\int_0^Tf(X_s,s;T;k+\eta)\,ds}\Bigr] = 0
		\]
	can be shown as follows. Using $f\leq0,$ the partial derivative with respect to $\eta$ satisfies
	\begin{align*}
			&\phantom{==} \biggl\vert \frac{\partial }{\partial \eta}\Big\vert_{\eta=0}\Bigl(\frac{1}{\phi(X_T;k+\eta)}e^{\int_0^Tf(X_s,s;T;k+\eta)\,ds}\Bigr)\biggr\vert \\
			& \leq e^{BX_T+\int_0^Tf(X_s,s;T;k)\,ds}\Bigl\vert X_T\frac{\partial B}{\partial k}\Bigr\vert + e^{BX_T+\int_0^Tf(X_s,s;T;k)\,ds}\biggl\vert \int_0^T\frac{\partial f}{\partial \eta}\Big\vert_{\eta=0}(X_s,s;T;k+\eta)\,ds\biggr\vert \\
			& \leq e^{BX_T}\Bigl\vert X_T\frac{\partial B}{\partial k}\Bigr\vert + e^{BX_T}\biggl\vert \int_0^T\frac{\partial f}{\partial \eta}\Big\vert_{\eta=0}(X_s,s;T;k+\eta)\,ds\biggr\vert.
		\end{align*}
	By the triangle inequality and the Cauchy--Schwarz inequality, it follows that 
	\begin{align*}
			& \phantom{=:} \mathbb{E}^\mathbb{Q}\biggl\vert \frac{\partial }{\partial \eta}\Big\vert_{\eta=0}\Bigl(\frac{1}{\phi(X_T;k+\eta)}e^{\int_0^Tf(X_s,s;T;k+\eta)\,ds}\Bigr)\biggr\vert\\
			& \leq \Bigl(\mathbb{E}^\mathbb{Q}e^{ 2BX_T}\Bigr)^{1/2}\Bigl(\mathbb{E}^\mathbb{Q}\Bigl\vert X_T\frac{\partial B}{\partial k}\Bigr\vert^{2}\Bigr)^{1/2} + \Bigl(\mathbb{E}^\mathbb{Q}e^{2BX_T}\Bigr)^{1/2}\biggl(\mathbb{E}^\mathbb{Q}\biggl\vert \int_0^T\frac{\partial f}{\partial \eta}\Big\vert_{\eta=0}(X_s,s;T;k+\eta)\,ds\biggr\vert^{2}\biggr)^{1/2}.
		\end{align*}
	By Eq.\eqref{eqn:Hes_hat_phi_bound}, the expectation $\mathbb{E}^\mathbb{Q}e^{ 2BX_T}$ is  uniformly bounded in $T$ on $[0,\infty).$ It is easy to show that $\mathbb{E}^\mathbb{Q} \vert X_T\frac{\partial B}{\partial k} \vert^{2}$ is also uniformly bounded in $T$ on $[0,\infty).$ Now, we show that the expectation $\mathbb{E}^\mathbb{Q}\bigl\vert \int_0^T\frac{\partial f}{\partial \eta}\vert_{\eta=0}(X_s,s;T;k+\eta)\,ds\bigr\vert^{2}$ is uniformly bounded in $T$. By direct calculation, one can choose positive constants $c$ and $d,$ which are independent of $s$ and $T$  but are dependent of $k,$ such that
	\[
			\biggl\vert \frac{\partial f }{\partial \eta}\Big\vert_{\eta=0}(x,s;T;k+\eta)\biggr\vert \leq de^{-c(T-s)}x.
		\]
	By the same change of variable $u=e^{c s}$ as in Eq.\eqref{eqn:KO_estimate_change_variable} and Eq.\eqref{eqn:KO_estimate_change_variable_2}, we have that
	\[
			\mathbb{E}^\mathbb{Q}\biggl\vert \int_0^Tde^{-c(T-s)}X_s\,ds\biggr\vert^{2}
		\]
	is uniformly bounded in $T$ on $[0,\infty).$ This gives the desired result. 	
\end{proof}

\subsectionfont{\scriptsize}
\subsection{Sensitivity with respect to \texorpdfstring{$\sigma$}{s}}
 
In this section we calculate the long-term sensitivity with respect to $\sigma.$

\begin{prop} 
	Under the Heston model, the long-term sensitivity with respect to the parameter $\sigma$ is
	\[
			\lim_{T\rightarrow\infty}\frac{1}{T}\frac{\partial }{\partial \sigma}\ln v(\chi,T)=-\frac{\partial\lambda}{\partial \sigma}.
		\]
\end{prop}

\begin{proof}
	We analyze the expectation term $\mathbb{E}^{\mathbb{Q}}[\frac{1}{\phi({X}_T)}\, e^{\int_0^Tf(X_s,s;T)\,ds}]$ appearing in the decomposition 
	\[
			v(\chi,T)=e^{-\lambda T}\phi(\chi)\mathbb{E}^{\mathbb{Q}}\Bigl[\frac{1}{\phi({X}_T)}\, e^{\int_0^Tf(X_s,s;T)\,ds}\Bigr]
		\]
	by using the method explained in Section \ref{sec:vega}. To apply Theorem \ref{thm:vega_vari}, consider the Lamperti transform
	\[
			\ell(x):=\int_{0}^x \frac{1}{\sigma\sqrt{y}}\,dy=\frac{2}{\sigma}\sqrt{x}.
		\]
	The process $\check{X}$ defined by $\check{X}_t:=\ell(X_t)=\frac{2}{\sigma}\sqrt{X_t},t\ge0,$ satisfies
	\[
			d\check{X}_t = \gamma(\check{X}_t,t;T)\,dt+dB_t, \qquad \check{X}_0=\frac{2}{\sigma}\sqrt{\chi},
		\]
	where the drift function is
	\[
			\gamma(\check{x},t;T) :=-\frac{1}{2}\Bigl(k+\frac{q\mu\rho\sigma}{\varsigma}+\sigma^2B+\frac{q(1-\rho^2)\sigma^2}{1-q}\beta(T-t)\Bigr)\check{x}+\Bigl(\frac{2k\overline{m}}{\sigma^2}-\frac{1}{2}\Bigr) \frac{1}{\check{x}}.
		\]
	Define
	\begin{align*}
			\Phi(\check{x};\sigma) & := e^{-\frac{1}{4}\sigma^2B\check{x}^2},\\
			F(\check{x},t;T;\sigma) & := -\frac{q(1-\rho^2)\sigma^4\check{x}^2}{8(1-q)}\bigl(B-\beta(T-t)\bigr)^2,
		\end{align*}
	and
	\[
			\tilde{w}(\check{x},T;\sigma) :=\mathbb{E}^{\mathbb{Q}}\Bigl[\frac{1}{\Phi(\check{X}_T;\sigma)}\, e^{\int_0^TF(\check{X}_s,s;T;\sigma)\,ds} \, \Big\vert \, \check{X}_0=\check{x}\Bigr]
		\]
	so that 
	\[
			\mathbb{E}^{\mathbb{Q}}\Bigl[\frac{1}{\phi({X}_T)}\, e^{\int_0^Tf(X_s,s;T)\,ds} \, \Big\vert \, X_0=\chi\Bigr] = \tilde{w}\Bigl(\frac{2}{\sigma}\sqrt{\chi},T;\sigma\Bigr).
		\]

	We want to analyze the large time behavior of $\frac{\partial}{\partial\sigma}\tilde{w}(\frac{2}{\sigma}\sqrt{\chi},T;\sigma).$ The perturbation parameter $\sigma$ appears only  in the functionals $\Phi$ and  $F$ as well as in the drift term and the initial value of $\check{X}$, but not in the volatility term of $\check{X}.$ It is easy to check that the map $(\check{x},\sigma)\mapsto \tilde{w}(\check{x},T;\sigma)$ is continuously differentiable by using Eq.\eqref{eqn:alpha_beta_Hes}. Using the chain rule, we have
	\[
			\frac{\partial}{\partial\sigma}\tilde{w}\Bigl(\frac{2}{\sigma}\sqrt{\chi},T;\sigma\Bigr) = -\frac{2\sqrt{\chi}}{\sigma^2} \bigl(\frac{\partial\tilde{w}}{\partial\check{x}}\bigr) \Bigl(\frac{2}{\sigma}\sqrt{\chi},T;\sigma\Bigr) +\bigl(\frac{\partial\tilde{w}}{\partial\sigma}\bigr)\Bigl(\frac{2}{\sigma}\sqrt{\chi},T;\sigma\Bigr).
		\]
	In the first derivative on the right hand, observe that
	\[
			\lim_{T\to\infty}\frac{1}{T} \bigl(\frac{\partial\tilde{w}}{\partial\check{x}}\bigr)\Bigl(\frac{2}{\sigma}\sqrt{\chi},T;\sigma\Bigr)=0
		\]
	because we already proved that the derivative in Eq.\eqref{eqn:Hes_delta_go_zero} is convergent as $T\to\infty.$ For the second derivative, one can show that
	\[
			\lim_{T\to\infty}\frac{1}{T}\bigl(\frac{\partial\tilde{w}}{\partial\sigma}\bigr)\Bigl(\frac{2}{\sigma}\sqrt{\chi},T;\sigma\Bigr) = 0
		\]
	by the same method as in Proposition \ref{prop:Hes_k} because the $\sigma$-perturbation happens only in the drift term in $\check{X}.$ This gives the desired result.
\end{proof}

\end{document}